%% file: Main.tex
\documentclass[reqno,letterpaper]{amsart}
\usepackage[margin=1.5in,bottom=1.5in]{geometry}		

\usepackage{amsmath}		
\usepackage{amssymb}		
\usepackage{amsfonts}		
\usepackage{amsthm}		
\usepackage[foot]{amsaddr}		

\usepackage{mathtools}		

\mathtoolsset{%
}

\usepackage[utf8]{inputenc}	
\usepackage[T1]{fontenc}	

\usepackage[
cal=cm,
]
{mathalfa}

\usepackage{dsfont}	

\usepackage[sf,mono=false]{libertine}

\usepackage{acronym}	
\newcommand{\acli}[1]{\emph{\acl{#1}}}	
\newcommand{\acdef}[1]{\define{\acl{#1}} \textup{(\acs{#1})}\acused{#1}}	
\newcommand{\acdefp}[1]{\define{\aclp{#1}} \textup{(\acsp{#1})}\acused{#1}}	

\usepackage[labelfont={bf,small},labelsep=colon,font=small]{caption}	

\usepackage{subcaption}	

\usepackage[svgnames]{xcolor}	
\colorlet{MyRed}{Crimson!60!DarkRed}
\colorlet{MyBlue}{DodgerBlue!75!black}
\colorlet{MyGreen}{DarkGreen}
\colorlet{MyViolet}{DarkMagenta}

\colorlet{MyLightBlue}{DodgerBlue!20}
\colorlet{MyLightGreen}{MyGreen!20}

\colorlet{PrimalColor}{MyBlue}
\colorlet{PrimalFill}{MyLightBlue}
\colorlet{DualColor}{MyRed}

\colorlet{AlertColor}{MyRed}	
\colorlet{BadColor}{MyRed}	
\colorlet{GoodColor}{MyGreen}	
\colorlet{LinkColor}{MediumBlue}	
\colorlet{MacroColor}{MyViolet}
\colorlet{RevColor}{MediumBlue}	

\usepackage{latexsym}	
\usepackage{fontawesome}	
\usepackage{pifont}	

\usepackage{tikz}	
\usetikzlibrary{3d}

\usepackage{array}	
\usepackage{booktabs}	
\usepackage[inline,shortlabels]{enumitem}	
\setlist[1]{topsep=\smallskipamount,itemsep=\smallskipamount,left=\parindent}
\setlist[2]{left=0pt}

\usepackage[kerning=true]{microtype}	

\usepackage{lipsum}	
\usepackage{tabto}	
\usepackage{xspace}	
\usepackage{pgffor}

\usepackage[authoryear,sort&compress]{natbib}		

\bibpunct[, ]{[}{]}{,}{}{,}{,}
\setcitestyle{numbers,square}

\usepackage{hyperref}
\hypersetup{
final,
colorlinks=true,
linktocpage=true,
pdfstartview=FitH,
breaklinks=true,
pdfpagemode=UseNone,
pageanchor=true,
pdfpagemode=UseOutlines,
plainpages=false,
bookmarksnumbered,
bookmarksopen=false,
bookmarksopenlevel=1,
hypertexnames=true,
pdfhighlight=/O,
urlcolor=LinkColor,linkcolor=LinkColor,citecolor=LinkColor,	
pdftitle={},
pdfauthor={},
pdfsubject={},
pdfkeywords={},
pdfcreator={pdfLaTeX},
pdfproducer={LaTeX with hyperref}
}

\newcommand{\EMAIL}[1]{\email{\href{mailto:#1}{#1}}}

\usepackage[sort&compress,capitalize,nameinlink]{cleveref}	
\crefname{equation}{Eq.}{Eqs.}
\crefname{algo}{Algorithm}{Algorithms}
\crefname{assumption}{Assumption}{Assumptions}
\crefname{case}{Case}{Cases}
\crefname{figure}{Fig.}{Figs.}
\crefrangeformat{equation}{\upshape(#3#1#4)\textendash(#5#2#6)}

\usepackage{algorithm}	
\usepackage{algpseudocode}	

\usepackage{thmtools}	
\usepackage{thm-restate}	

\theoremstyle{plain}
\newtheorem{lemma}{Lemma}	
\newtheorem{proposition}{Proposition}	

\newtheorem*{theorem*}{Theorem}	
\newtheorem*{corollary*}{Corollary}	

\theoremstyle{definition}
\newtheorem{definition}{Definition}	
\newtheorem{example}{Example}	

\newtheorem*{definition*}{Definition}	
\newtheorem*{assumption*}{Assumptions}	
\newtheorem*{example*}{Example}	

\theoremstyle{remark}
\newtheorem{remark}{Remark}	

\newtheorem*{remark*}{Remark}	
\newtheorem*{notation*}{Notation}	

\def\endenv{\hfill{\large$\diamond$}}	

\numberwithin{remark}{section}	
\numberwithin{example}{section}	

\usepackage[showdeletions]{color-edits}	
\newcommand{\draft}[1]{#1}	

\newcommand{\define}[1]{\emph{\draft{#1}}}	

\newcommand{\newmacro}[2]{\newcommand{#1}{\draft{#2}}}	
\newcommand{\newop}[2]{\DeclareMathOperator{#1}{\draft{#2}}}	
\newcommand{\newoplims}[2]{\DeclareMathOperator*{#1}{\draft{#2}}}	

\newcommand{\pd}{\partial}	

\DeclarePairedDelimiter{\braces}{\{}{\}}	
\DeclarePairedDelimiter{\bracks}{[}{]}	
\DeclarePairedDelimiter{\parens}{(}{)}	
\DeclarePairedDelimiter{\of}{(}{)}	

\DeclarePairedDelimiter{\abs}{\lvert}{\rvert}	

\DeclarePairedDelimiter{\setof}{\{}{\}}	
\DeclarePairedDelimiterX{\setdef}[2]{\{}{\}}{#1:#2}	
\DeclarePairedDelimiterXPP{\exclude}[1]{\mathopen{}\setminus}{\{}{\}}{}{#1}

\DeclarePairedDelimiterX{\braket}[2]{\langle}{\rangle}{#1,#2}	
\DeclarePairedDelimiterX{\inner}[2]{\langle}{\rangle}{#1,#2}	
\newcommand{\dualp}[2]{#1^{\top} \cdot #2}	
\DeclarePairedDelimiter{\norm}{\lVert}{\rVert}	
\DeclarePairedDelimiterXPP{\dnorm}[1]{}{\lVert}{\rVert}{_{\ast}}{#1}	
\DeclarePairedDelimiterXPP{\onenorm}[1]{}{\lVert}{\rVert}{_{1}}{#1}	
\DeclarePairedDelimiterXPP{\twonorm}[1]{}{\lVert}{\rVert}{_{2}}{#1}	
\DeclarePairedDelimiterXPP{\supnorm}[1]{}{\lVert}{\rVert}{_{\infty}}{#1}	

\newcommand{\quotes}[1]{``#1''}	

\newcommand{\defeq}{\coloneqq}	

\newcommand{\from}{\colon}	

\newcommand{\map}[5]{
\begin{aligned}
#1 \from \, &#2 \to #3 \\
       & #4 \longmapsto #5
\end{aligned}
}

\newmacro{\nat}{i}	
\newmacro{\natalt}{j}	
\newmacro{\nats}{\mathbb{N}}	
\newmacro{\N}{\nats}	

\newmacro{\integer}{a}	
\newmacro{\intalt}{b}	
\newmacro{\integers}{\mathbb{Z}}	
\newmacro{\Z}{\integers}	

\newmacro{\rational}{r}	
\newmacro{\ratalt}{s}	
\newmacro{\rationals}{\mathbb{Q}}	
\newmacro{\Q}{\rationals}	

\newmacro{\real}{x}	
\newmacro{\realalt}{y}	
\newmacro{\reals}{\mathbb{R}}	
\newmacro{\R}{\reals}	

\newmacro{\complex}{z}	
\newmacro{\complexalt}{w}	
\newmacro{\complexes}{\mathbb{C}}	

\newoplims{\argmax}{arg\,max}	
\newoplims{\argmin}{arg\,min}	
\newoplims{\intersect}{\bigcap}	
\newoplims{\union}{\bigcup}	

\newop{\aff}{aff}	
\newop{\bd}{bd}	
\newop{\bigoh}{\mathcal{O}}	
\newop{\card}{card}	
\newop{\cl}{cl}	
\newop{\conv}{conv}	
\newop{\crit}{crit}	
\newop{\diag}{diag}	
\newop{\diam}{diam}	
\newop{\dist}{dist}	
\newop{\dom}{dom}	
\newop{\eig}{eig}	
\newop{\ess}{ess}	
\newop{\grad}{grad}	
\newop{\Hess}{Hess}	
\newop{\ind}{ind}	
\newop{\im}{im}	
\newop{\intr}{int}	
\newop{\Jac}{Jac}	
\newop{\one}{\mathds{1}}	
\newop{\proj}{pr}	
\newop{\prox}{prox}	
\newop{\rank}{rank}	
\newop{\relint}{ri}	
\newop{\sign}{sgn}	
\newop{\supp}{supp}	
\newop{\Sym}{Sym}	
\newop{\tr}{tr}	
\newop{\unif}{unif}	
\newop{\vol}{vol}	

\newcommand{\cf}{cf.\xspace}	
\newcommand{\eg}{e.g.,\xspace}	
\newcommand{\ie}{i.e.,\xspace}	
\newcommand{\vs}{vs.\xspace}	
\newcommand{\viz}{viz.\xspace}	

\newcommand{\txs}{\textstyle}	

\newcommand{\alt}[1]{#1'}	
\newcommand{\altalt}[1]{#1''}	

\newmacro{\ball}{\mathbb{B}}	
\newmacro{\sphere}{\mathbb{S}}	

\newmacro{\argdot}{\kern1pt\boldsymbol{\cdot}\kern1pt}	
\newmacro{\dd}{\:d}	
\newcommand{\ddt}[1]{\frac{d#1}{dt}}	
\newmacro{\dir}{\kern0pt\pd\mkern-1mu{}}	

\newcommand{\insum}{\sum\nolimits}	

\newmacro{\const}{c}	
\newmacro{\Const}{C}	

\newmacro{\param}{\theta}	
\newmacro{\params}{\Theta}	

\newmacro{\coef}{\lambda}	
\newmacro{\coefalt}{\mu}	
\newmacro{\coefaltalt}{\nu}	

\newmacro{\pexp}{p}	
\newmacro{\qexp}{q}	
\newmacro{\rexp}{r}	

\newmacro{\aposs}{their} 

\newmacro{\vecspace}{\R^{\vdim}}	
\newmacro{\dirspace}{Z}	
\newmacro{\subspace}{\mathcal{W}}	

\newmacro{\coord}{a}	
\newmacro{\coordalt}{b}	
\newmacro{\coordaltalt}{c}	
\newmacro{\nCoords}{n}	
\newmacro{\dims}{\nCoords}	
\newmacro{\vdim}{\nCoords}	

\newmacro{\bvec}{e}	
\newmacro{\uvec}{u}	
\newmacro{\bvecs}{\mathcal{E}}	

\newmacro{\point}{x}	
\newmacro{\pointalt}{\alt\point}	
\newmacro{\pointaltalt}{\altalt\point}	
\newmacro{\points}{\mathcal{X}}	
\newmacro{\intpoints}{\relint\points}	

\newmacro{\base}{p}	
\newmacro{\basealt}{q}	
\newmacro{\basealtalt}{u}	

\newmacro{\set}{\mathcal{S}}	
\newmacro{\setalt}{\alt\set}	
\newmacro{\setaltalt}{\altalt\set}	

\newmacro{\fn}{f}		

\newop{\Span}{Span} 

\newop{\diver}{div}	
\newop{\curl}{curl}	
\let\d\relax   
\newmacro{\d}{d} 
\newop{\cod}{\delta}	

\newmacro{\closed}{\mathcal{C}}	
\newmacro{\cpt}{\mathcal{K}}	
\newmacro{\nhd}{\mathcal{U}}	
\newmacro{\open}{\mathcal{U}}	

\newmacro{\mfld}{\mathcal{M}}	
\newmacro{\surf}{\Sigma}	
\newmacro{\gmat}{g}	
\newmacro{\g}{\gmat}	
\newmacro{\gdist}{\dist_{\gmat}}	

\newmacro{\tvec}{z}	
\newmacro{\nvec}{n}	
\newmacro{\form}{\omega}	

\newmacro{\vfield}{X}	
\newmacro{\vfieldalt}{Y}	

\newmacro{\sub}{S} 
\newmacro{\cvx}{\mathcal{C}}	
\newmacro{\subd}{\partial}	

\newop{\tspace}{T}	
\newop{\tcone}{TC}	
\newop{\dcone}{\tcone^{\ast}}	
\newop{\ncone}{NC}	
\newop{\pcone}{PC}	
\newop{\hull}{\Delta}	

\newop{\Opt}{Opt}	
\newop{\Sol}{Sol}	
\newop{\gap}{Gap}	
\newop{\orcl}{Or}	

\newmacro{\obj}{f}	
\newmacro{\objalt}{g}	
\newmacro{\sobj}{F}	

\newcommand{\sol}[1][\point]{#1^{\ast}}	

\newmacro{\gvec}{g}	
\newmacro{\gbound}{G}	

\newmacro{\vecfield}{F}	
\newmacro{\vbound}{V}	

\newmacro{\lips}{L}	
\newmacro{\strong}{\alpha}	
\newmacro{\smooth}{\beta}	

\newop{\ex}{\mathbb{E}}	
\newop{\prob}{\mathbb{P}}	
\newop{\Var}{\mathbb{V}}	
\newop{\simplex}{\hull}	
\newmacro{\1}{\mathbf{1}} 

\DeclarePairedDelimiterXPP{\exof}[1]{\ex}{[}{]}{}{
 #1}

\DeclarePairedDelimiterXPP{\exwrt}[2]{\ex_{#1}}{[}{]}{}{
 #2}

\DeclarePairedDelimiterXPP{\probof}[1]{\prob}{(}{)}{}{
 #1}

\DeclarePairedDelimiterXPP{\oneof}[1]{\one}{\{}{\}}{}{#1}	

\newmacro{\event}{E}       
\newmacro{\eventalt}{H}       

\newmacro{\seed}{\theta}	
\newmacro{\seeds}{\Theta}	
\newmacro{\pdist}{P}	
\newmacro{\history}{\mathcal{H}}	

\newmacro{\sample}{\omega}	
\newmacro{\samples}{\Omega}	

\newmacro{\filter}{\mathcal{F}}	
\newmacro{\probspace}{(\samples,\filter,\prob)}	

\newmacro{\mean}{\mu}	
\newmacro{\sdev}{\sigma}	
\newmacro{\variance}{\sdev^{2}}	

\newmacro{\beforestart}{0}	
\newmacro{\start}{1}	
\newmacro{\afterstart}{2}	
\newmacro{\running}{\start,\afterstart,\dotsc}	

\newmacro{\run}{n}	
\newmacro{\runalt}{k}	
\newmacro{\runaltalt}{\ell}	
\newmacro{\nRuns}{T}	
\newmacro{\runs}{\mathcal{\nRuns}}	

\newmacro{\tstart}{0}	
\renewcommand{\time}{\draft{t}}	
\newmacro{\timealt}{s}	
\newmacro{\timealtalt}{\tau}	
\newmacro{\horizon}{T}	

\newmacro{\seq}{a}	
\newmacro{\seqalt}{b}	
\newmacro{\seqaltalt}{c}	

\newmacro{\state}{x}	
\newmacro{\statealt}{y}	
\newmacro{\statealtalt}{z}	

\newcommand{\curr}[1][\state]{\draft{#1}_{\run}}	

\newmacro{\mat}{M}	
\newmacro{\hmat}{H}	

\newmacro{\ones}{\mathbf{1}}	
\newmacro{\eye}{I}	
\newmacro{\zer}{\mathbf{0}}	

\newop{\Nash}{NE}	
\newop{\CE}{CE}	
\newop{\CCE}{CCE}	
\newop{\NI}{NI}	

\newop{\brep}{br}	
\newop{\reg}{Reg}	
\newop{\preg}{\overline{Reg}}	
\newop{\val}{val}	

\newmacro{\play}{i}	
\newmacro{\playalt}{j}	
\newmacro{\playaltalt}{k}	
\newmacro{\nPlayers}{N}	
\newmacro{\players}{\mathcal{\nPlayers}}	

\newmacro{\pure}{\alpha}	
\newmacro{\purealt}{\beta}	
\newmacro{\purealtalt}{\gamma}	
\newmacro{\nPures}{A}	
\newmacro{\pures}{\mathcal{\nPures}}	

\newmacro{\strat}{x}	
\newmacro{\stratalt}{\alt\strat}	
\newmacro{\strataltalt}{\altalt\strat}	
\newmacro{\strats}{\mathcal{X}}	
\newmacro{\intstrats}{\strats^{\circ}}
\newmacro{\dimStrats}{\nEffs}	
\DeclarePairedDelimiterXPP{\stratof}[1]{\strat}{(}{)}{}{#1}	

\newcommand{\eq}{\sol[\strat]}	

\newmacro{\pay}{u}	
 \newmacro{\payfield}{v}	
\newmacro{\loss}{\ell}	

\newmacro{\game}{\mathcal{G}}	
\newmacro{\gamefull}{\game(\players,\points,\pay)}	

\newmacro{\fingame}{\Gamma}	
\newmacro{\fingamefull}{\Gamma(\players,\pures,\pay)}	
\newmacro{\mixgame}{\Delta(\fingame)}	

\newmacro{\minmax}{L}	
\newmacro{\minvar}{\point_{1}}	
\newmacro{\minvaralt}{\alt\minvar}	
\newmacro{\minvars}{\points_{1}}	
\newmacro{\maxvar}{\point_{2}}	
\newmacro{\maxvaralt}{\alt\maxvar}	
\newmacro{\maxvars}{\points_{2}}	

\newmacro{\pot}{\phi}	
\newmacro{\potgame}{\fingame_{\!\textup{pot}}}	
\newmacro{\potfield}{P}	
\newmacro{\potrepfield}{\repfield_{\textup{pot}}}	

\newmacro{\harmgame}{\fingame_{\!\textup{harm}}}	
\newmacro{\harmfield}{H}	

\newmacro{\gauge}{\const}	
\newmacro{\nsgame}{\mathrm{NS}}	
\newmacro{\nsfield}{\Const}	

\newmacro{\incgame}{\fingame_{\!\textup{inc}}}	
\newmacro{\incfield}{B}	

\newmacro{\pureA}{\mathtt{A}}	
\newmacro{\pureB}{\mathtt{B}}	

\newmacro{\hreg}{h}	
\newmacro{\breg}{D}	
\newmacro{\mprox}{P}	
\newmacro{\mirror}{Q}	
\newmacro{\fench}{F}	
\newmacro{\hstr}{K}	
\newmacro{\hrange}{H}	
\newmacro{\proxdom}{\points^{\hreg}}	

\DeclarePairedDelimiterXPP{\bregof}[2]{\breg}{(}{)}{}{#1,#2}	
\DeclarePairedDelimiterXPP{\proxof}[2]{\mprox_{#1}}{(}{)}{}{#2}	

\newmacro{\dpoint}{y}	
\newmacro{\dpointalt}{\alt\dpoint}	
\newmacro{\dpointaltalt}{\altalt\dpoint}	
\newmacro{\dpoints}{\mathcal{Y}}	

\newmacro{\dstate}{Y}	
\newmacro{\dvec}{w}	

\newmacro{\zone}{\mathbb{D}}	

\newop{\Eucl}{\Pi}	
\newop{\logit}{LC}	
\newop{\dkl}{KL}	

\newmacro{\flowmap}{\Theta}	
\DeclarePairedDelimiterXPP{\flowof}[2]{\flowmap_{#1}}{(}{)}{}{#2}	

\newmacro{\traj}{x}	
\DeclarePairedDelimiterXPP{\trajof}[1]{x}{(}{)}{}{#1}	

\newmacro{\trajalt}{y}	
\newmacro{\trajaltalt}{z}	

\newcommand{\est}[1]{\hat #1}	

\newmacro{\signal}{\hat\vecfield}	
\newmacro{\step}{\gamma}	
\newmacro{\learn}{\eta}	

\newmacro{\efftime}{\tau}	

\newmacro{\error}{Z}	

\newmacro{\noise}{U}	
\newmacro{\snoise}{\xi}	
\newmacro{\noisepar}{\sdev}	
\newmacro{\noisevar}{\variance}	
\newmacro{\aggnoise}{\mathrm{\uppercase\expandafter{\romannumeral1}}}	
\newmacro{\supnoise}{\aggnoise_{\infty}}	
\newmacro{\maxnoise}{\aggnoise^{\ast}}	

\newmacro{\bias}{b}	
\newmacro{\bbound}{B}	
\newmacro{\sbias}{\chi}	
\newmacro{\aggbias}{\mathrm{\uppercase\expandafter{\romannumeral2}}}	
\newmacro{\supbias}{\aggbias_{\infty}}	
\newmacro{\maxbias}{\aggbias^{\ast}}	

\newmacro{\second}{\psi}	
\newmacro{\sbound}{M}	
\newmacro{\aggsecond}{\mathrm{\uppercase\expandafter{\romannumeral3}}}	
\newmacro{\supsecond}{\aggsecond_{\infty}}	
\newmacro{\maxsecond}{\aggsecond^{\ast}}	

\newmacro{\mix}{\delta}	
\newmacro{\unitvec}{w}	
\newmacro{\unitvar}{W}	
\newmacro{\perturb}{z}	

\newmacro{\purequery}{\est\pure}	
\newmacro{\query}{\est\state}	
\newmacro{\pivot}{\point}	
\newmacro{\querypoint}{\est\point}	

\newmacro{\vertex}{v}	
\newmacro{\vertexalt}{w}	
\newmacro{\vertexaltalt}{u}	
\newmacro{\nVertices}{V}	
\newmacro{\vertices}{\mathcal{V}}	

\newmacro{\edge}{e}	
\newmacro{\edgealt}{\al\edge}	
\newmacro{\edgealtalt}{\altalt\edge}	
\newmacro{\nEdges}{E}	
\newmacro{\edges}{\mathcal{\nEdges}}	

\newmacro{\graph}{\mathcal{G}}	
\newmacro{\graphFull}{\graph(\vertices,\edges)}	

\addauthor[Davide - action]{DL}{Red}

\addauthor[Davide - revision]{DLrev}{Green}

\addauthor[Bary]{BP}{DarkRed}

\addauthor[Pan]{PM}{MediumBlue}

\newcommand{\negspace}{\!\!\!}	
\newcommand{\shah}{Shahshahani\xspace}

\newmacro{\repfield}{\payfield^{\sharp}}	
\newmacro{\RD}{\repfield} 

\newmacro{\pspace}{\mathcal{V}}	
\newmacro{\dspace}{\vecspace^{\ast}}	

\newmacro{\score}{y}	
\newmacro{\scorealt}{\alt\strat}	
\newmacro{\scorealtalt}{\altalt\strat}	
\newmacro{\scores}{\mathcal{Y}}	
\DeclarePairedDelimiterXPP{\scoreof}[1]{\score}{(}{)}{}{#1}	

\newcommand{\clorthant}[1][\vecspace]{\draft{#1}_{+}}	
\newcommand{\orthant}[1][\vecspace]{\draft{#1}_{++}}	
\newmacro{\intsimplex}{\simplex^{\!\circ}}	
\newmacro{\tanplane}{\mathcal{Z}}	

\newmacro{\orthantplay}{\orthant[\R]^{\pures_{\play}}} 
\newmacro{\clorthantplay}{\clorthant[\R]^{\pures_{\play}}} 

\newmacro{\clorthantplayalt}{\clorthant[\R]^{\nEffs_{\play}+1}} 
\newmacro{\orthantplayalt}{\orthant[\R]^{\nEffs_{\play}+1}} 

\newmacro{\clorthantplayeff}{\clorthant[\R]^{\nEffs_{\play}}} 
\newmacro{\orthantplayeff}{\orthant[\R]^{\nEffs_{\play}}} 

\newmacro{\gfield}{G}	

\newmacro{\push}{\pi_{0}}	
\newmacro{\pull}{\iota_{0}}	

\newmacro{\chart}{\push} 
\newmacro{\incl}{\pull} 
\newmacro{\jac}{J} 

\newmacro{\corcube}{\mathcal{C}}	
\newmacro{\intcorcube}{\corcube^{\circ}}	

\newmacro{\effstrat}{\tilde\strat}	
 \newmacro{\effpayfield}{\tilde\payfield}	
\newmacro{\effvecfield}{\tilde\vecfield}	
\newmacro{\effbvec}{\tilde\bvec}	
\newmacro{\effgmat}{\tilde\gmat}	
\newmacro{\effPures}{\tilde\pures}	
\newmacro{\effpay}{\tilde\pay}

\newmacro{\eff}{\mu}	
\newmacro{\effalt}{\nu}	
\newmacro{\effaltalt}{\rho}	
\newmacro{\nEffs}{m}	
\newcommand{\Eff}[1]{\tilde{#1}} 

\newmacro{\increpfield}{\repfield_{\textup{inc}}}	
\newmacro{\borel}{\open}	
\newmacro{\effborel}{\tilde\borel}	
\newmacro{\energy}{E}	

\input{Macros/DavideMacros}

\begin{document}


\title
[A Geometric Decomposition of Games]
{A Geometric Decomposition of Finite Games:\\
Convergence Vs. Recurrence under Exponential Weights}

\author
[D.~Legacci]
{Davide Legacci$^{\ast}$}
\address{$^{\ast}$\,%
Univ. Grenoble Alpes, CNRS, Inria, Grenoble INP, LIG, 38000 Grenoble, France.}
\EMAIL{davide.legacci@univ-grenoble-alpes.fr}
\author
[P.~Mertikopoulos]
{Panayotis Mertikopoulos$^{\ast}$}
\EMAIL{panayotis.mertikopoulos@imag.fr}
\author
[B.~Pradelski]
{Bary Pradelski$^{\sharp}$}
\address{$^{\sharp}$\,%
CNRS, Maison Française d'Oxford, 2\textendash10 Norham Road, Oxford, OX2 6SE, United Kingdom.}
\EMAIL{bary.pradelski@cnrs.fr}



\thanks{The authors are grateful to Victor Boone and Marco Scarsini for fruitful discussions.}

\newacro{LHS}{left-hand side}
\newacro{RHS}{right-hand side}
\newacro{iid}[i.i.d.]{independent and identically distributed}
\newacro{lsc}[l.s.c.]{lower semi-continuous}
\newacro{NE}{Nash equilibrium}
\newacroplural{NE}[NE]{Nash equilibria}
\newacro{CE}{correlated equilibrium}
\newacroplural{CE}[CE]{correlated equilibria}
\newacro{CCE}{coarse correlated equilibrium}
\newacroplural{CCE}[CCE]{coarse correlated equilibria}
\newacro{VI}{variational inequality}
\newacroplural{VI}{variational inequalities}

\newacro{ZSG}{zero-sum game}
\newacroplural{ZSG}[ZSGs]{zero-sum games}

\newacro{GAN}{generative adversarial network}

\newacro{PG}{potential game}
\newacro{HG}{harmonic game}

\newacro{EW}{exponential\,/\,multiplicative weights}
\newacro{RD}{replicator dynamics}
\newacro{FTRL}{follow-the-regularized-leader}

\begin{abstract}
\input{Abstract}
\end{abstract}
\maketitle

\allowdisplaybreaks	
\acresetall	
\acused{LHS}
\acused{RHS}

\section{Introduction}
\label{sec:introduction}
\input{Body/Introduction}

\section{Preliminaries}
\label{sec:prelims}
\input{Body/Prelims}

\section{Learning via exponential weights}
\label{sec:dynamics}
\input{Body/Dynamics}

\section{The geometry of exponential weights}
\label{sec:decomposition}
\input{Body/Decomposition}

\section{Analysis and results}
\label{sec:results}
\input{Body/Results}

\section{Concluding remarks}
\label{sec:discussion}
\input{Body/Discussion}

\section*{Acknowledgments}
\begingroup
\small
\input{Acks}
\endgroup

\appendix
\setcounter{remark}{0}
\numberwithin{equation}{section}	
\numberwithin{lemma}{section}	
\numberwithin{proposition}{section}	
\numberwithin{theorem}{section}	
\numberwithin{corollary}{section}	

\section{Basic facts and definitions from Riemannian geometry}
\label{app:geometry}
\input{Appendix/App-Geometry}

\section{Effective representation of games}
\label{app:reduction}
\input{Appendix/App-Reduction}

\section{Replicator dynamics as an individual \shah gradient system}
\label{app:replicator-geometry}
\input{Appendix/App-RD}

\section{Proofs on incompressible games}
\label{app:incompressible}
\input{Appendix/App-Incompressible}

\section{Additional related work}
\label{app:additional}
\input{Appendix/App-Additional}

\section{A geometric tour}
\label{app:geometric-tour}
\input{Appendix/App-GeometricTour}

\bibliographystyle{icml}
\bibliography{bibtex/IEEEabrv,bibtex/Bibliography-PM,bibtex/Bibliography-DL}

\end{document}

%% file: Macros/DavideMacros.tex
\newcommand{\eqcomma}{\, ,} 
\newcommand{\eqdot}{\, .} 

\newmacro{\tvecalt}{\alt\tvec}
\newmacro{\vfields}{\mathfrak{X}}
\newmacro{\funcs}{\mathfrak{F}}
\newmacro{\flow}{\theta} 
\newmacro{\flowt}{\flow_{\time}} 
\newmacro{\flowp}{\flow^{\point}} 
\newmacro{\curve}{\gamma}
\newmacro{\ipoint}{\point_{0}} 
\newmacro{\iopen}{\open_{0}} 
\newmacro{\itime}{\time_{0}} 
\newmacro{\chartalt}{\pi}

\newmacro{\others}{-\play}
\newmacro{\fingamefullalt}{\alt\fingame(\players,\pures,\alt\pay)}
\DeclareMathOperator{\KL}{KL}

\newmacro{\payns}{k} 
\newmacro{\fingamefullns}{\fingame\parens{\players, \pures, \payns }} 

\newmacro{\effstrats}{\intcorcube} 

\newmacro{\pureplay}{\pure_{\play}} 
\newmacro{\purealtplay}{\purealt_{\play}} 

\newmacro{\effplay}{\eff_{\play}} 
\newmacro{\effaltplay}{\effalt_{\play}} 

\newmacro{\findex}{\play\pure_{\play}} 
\newmacro{\findexalt}{\play\purealt_{\play}} 

\newmacro{\findexothers}{\others\pure_{\others}} 

\newmacro{\rindex}{\play\eff_{\play}} 
\newmacro{\rindexalt}{\play\effalt_{\play}} 
\newmacro{\effindex}{\rindex} 
\newmacro{\effindexalt}{\rindexalt} 

\newmacro{\zindex}{\play0_{\play}} 

\newmacro{\shindex}{\play\pure_{\play}\purealt_{\play}} 
\newmacro{\effshindex}{\play\eff_{\play}\effalt_{\play}} 

\newmacro{\rangeplay}{\setof{0, \dots, \nEffs_{\play}}} 
\newmacro{\effrangeplay}{\setof{1, \dots, \nEffs_{\play}}} 

\newmacro{\effd}{\Eff{d}}
\newmacro{\effpot}{\Eff{\pot}}
\newmacro{\effg}{\effgmat} 

\newmacro{\effrepfield}{\tilde{\payfield}^{\sharp}} 

\newmacro{\dummyB}{B} 

\newcommand{\varEx}[3]{{#1_{#2, #3}}} 
\newcommand{\stratEx}[2]{\varEx{\strat}{#1}{#2}}
\newcommand{\effstratEx}[2]{\varEx{\effstrat}{#1}{#2}}
\newmacro{\potparam}{\gamma} 

\newmacro{\flowD}{\mathcal{D}} 
\newmacro{\me}{\mu} 
\newmacro{\ms}{\Omega} 
\newmacro{\regl}{h}

\newmacro{\ftrlrefs}{
\citep{mertikopoulosRiemannianGameDynamics2018, SS11, shalev-shwartzConvexRepeatedGames2006, ABB04}%
}

\newcommand{\dex}[2]{\pd \strat_{#1 #2_{#1}}}
\newcommand{\dey}[2]{\pd \effstrat_{#1 #2_{#1}}}

\renewcommand{\ll}{\left(}
\newcommand{\rr}{\right)}

\renewcommand{\div}{\diver}
\newmacro{\detg}{D} 
\newmacro{\de}{\pd} 

\newmacro{\coexact}{incompressible} 
\newmacro{\Coexact}{Incompressible} 

\renewcommand{\b}{\mathbf{b}} 

\newcommand{\triple}[3]{ #1_{#2 #3_{#2}} } 
\newcommand{\x}[2]{\strat_{#1 #2_{#1}}} 
\newcommand{\V}[2]{\payfield_{#1 #2_{#1}}} 
\renewcommand{\v}[2]{\effpayfield_{#1 #2_{#1}}} 



%% file: Abstract.tex
%
%
In view of the complexity of the dynamics of learning in games, we seek to decompose a game into simpler components where the dynamics' long-run behavior is well understood.
A natural starting point for this is Helmholtz's theorem, which decomposes a vector field into a potential and an incompressible component.
However, the geometry of game dynamics \textendash\ and, in particular, the dynamics of \ac{EW} schemes \textendash\ is not compatible with the Euclidean underpinnings of Helmholtz's theorem.
This leads us to consider a specific Riemannian framework based on the so-called \emph{\shah metric}, and introduce the class of \define{incompressible games}, for which we establish the following results:
First, in addition to being volume-preserving, the continuous-time \ac{EW} dynamics in incompressible games admit a constant of motion and are \define{Poincaré recurrent} \textendash\ \ie almost every trajectory of play comes arbitrarily close to its starting point infinitely often.
Second, we establish a deep connection with a well-known decomposition of games into a potential and harmonic component (where the players' objectives are aligned and anti-aligned respectively):
a game is incompressible if and only if it is harmonic, implying in turn that the \ac{EW} dynamics lead to Poincaré recurrence in harmonic games.

%% file: Body/Introduction.tex

One of the driving open questions in game-theoretic learning is whether \textendash\ and under what conditions \textendash\ players eventually learn to emulate rational behavior through repeated interactions.
Put differently, whether a game-theoretic learning process converges to a rational outcome,
what type of outcome this could be,
under which mode of convergence,
in which games,
etc.
This question has long been one of the mainstays of non-cooperative game theory, and it has recently received increased attention owing to a surge of breakthrough applications in machine learning and AI, from \acp{GAN}, to multi-agent reinforcement learning and online ad auctions.

Depending on the precise context, this question may admit a wide range of answers, from positive to negative.
Starting with the positive, a folk result states that if the players of a finite game follow a no-regret learning process, the players' empirical frequency of play converges in the long run to the set of \acp{CCE} \textendash\ also known as the game's \define{Hannan set} \cite{Han57}.
This result has been pivotal for the development of the field because no-regret play can be achieved through fairly simple myopic processes like the \acf{EW} update scheme \cite{Vov90,LW94,ACBFS95,ACBFS02} and its many variants \cite{RS13-NIPS,SS11,Sor09}.
On the downside however
\begin{enumerate*}
[\upshape(\itshape a\upshape)]
\item
this convergence result does not concern the actual strategies employed by the players day-to-day;
and
\item
in many games, the notion of a \ac{CCE} can lead to outcomes that fail even the weakest axioms of rationalizability.
\end{enumerate*}
For example, as was shown by \citet{VZ13}, players may enjoy \emph{negative regret} for all time, but still play \emph{only strictly dominated strategies} for the entire horizon of play.

This takes us to the negative end of the spectrum.
If we focus on the evolution of the players' strategies, a series of well-known impossibility results by \citet{HMC03,HMC06} have established that there are no uncoupled learning dynamics \textendash\ deterministic or stochastic, in either continuous or discrete time \textendash\ that converge to \ac{NE} in \emph{all} games from any initial condition.%
\footnote{The adjective ``uncoupled'' refers here to learning processes where a player's update rule does not explicitly depend on the other players' strategies.}
In turn, this lends further weight to the question of determining in which games a learning process converges to \acl{NE} in the day-to-day sense, and in which it does not.

In this regard, the class of games with arguably the strongest convergence guarantees is the class of potential games \citep{MS96}.
Here, the dynamics of \ac{EW} methods are known to converge, in both continuous and discrete time, and even when the players only have bandit, payoff-based information at their disposal \cite{HS98,HCM17}.
By contrast, in two-player, \acp{ZSG} with fully mixed equilibria (like Matching Pennies) 
the standard implementation of the \ac{EW} algorithm diverges, even with perfect, mixed payoff observations \cite{MLZF+19};
the so-called ``optimistic'' variant of \citet{RS13-NIPS} converges at a geometric rate if run with perfect payoff observations \cite{WLZL21} but diverges if such information is not available \cite{HIMM20,HAM21,HACM22};
and, finally, the continuous-time version of the \ac{EW} dynamics \textendash\ the \acli{RD} \textendash\  is \emph{Poincaré recurrent}, \ie the trajectory of play returns infinitely close to where it started, infinitely often \cite{PS14,MPP18}.

Going back to the two classes of games above, potential games are quite special in that the players' incentives are \emph{aligned} (their externalities are positive);
on the other hand, 
in two-player zero-sum games, the players' incentives are \emph{anti-aligned} (externalities are negative).
Largely motivated by this observation, \citet{CMOP11} introduced a principled framework of decomposing a game into a \define{potential} and a \define{harmonic} component:
the potential component of the game captures interactions that amount to a common interest game, while the harmonic component captures the conflicts between the players' interests.%
\footnote{The class of harmonic games contains \emph{some} zero-sum games (like Matching Pennies), but not all;
likewise, the class of zero-sum games contains \emph{some} harmonic games (\eg when all players have the same number of strategies), but not all.
In general, the classes of harmonic and zero-sum games are distinct, and they represent different incarnations of ``anti-aligned objectives''.}
In this way, the decomposition of \citet{CMOP11} effectively maps all games to a spectrum ranging from \emph{fully aligned} (when the harmonic component of the game is zero) to \emph{fully anti-aligned} (when the potential component is zero).

Building on this decomposition, a natural question that arises is whether a similar conclusion can be drawn for the players' \emph{learning dynamics}.
Specifically, focusing for concreteness on continuous time (which eliminates complications related to the players' hyperparameters or feedback structure), a key question is whether the space of games can be likewise mapped to a ``convergence spectrum'', with (global) convergence on one end, and global non-convergence\,/\,Poincaré recurrence on the other.
A version of this question was already treated in a series of follow-up works by
\citet{candoganLearningNearpotentialGames2011,
candoganDynamicsNearPotentialGames2011,
candoganNearPotentialGamesGeometry2013}
who showed that the best-response dynamics remain convergent in slight perturbations of potential games.
However, moving further toward the class of harmonic games hit an important obstacle, and has remained an open question since the original work of \citet{CMOP11}:
except for some special cases,
\emph{the behavior of the replicator dynamics in harmonic games is not well understood.}

\subsection*{Our contributions}

In view of the above, our paper's overarching objective is to derive a dynamics-driven decomposition of games \textendash\ and, in so doing, to shed light on the dynamics of harmonic games.
Motivated by Helmholtz's theorem for the decomposition of vector fields into a potential and an incompressible, divergence-free component, we first seek to define a class of \define{incompressible games} at the opposite end of potential games.
However, the geometry of the dynamics turns out to be incompatible with the standard Euclidean geometry of the simplex, so we are led to consider a nonlinear Riemannian structure on the simplex, the \shah metric \cite{Sha79}.
This ends up complicating the construction significantly, but it allows us to show that the class of incompressible games that we introduce has the characteristic property that the players' learning dynamics are \define{volume-preserving} (\ie a set of initial conditions does not decrease in volume relative to the \shah metric).

As a consequence of this, the class of incompressible games is shown to exhibit two fairly unexpected properties:
\begin{enumerate}
\item \emph{A game is harmonic if and only if it is incompressible,}
and the decomposition of a game into a potential and incompressible component (relative to the \shah metric) is equivalent to that of \citet{CMOP11}.
\item \emph{Incompressible games are conservative}, 
\ie the dynamics admit a constant of motion.
\end{enumerate}
Both properties are surprising, for different reasons.
The first, because harmonic and incompressible games have completely different origins:
the former is coming from the combinatorial decomposition of \citet{CMOP11}, the latter from the kernel of the \shah divergence operator, so there is no reason to expect these notions to coincide.
The second, because volume preservation and constants of motion are two complementary and independent properties, so the fact that the former implies the latter is quite mysterious.%
\footnote{Notably, \citet{FVGL+20} showed that the \ac{EW} dynamics are volume-preserving in \emph{every} game relative to a differnt volume form on the simplex;
however, only very special classes of games admit a constant of motion \textendash\ \cf the discussion following \cref{thm:constant}.}

Building further on the above, we also show that \emph{the \ac{EW} dynamics are Poincaré recurrent in harmonic games.}
By itself, this provides a partial answer to the open-ended question of whether harmonic games should be placed in the non-convergent end of the spectrum \citep{CMOP11}.
Moreover, to the best of our knowledge, Poincaré recurrence in the context of game-theoretic learning has been so far established only in zero-sum games with a fully mixed equilibrium and variations of the above \citep{booneDarwinPoincareNeumann2019,
MPP18,
PS14}.
Seeing as harmonic games are related to zero-sum games (though neither property implies or is implied by the other, \cf \cref{rem:harmonic-vs-zero-sum}), this result identifies an important new class of games where no-regret learning in continuous time fails to converge.

\subsection*{Related work}

Before the general definition of harmonic games by \citet{CMOP11}, specific instances thereof were already studied in the context of cyclic games, the battle of the sexes, buyer/seller games, and crime deterrence games \citep{hofbauerSophisticatedImitationCyclic2000,
smithBattleSexesGenetic1987,
Fri91,
cressmanEvolutionaryDynamicsCrime1998,
candoganDynamicStrategicInteractions2013}.
Building on these early works,
\citet{wangWeightedPotentialGame2017,
liNoteOrthogonalDecomposition2019,
abdouDecompositionGamesStrategic2022}
proposed a weighted versions of the decomposition by \citet{CMOP11} based on different inner products on the space of games.
\citet{chengDecomposedSubspacesFinite2016} proposed in particular a concise derivation of the decomposition of \citet{CMOP11} with applications to (network) evolutionary games and near-potential games.
\citet{hwangStrategicDecompositionsNormal2020} present a projection-based decomposition method, equivalent to that of \citet{CMOP11} for finite games and that applies also to mixed extensions of normal form games with continuous action spaces.

On the interplay between decomposition methods and dynamics, beyond the already mentioned follow-up works by \citet{candoganLearningNearpotentialGames2011,
candoganDynamicsNearPotentialGames2011,
candoganNearPotentialGamesGeometry2013} on near-potential games,
\citet{cheungChaosLearningZerosum2020} applied volume analysis techniques to the canonical decomposition of a game into zero-sum and coordination components \citep{kalaiCooperationCompetitionStrategic2010,basarInformationalPropertiesNash1974} to characterize bimatrix games where standard classes of no-regret learning exhibit Lyapunov chaos.
More recently, \citet{letcherDifferentiableGameMechanics2019} employed a decomposition argument to design a novel algorithm for finding stable fixed points in differentiable games.
The machinery we develop in this work connects the differential-geometric Hodge/Helmholtz decomposition to a constrained setting, thus providing a partial answer to an open question raised in \citet{letcherDifferentiableGameMechanics2019};
however, there is a key difference between the spirit of our approach and that of \citet{letcherDifferentiableGameMechanics2019}, that we discuss in \cref{app:letcher}.

To the best of our knowledge, the only other works in the literature that study the dynamics of harmonic games are the papers by \citet{liFiniteHarmonicGames2016} and \citet{chengDecomposedSubspacesFinite2016}, which discuss a dynamical equivalence between basis games and evolutionary harmonic games.
Except for these works, we are not aware of a similar approach in the literature.

%% file: Body/Prelims.tex

\subsection{Elements of game theory}

To fix notation, we begin by recalling some basics from game theory, roughly following \citet{FT91}.
First, a \define{finite game in normal form} consists of
a finite set of \define{players} $\play \in \players \equiv \{1,\dotsc,\nPlayers\}$,
each equipped with
\begin{enumerate*}
[(\itshape i\hspace*{1pt}\upshape)]
\item
a finite set of \define{actions} \textendash\ or \define{pure strategies} \textendash\ indexed by $\pure_{\play}\in\pures_{\play} = \{0,1,\dotsc,\nEffs_{\play}\}$ (so $\abs{\pures_{\play}} = \nEffs_{\play}+1$);
and
\item
a \define{payoff function} $\pay_{\play}\from\prod_{\playalt} \pures_{\playalt} \to \R$, which determines the player's reward $\pay_{\play}(\pure)$ at a given \define{action profile} $\pure = (\pure_{1},\dotsc,\pure_{\nPlayers})$.
\end{enumerate*}
Collectively, we will write $\pures = \prod_{\play} \pures_{\play}$ for the game's \define{action space} and $\fingame \equiv \fingamefull$ for the game with primitives as above.

During play, players may randomize their choices by playing \define{mixed strategies}, \ie probability distributions $\strat_{\play} \in \strats_{\play} \defeq \simplex(\pures_{\play})$ over $\pures_{\play}$.
In this case, we will write $\strat_{\play\pure_{\play}}$ for the probability with which player $\play\in\players$ selects $\pure_{\play}\in\pures_{\play}$ under $\strat_{\play}$,
and
we will identify $\pure_{\play} \in \pures_{\play}$ with the mixed strategy that assigns all weight to $\pure_{\play}$ (thus justifying the terminology ``pure strategies'').
Then, writing
$\strat = (\strat_{\play})_{\play\in\players}$ for the players' \define{strategy profile}
and
$\strats = \prod_{\play}\strats_{\play}$ for the game's \define{strategy space},
the players' \define{mixed payoffs} under $\strat\in\strats$ will be
\(
\pay_{\play}(\strat)
	\defeq \exwrt{\pure\sim\strat}{\pay_{\play}(\pure)}
	= \insum_{\pure\in\pures}  \pay_{\play}(\pure) \, \strat_{\pure}
\)
where, in a slight abuse of notation, we write $\strat_{\pure} \equiv \prod_{\play} \strat_{\play\pure_{\play}}$ for the joint probability of playing $\pure\in\pures$ under $\strat$.

For notational convenience, we will also write $(\strat_{\play};\strat_{-\play}) = (\strat_{1},\dotsc,\strat_{\play},\dotsc,\strat_{\nPlayers})$ for the strategy profile where player $\play$ plays $\strat_{\play}\in\strats_{\play}$ against the strategy $\strat_{-\play} \in \strats_{-\play} \defeq \prod_{\playalt\neq\play} \strats_{\playalt}$ of all other players (and likewise for pure strategies).
In this notation, each player's \define{individual payoff field} is defined as
\begin{equation}
\label{eq:payv-mixed}
\payfield_{\play}(\strat)
	= (\pay_{\play}(\pure_{\play};\strat_{-\play}))_{\pure_{\play}\in\pures_{\play}}
\end{equation}
so the mixed payoff of player $\play\in\players$ under $\strat\in\strats$ becomes
\begin{equation}
\label{eq:pay-lin}
\pay_{\play}(\strat)
	= \insum_{\pure_{\play}\in\pures_{\play}}
		\pay_{\play}(\pure_{\play};\strat_{-\play}) \, \strat_{\play\pure_{\play}}
	= \dualp{\payfield_{\play}(\strat)}{\strat_{\play}}.
\end{equation}
In view of the above, the \define{aggregate payoff field} $\payfield(\strat) = (\payfield_{1}(\strat),\dotsc,\payfield_{\nPlayers}(\strat))$ collectively captures all strategic information of the game, so we will use it interchangeably as a more compact description of the game $\fingamefull$.

The most widely used solution concept in game theory is that of a \acdef{NE}, \ie a strategy profile $\eq\in\strats$ which discourages unilateral deviations in the sense that
\begin{equation}
\label{eq:Nash}
\tag{NE}
\pay_{\play}(\eq)
	\geq \pay_{\play}(\strat_{\play};\eq_{-\play})
	\quad
	\text{for all $\strat_{\play}\in\strats_{\play}$, $\play\in\players$}.
\end{equation}
Since a game's equilibria only depend on pairwise payoff comparisons, two games $\fingame(\players,\pures,\pay)$ and $\alt\fingame(\players,\pures,\alt\pay)$ are called \define{strategically equivalent}
\textendash\ and we write $\fingame \sim \alt\fingame$ \textendash\
if, for all $\pure,\purealt\in\pures$ and all $\play\in\players$, we have
\begin{equation}
\label{eq:strat-equiv}
\alt\pay_{\play}(\purealt_{\play};\pure_{-\play}) - \alt\pay_{\play}(\pure_{\play};\pure_{-\play})
	= \pay_{\play}(\purealt_{\play};\pure_{-\play}) - \pay_{\play}(\pure_{\play};\pure_{-\play}).
\end{equation}
Clearly, strategically equivalent games yield identical payoff comparisons per player, so they share the same set of \aclp{NE}.

\subsection{A strategic decomposition of games}

One of the most important classes of normal form games is the class of \acdefp{PG}.
First introduced by \citet{MS96}, \aclp{PG} enjoy several properties of interest \textendash\ existence of equilibria in pure strategies, lack of best-response cycles, convergence of standard learning dynamics and algorithms, etc.
Formally, a finite game $\fingame$ is said to be a \acli{PG} if it admits a \define{potential function} $\pot\from\strats\to\R$ such that
\begin{equation}
\label{eq:pot}
\tag{PG}
\pay_{\play}(\purealt_{\play};\pure_{-\play}) - \pay_{\play}(\pure_{\play};\pure_{-\play})
	= \pot(\purealt_{\play};\pure_{-\play}) - \pot(\pure_{\play};\pure_{-\play})
\end{equation}
for all $\pure,\purealt\in\pures$ and all $\play\in\players$.
Equivalently, in terms of mixed payoffs, this condition can be rewritten in differential form as
\begin{equation}
\label{eq:pot-diff}
\payfield(\strat)^{\top}(\stratalt - \strat)
	= \dir\pot(\strat;\stratalt - \strat)
	\quad
	\text{for all $\strat,\stratalt \in \strats$}
\end{equation}
where
$\pot(\strat) \defeq \insum_{\pure} \pot\of\pure \,  \strat_{\pure}$ denotes the mixed extension of $\pot$ to $\strats$,
and
$\dir\pot(\strat;\tvec) = \lim_{t\to0^{+}} \bracks{\pot(\strat + t\tvec) - \pot(\strat)} / t$ denotes the (one-sided) directional derivative of $\pot$ at $\strat$ along $\tvec$.

\Aclp{PG} capture strategic interactions with ``aligned incentives'' (as in common interest and congestion games).
Dually to this, \citet{CMOP11} introduced the class of \acdefp{HG} as those with ``anti-aligned incentives'', \viz
\begin{equation}
\label{eq:harm}
\tag{HG}
\insum_{\play\in\players} \insum_{\purealt_{\play}\in\pures_{\play}}
	\bracks{\pay_{\play}(\purealt_{\play};\pure_{-\play}) - \pay_{\play}(\pure_{\play};\pure_{-\play})}
	= 0
\end{equation}
for all $\pure\in\pures$, meaning that the net incentive to deviate toward and away from any pure strategy profile is zero.
In contrast to \aclp{PG}, \aclp{HG} generically do not admit pure equilibria and they possess non-terminating best-response paths,
so they can be seen as ``orthogonal'' to potential games.

This observation was made precise by \citet{CMOP11} who showed that
any finite game admits the \define{strategic decomposition}
\begin{equation}
\label{eq:Candogan}
\fingame
	= \potgame + \harmgame
\end{equation}
where $\potgame$ is potential and $\harmgame$ is harmonic.%
\footnote{The notation $\fingame + \alt\fingame$ for two games $\fingame\equiv\fingame(\players,\pures,\pay)$ and $\alt\fingame \equiv \alt\fingame(\players,\pures,\alt\pay)$ denotes the game
with the same player/action structure as $\fingame$ and $\alt\fingame$, and payoff functions $\pay_{\play} + \alt\pay_{\play}$ for all $\play\in\players$.}
This decomposition is achieved by representing $\fingame$ as a weighted \define{preference graph}, endowing said graph with a specific, Euclidean-like structure, and using the combinatorial Helmholtz decomposition theorem \citep{jiangStatisticalRankingCombinatorial2011} to obtain \eqref{eq:Candogan}.
In general, this decomposition is only unique up to strategic equivalence:
more precisely, if $\fingame$ admits the alternative decomposition $\fingame = \alt\potgame + \alt\harmgame$ with $\alt\potgame$ potential and $\alt\harmgame$ harmonic, then $\alt\potgame$ is strategically equivalent to $\potgame$ and $\alt\harmgame$ to $\harmgame$.
We will return to this decomposition later.

%% file: Body/Dynamics.tex

Throughout our paper, we will focus on dynamic learning proceses where the players seek to myopically improve their individual payoffs over time.
A crucial requirement in this regard is the minimization of the players' \emph{regret}, that is, the difference between a player's cumulative payoff and the player's best strategy in hindsight.
Formally, assuming that play evolves in continuous time, the \define{regret} of a player $\play\in\players$ relative to a sequence of play $\stratof{\time}\in\strats$, $\time\geq0$, is defined as
\begin{equation}
\label{eq:regret}
\reg_{\play}(\horizon)
	= \max_{\base_{\play}\in\strats_{\play}}
		\int_{0}^{\horizon}
			\bracks{\pay_{\play}(\base_{\play};\traj_{-\play}(\time)) - \pay_{\play}(\stratof{\time})}
			\dd\time
\end{equation}
and we say that player $\play$ has \define{no regret} if $\reg_{\play}(\horizon) = o(\horizon)$.

The archetypal method for attaining no regret is the so-called \acdef{EW} update scheme, whereby an action is played with probability that is exponentially proportional to its cumulative payoff.
This simple stimulus-response model goes back to \citet{Vov90}, \citet{LW94} and \citet{ACBFS95}, and, in our setting, it boils down to the dynamics
\begin{equation}
\label{eq:EW}
\tag{EW}
\begin{aligned}
\score_{\play}(\time)
	= \score_{\play}\of{\tstart}
		+ \int_{0}^{\time} \negspace \payfield_{\play}\of{\stratof{\time}} \dd\time
	\qquad
\strat_{\play}(\time)
	= \logit_{\play}(\score_{\play}(\time))
\end{aligned}
\end{equation}
where $\logit_{\play}\from\R^{\pures_{\play}}\to\strats_{\play}$ denotes the \define{logit choice} map
\begin{equation}
\label{eq:logit}
\logit_{\play}(\score_{\play})
	= \frac{(\exp(\score_{\play\pure_{\play}}))_{\pure_{\play}\in\pures_{\play}}}{\sum_{\pure_{\play}\in\pures_{\play}} \exp(\score_{\play\pure_{\play}})}.
\end{equation}
As was first shown by \cite{Sor09,KM17}, the dynamics \eqref{eq:EW} enjoy a \emph{constant}, $\bigoh(1)$ regret bound, namely
\begin{equation}
\reg_{\play}(\horizon)
	\leq \log\abs{\pures_{\play}}.
\end{equation}
Owing to this remarkable regret guarantee, \eqref{eq:EW} and its variants have become the ``gold standard'' for no-regret learning;
for an introduction to the vast corpus of literature surrounding the topic, we refer the reader to \cite{AHK12,SS11,LS20}.

One last important property of \eqref{eq:EW} is that, by a standard calculation, the evolution of the players' mixed strategies $\strat_{\play} \in \strats_{\play}$ under \eqref{eq:EW} follows the \acli{RD} of \citet{TJ78}, \viz
\begin{equation}
\label{eq:RD}
\tag{RD}
\dot\strat_{\play\pure_{\play}}
	= \strat_{\play\pure_{\play}}
		\bracks{\pay_{\play}(\pure_{\play};\strat_{-\play}) - \pay_{\play}(\strat)}
\end{equation}
The \ac{RD} comprise the cornerstone of evolutionary game theory and, as such, their rationality properties have been the subject of intense study in the literature, \cf \cite{HS98,Wei95,San10} and references therein.
For all these reasons, the dynamics \eqref{eq:EW}/\eqref{eq:RD} will be our main focus in the sequel.

%% file: Body/Decomposition.tex

We now turn to our overarching objective, that is, to identify in which classes of games we can expect the dynamics of \acl{EW} to converge, and in which classes we cannot.
Our main tool for this will be \define{Helmholtz's theorem},
a simpler variant of the \define{Hodge decomposition theorem}, itself one of the most foundational results in differential geometry \cite{derhamDifferentiableManifoldsForms1984,hodgeTheoryApplicationsHarmonic1989,bhatiaHelmholtzHodgeDecompositionSurvey2013}.

To set the stage for the analysis to come, we begin by presenting the original Helmholtz decomposition of vector fields in the Euclidean setting of $\vecspace$.
Subsequently, we develop the geometric background needed to define and describe the class of \define{incompressible games} later in this section.

\subsection{The Helmholtz decomposition}

Consider the dynamics
\begin{equation}
\label{eq:dyn}
\tag{Dyn}
\dot\point
	= \vecfield(\point)
\end{equation}
induced by some sufficiently smooth vector field $\vecfield\from\vecspace\to\vecspace$ on $\vecspace$.
Helmholtz's theorem states that, if $\vecfield$ decays at infinity as $\norm{\vecfield(\point)} = o(\norm{\point}^{-2})$, it can be resolved as
\begin{equation}
\label{eq:Helmholtz}
\vecfield(\point)
	= \nabla\pot(\point)
		+ \incfield(\point)
\end{equation}
where
$\pot\from\vecspace\to\R$ is a \define{scalar potential} for $\vecfield$
and the vector field $\incfield\from\vecspace\to\vecspace$ is \define{incompressible}, \ie it has vanishing divergence:
\begin{equation}
\label{eq:div-free}
\nabla\!\cdot\!\incfield(\point)
	\defeq \insum_{\coord=1}^{\vdim} \pd\incfield_{\coord} / \pd\point_{\coord}
	= 0
	\quad
	\text{for all $\point\in\vecspace$}.
\end{equation}

The decomposition \eqref{eq:Helmholtz} is known as the \define{Helmholtz decomposition} of $\vecfield$, and it is particularly important from a dynamical standpoint because its two components exhibit ``orthogonal'' behaviors in terms of convergence.
More precisely, by standard Lyapunov arguments, the flow $\dot\point = \nabla\pot(\point)$ of the gradient component of $\vecfield$ generically converges to the critical set of $\pot$ \citep{khalilNonlinearSystems2002}.
On the other hand, Liouville's theorem shows that the flow $\dot\point = \incfield(\point)$ of the incompressible component of $\vecfield$ is volume-preserving,\footnote{Applications of Liouville's theorem \citep{Arn89} in the context of game dynamics go back at least to \citet{amannPermanenceLotkaReplicator1985,HS98} and \citet[pp.~175-227]{Wei95}.}
so it does not admit any stable attractors (asymptotically stable points or limit cycles).
In this sense, the potential component of $\vecfield$ represents the \emph{convergent} part of \eqref{eq:dyn}, while the incompressible component encapsulates the \emph{non-convergent} part thereof.

In view of the above, a natural idea to characterize convergent and non-convergent behaviors under \eqref{eq:RD} would be to apply Helmholtz's theorem to the vector field
\begin{align}
\label{eq:repfield}
\repfield_{\play\pure_{\play}}(\strat)
	&\defeq \strat_{\play\pure_{\play}}
		\bracks{\pay_{\play}(\pure_{\play};\strat_{-\play}) - \pay_{\play}(\strat)}
	\notag\\
	&= \strat_{\play\pure_{\play}}
		\bracks*{
			\payfield_{\play\pure_{\play}}(\strat)
			- \insum_{\purealt_{\play}\in\pures_{\play}}
				\strat_{\play\purealt_{\play}} \payfield_{\play\purealt_{\play}}(\strat)}
\end{align}
of \eqref{eq:RD} that describes the evolution of the players' mixed strategies under \eqref{eq:EW}.
Unfortunately however, a direct decomposition of $\repfield$ into a potential and incompressible component \textendash\ in the sense of Helmholtz's theorem \textendash\ is not well-aligned with the properties of the underlying game.

To see this, consider the single-player game with actions ``$\pureA$'' and ``$\pureB$'' and payoffs $\pay(\pureA) = 0$ and $\pay(\pureB) = 1$.
Since there is only one player, the game admits the potential function $\pot\of\strat = \pay(\strat) = 0\cdot\strat_{\pureA} + 1\cdot\strat_{\pureB} = \strat_{\pureB}$, so it is a potential one.
However, the \acl{RD} for this toy example are
\begin{equation}
\begin{alignedat}{2}
\dot\strat_{\pureA}
	&= \repfield_{\pureA}(\strat)
	\equiv \strat_{\pureA} \bracks{0 - \pay(\strat)}
	&&= -\strat_{\pureA}\strat_{\pureB}
	\\
\dot\strat_{\pureB}
	&= \repfield_{\pureB}(\strat)
	\equiv \strat_{\pureB} \bracks{1 - \pay(\strat)}
	&&= \strat_{\pureB} - \strat_{\pureB}^{2}
\end{alignedat}
\end{equation}
and a simple check shows that
\(
\cramped{\pd_{\pureB}\repfield_{\pureA}}
	= - \strat_{\pureA}
	\neq 0
	= \cramped{\pd_{\pureA}\repfield_{\pureB}}.
\)
By a routine application of Poincaré's lemma, this further shows that $\repfield(\strat)$ is \emph{not} the gradient of a potential function in the sense of \eqref{eq:Helmholtz}.
As a result, the game is \emph{not} a potential one in the sense of Helmholtz's theorem.

The above shows that the property of \eqref{eq:RD} being a potential system in the sense of Helmholtz (which is more relevant from a dynamical standpoint) is \emph{not} aligned with the property of admitting a potential in the sense of \citet{MS96} (which is more relevant from a game-theoretic standpoint).
In view of this, our goal in the sequel will be to bridge this gap by means of an alternate decomposition in which the discrepancy between ``strategically potential'' and ``dynamically potential'' games disappears.

\subsection{The geometry of the \acl{RD}}

The starting point of our analysis is the observation that, under \eqref{eq:RD}, \emph{players track the direction of steepest individual payoff ascent;}
however, this ascent is not defined relative to the standard Euclidean geometry of $\vecspace$ (which underlies Helmholtz's theorem), but relative to a non-Euclidean structure known as the \define{\shah  metric}.

To make this precise, we begin by introducing the notion of a \define{Riemannian metric}, a fundamental geometric concept which generalizes the ordinary Euclidean scalar product between vectors.
Formally, a \define{Riemannian metric} on an open set $\open$ of $\vecspace$ is a smooth assignment of an \define{inner product} to each $\point\in\open$, \ie a family of bilinear pairings $\inner{\argdot}{\argdot}_{\point}$, $\point\in\open$, that satisfies the following requirements for all $\tvec,\alt\tvec\in\vecspace$ and all $\point\in\open$:
\begin{enumerate}
\item
\define{Symmetry:}
	$\inner{\tvec}{\alt\tvec}_{\point} = \inner{\alt\tvec}{\tvec}_{\point}$.
\item
\define{Positive-definiteness:}
	$\inner{\tvec}{\tvec}_{\point} \geq 0$ with equality iff $\tvec=0$.
\end{enumerate}
This definition can be made more concrete in the standard frame $\setof{\bvec_{\coord}}_{\coord=1}^{\vdim}$ of $\vecspace$ by defining the \define{metric tensor} of $\inner{\argdot}{\argdot}_{\point}$ as the matrix $\gmat(\point) \in \R^{\vdim\times\vdim}$ with entries
\begin{equation}
\label{eq:gmat}
\gmat_{\coord\coordalt}(\point)
	= \inner{\bvec_{\coord}}{\bvec_{\coordalt}}_{\point}
	\quad
	\text{for $\coord,\coordalt=1,\dotsc,\vdim$}.
\end{equation}
The \define{\shah metric} \cite{Sha79} on the positive orthant $\orthant$ of $\vecspace$ is then defined as
\begin{equation}
\label{eq:Shah}
\gmat_{\coord\coordalt}(\point)
	= \delta_{\coord\coordalt} / \point_{\coord}
	\quad
	\text{for all $\point\in\orthant$}
\end{equation}
where $\delta_{\coord\coordalt}$ denotes the standard Kronecker delta.


\begin{figure}
\centering
\footnotesize
\includegraphics[height=40ex]{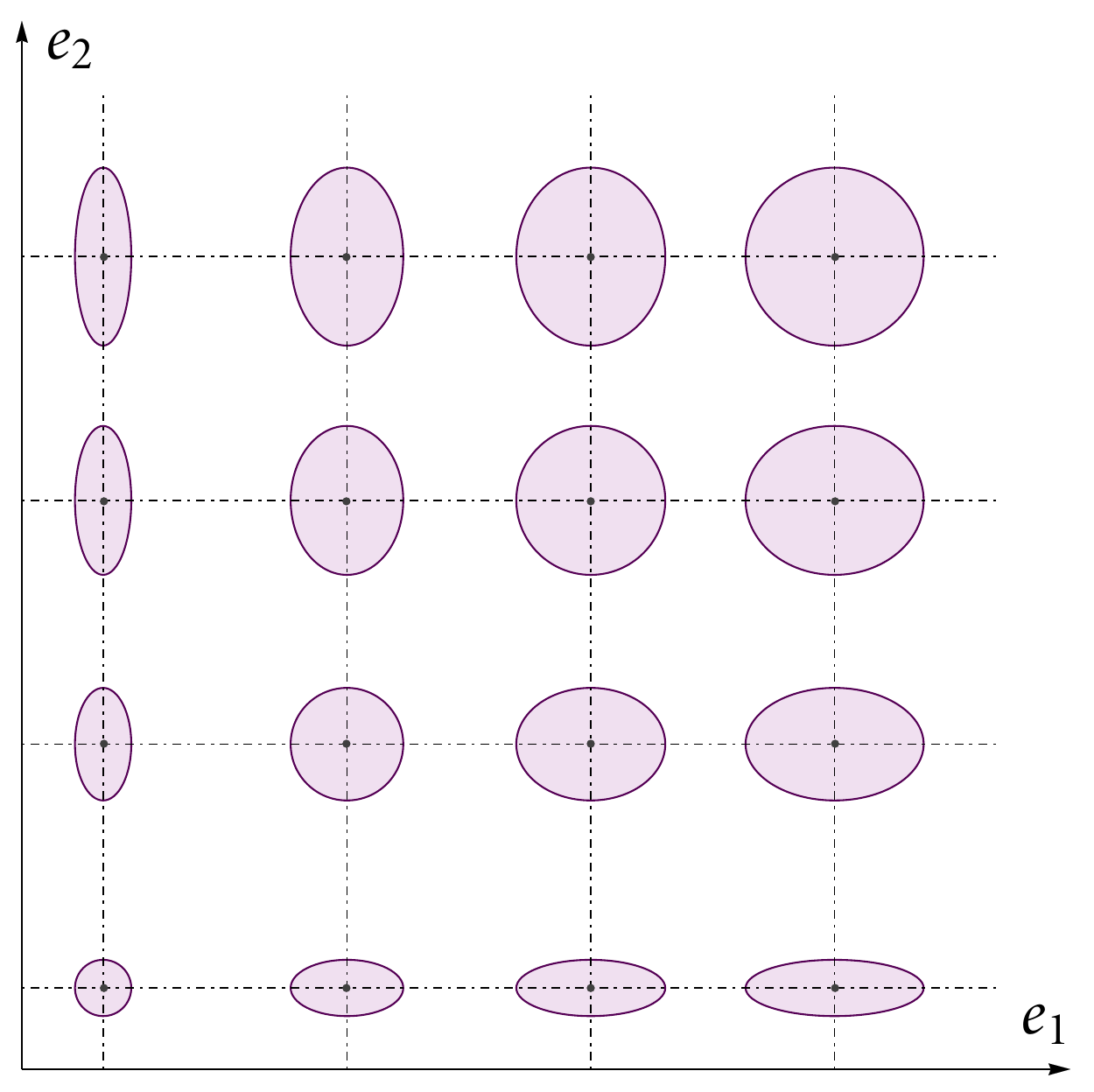}
\qquad
\includegraphics[height=40ex]{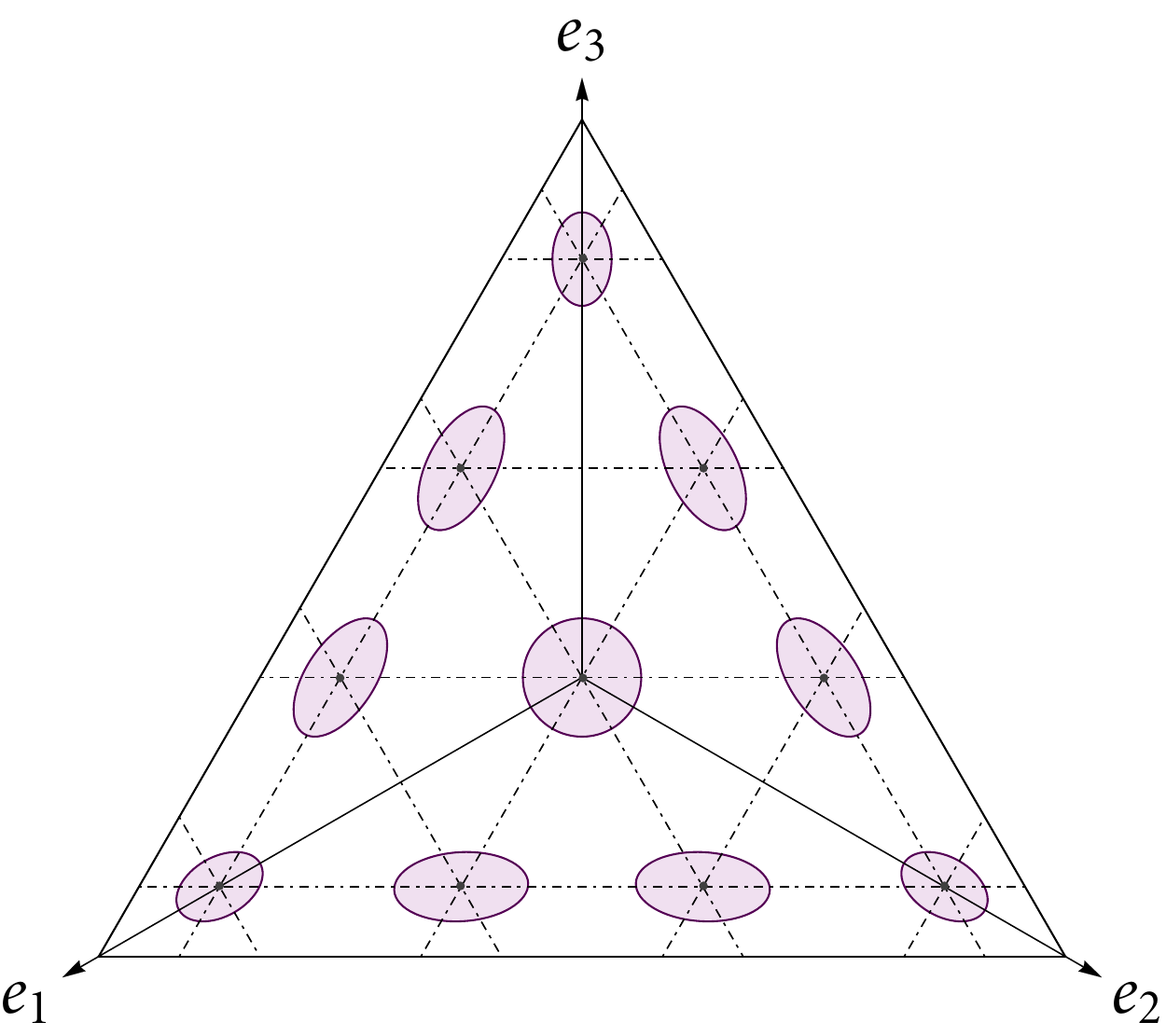}
\\
\caption{Unit balls on the orthant and the simplex under the \shah metric (left and right respectively).
Notice how the \shah metric distorts distances near the boundary and flattens the balls along the axis that they are closest to.}
\label{fig:balls}
\end{figure}


Importantly, the \shah unit spheres $\sphere_{\point} \defeq \setdef{\tvec\in\vecspace}{\inner{\tvec}{\tvec}_{\point} = 1}$ at $\point$ become increasingly flattened along the $\point_{\coord}$-axis as $\point_{\coord}\to0$ (for an illustration, \cref{fig:balls} below).
Because of this distortion, the notion of a ``gradient'' and the ``direction of steepest ascent'' must both be redefined to account for the fact that all displacements of interest take place in the (open) unit simplex $\intsimplex = \setdef{\point\in\orthant}{\sum_{\coord} \point_{\coord} = 1}$ of $\vecspace$.

To do so, we proceed as follows:
Given a differentiable function $\fn\from\orthant\to\R$, we define the \define{its \shah gradient along $\intsimplex$} as the vector field $\grad\fn(\point)$ which is
\begin{enumerate*}
[(\itshape a\upshape)]
\item
tangent to $\intsimplex$;
and
\item
satisfies the defining relation
\end{enumerate*}
\begin{equation}
\label{eq:grad-Shah}
\inner{\grad\fn(\point)}{\tvec}_{\point}
	= \dir\fn(\point;\tvec)
\end{equation}
for all $\point\in\intsimplex$ and all $\tvec$ that are tangent to $\intsimplex$ (\ie $\sum_{\coord}\tvec_{\coord} = 0$ in the standard basis of $\vecspace$).
This relation clearly mirrors the corresponding Euclidean definition $\nabla\fn(\point)^{\top} \cdot \tvec = \dir\fn(\point;\tvec)$,
and as we show in \cref{app:geometry}, it can be equivalently characterized as the direction of ``\define{steepest ascent}'', namely
\begin{equation}
\label{eq:steepest-Shah}
\txs
\grad\fn(\point)
	\propto \argmax\setdef[\big]{\dir\fn(\point;\tvec)}{\tvec\in\sphere_{\point}, \sum_{\coord}\tvec_{\coord} = 0}.
\end{equation}
In other words, $\grad\fn(\point)$ points in the direction that maximizes the rate of increase of $\fn$ at $\point$ among all vectors that are tangent to $\intsimplex$
and have
unit \shah norm.

Now, to obtain an explicit expression for $\grad\fn(\point)$ in the standard basis of $\vecspace$,
note that \eqref{eq:grad-Shah} gives
\begin{equation}
\label{eq:grad-Shah-coords1}
\sum_{\coord=1}^{\vdim} \frac{\bracks{\grad\fn(\point)}_{\coord} \, \tvec_{\coord}}{\point_{\coord}}
	= \sum_{\coord=1}^{\vdim} \frac{\pd\fn}{\pd\point_{\coord}} \tvec_{\coord}
\end{equation}
for all $\tvec\in\vecspace$ such that $\sum_{\coord}\tvec_{\coord} = 0$.
Then, as we show in \cref{app:replicator-geometry},
solving this equation yields the expression
\begin{equation}
\label{eq:grad-Shah-coords}
\bracks{\grad\fn(\point)}_{\coord}
	= \point_{\coord} \bracks*{\frac{\pd\fn}{\pd\point_{\coord}} - \insum_{\coordalt=1}^{\vdim} \point_{\coordalt} \frac{\pd\fn}{\pd\point_{\coordalt}}}.
\end{equation}
This last expression is strongly reminiscent of the vector field $\repfield$ defining \eqref{eq:RD}, a link which we make precise below.

Now, to return to a game-theoretic context,
let $\intstrats_{\play}$ denote the relative interior of the mixed strategy space $\strats_{\play}\equiv \simplex(\pures_{\play})$,
and endow $\orthantplay$ with the \shah metric as above.
We then define the \define{individual payoff gradient} of player $\play\in\players$ as the vector field $\grad_{\play}\pay_{\play}$ which is
\begin{enumerate*}
[(\itshape a\upshape)]
\item
tangent to $\intstrats_{\play}$;
and
\item
satisfies the defining relation
\end{enumerate*}
\begin{equation}
\label{eq:paygrad-Shah}
\inner{\grad_{\play}\pay_{\play}(\strat)}{\tvec_{\play}}
	= \dir\pay_{\play}(\strat;\tvec_{\play})
\end{equation}
for all $\tvec_{\play} \in \R^{\pures_{\play}}$ that are tangent to $\intstrats_{\play}$ at $\strat_{\play}$ (that is, $\sum_{\pure_{\play}\in\pures_{\play}} \tvec_{\play\pure_{\play}} = 0$).
Then, by invoking the explicit expression \eqref{eq:grad-Shah-coords} and observing that $\pd\pay_{\play} / \pd\strat_{\play\pure_{\play}} = \pay_{\play}(\pure_{\play};\strat_{-\play})$, we finally obtain the following geometric characterization of the \acl{RD}.

\begin{restatable}{proposition}{RepShah}
\label{prop:rep-Shah}
Under the \shah metric, \eqref{eq:RD} is equivalent to the steepest individual payoff ascent dynamics
\begin{equation}
\label{eq:dyn-paygrad}
\dot\strat_{\play}
	= \grad_{\play} \pay_{\play}(\strat)
\end{equation}
\ie $\repfield_{\play}(\strat) = \grad_{\play} \pay_{\play}(\strat)$ for all $\play\in\players$.
\end{restatable}

A version of this result appears without proof in \cite{LM15};
for completeness, we defer the details of the proof of \cref{prop:rep-Shah} to \cref{app:replicator-geometry}.
What is more important for our purposes is that, as we show in \cref{sec:app-potential-games}, if the game admits a potential in the sense of \eqref{eq:pot}, combining \cref{prop:rep-Shah,eq:pot-diff,eq:grad-Shah} shows that \eqref{eq:RD} \emph{is a \shah potential system}, that is, $\dot\strat = \grad\pot$.

As far as we are aware, the closest result to \cref{prop:rep-Shah} in the literature is \emph{Kimura's maximum principle} \cite{Kim58} which states that, in potential games, \eqref{eq:RD} is a \shah gradient system \textendash\ thus lifting the discrepancy between the ``dynamic'' and ``strategic'' notions of potential that arose before.
\Cref{prop:rep-Shah} provides a broad generalization of this principle to the effect that, \emph{in any game}, players following \eqref{eq:EW}/\eqref{eq:RD} track the direction of steepest unilateral payoff ascent, provided that displacements are measured relative to the \shah metric.

\subsection{Incompressible games}
\label{sec:incompressible}

Going back to the Helmholtz decomposition \eqref{eq:Helmholtz}, we see that it involves two Euclidean differential operators, the gradient $\nabla\pot$ and the divergence $\nabla\!\cdot\!\incfield$.
\Cref{eq:grad-Shah} shows how to redefine gradients relative to the \shah metric, but the corresponding construction for the divergence is more intricate.
The reason for this is that \eqref{eq:RD} has an inert degree of freedom along $(1,\dotsc,1)$, so the standard definition of the divergence on the ambient space of the simplex is not appropriate \citep[Chap.~6]{carmoRiemannianGeometry1992}.
To circumvent this, we will introduce a more parsimonious representation of \eqref{eq:RD} which has no redundant directions;
we do this first in the case of a single player with action set $\pures = \setof{0,1,\dotsc,\nEffs}$, and only reinstate the player index $\play\in\players$ toward the end of this section.

To proceed, consider the coordinate transformation
\begin{equation}
\push(\strat_{0},\strat_{1},\dotsc,\strat_{\nEffs})
	= (\strat_{1},\dotsc,\strat_{\nEffs})
\end{equation}
which maps the standard unit simplex of $\R^{\nEffs+1}$ to the ``corner of cube''
\begin{equation}
\corcube
	= \setdef{\effstrat\in\clorthant[\R^{\nEffs}]}{\sum_{\eff=1}^{\nEffs} \effstrat_{\eff} \leq 1}
\end{equation}
of $\R^{\nEffs}$ by eliminating $\strat_{0}$ (\ie by replacing the constraint ``summing to $1$'' with ``summing to at most $1$'').%
\footnote{This coordinate transformation goes back at least to \citet{ritzbergerNashField1990}; see also \citet[p.227]{Wei95}.}
Then, for all $\eff=1,\dotsc,\nEffs$, the dynamics \eqref{eq:RD} become
\begin{align}
\ddt{\effstrat_{\eff}}
	&= \dot\strat_{\eff}
	= \strat_{\eff}
		\bracks*{\payfield_{\eff}(\strat) - \insum_{\pure=0}^{\nEffs} \strat_{\pure} \payfield_{\pure}(\strat)}
	\notag\\
	&\quad
	= \effstrat_{\eff}
		\bracks*{\effpayfield_{\eff}(\effstrat) - \insum_{\effalt=1}^{\nEffs} \effstrat_{\effalt} \effpayfield_{\effalt}(\effstrat)}
\label{eq:RD-eff}
\tag{RD$_{0}$}
\end{align}
where, in obvious notation, we set
\begin{equation}
\label{eq:payfield-eff}
\effpayfield_{\eff}(\effstrat)
	= \payfield_{\eff}(\strat) - \payfield_{0}(\strat)
	\quad
	\text{for all $\eff=1,\dotsc,\nEffs$}.
\end{equation}

Seeing as the dynamics evolve in an open set of $\R^{\nEffs}$ (as opposed to a hyperplane of $\R^{\nEffs+1}$), there is no longer any redundancy in the dynamics' degrees of freedom.
In view of this, we will need to ``push forward'' the \shah metric from $\R^{\nEffs+1}$ to $\R^{\nEffs}$ (or, more precisely, the positive orthants thereof) in a way that is compatible with $\chart$.

To do so, we begin by noting that the preimage of the standard frame $\setof{\effbvec_{\eff}}_{\eff=1}^{\nEffs}$ of $\R^{\nEffs}$ restricted to the tangent space $\tanplane$ of $\simplex$ in $\R^{\nEffs+1}$ is
\begin{equation}
\push^{\ast}(\effbvec_{\eff}) \defeq \bvec_{\eff} - \bvec_{0}
	\quad{\text{for all $\eff=1,\dotsc,\nEffs$}}
	\,.
\end{equation}
Accordingly, as we explain in more detail in \cref{app:reduction}, the metric transported in this way to the interior $\intcorcube$ of $\corcube$ will be given by the metric tensor
\begin{equation}
\label{eq:Shah-eff}
\effgmat_{\eff\effalt}(\effstrat)
	= \inner{\effbvec_{\eff}}{\effbvec_{\effalt}}_{\effstrat}
	= \inner{\bvec_{\eff} - \bvec_{0}}{\bvec_{\effalt} - \bvec_{0}}_{\strat}
	= \frac{\delta_{\eff\effalt}}{\strat_{\eff}} + \frac{1}{\strat_{0}}
\end{equation}
for all $\eff,\effalt=1,\dotsc,\nEffs$ and all $\effstrat\in\intcorcube$.

We now have all the ingredients required to define the \shah divergence operator on the interior $\intcorcube$ of $\corcube$.
Since $\intcorcube$ is an open subset of $\R^{\nEffs}$ (which was not the case for the \emph{relative} interior $\intsimplex$ of $\simplex$ in $\R^{\nEffs+1}$), the \define{\shah divergence} of a vector field $\effvecfield\from\intcorcube\to\R^{\nEffs}$ may be defined by the Riemannian expression%
\footnote{The divergence on a Riemannian manifold is a generalization of the divergence operator from vector calculus to curved spaces;
for details, see \cref{app:div-cod-product}.}
\begin{equation}
\label{eq:div-Shah}
\diver\effvecfield(\effstrat)
	\defeq \frac{1}{\sqrt{\det\effgmat(\effstrat)}}
		\sum_{\eff=1}^{\nEffs} \frac{\pd}{\pd\effstrat_{\eff}}
			\parens*{\sqrt{\det\effgmat(\effstrat)} \, \effvecfield_{\eff}(\effstrat)}
\end{equation}
with $\effgmat$ given by \eqref{eq:Shah-eff}.
In particular, when applied to each player's individual steepest ascent payoff field $\repfield_{\play}(\strat) = \grad_{\play} \pay_{\play}(\strat)$, the coordinate expression \eqref{eq:div-Shah} yields
\begin{equation}
\label{eq:div-payfield}
\diver_{\play} \repfield_{\play}(\effstrat)
	= \frac{1}{\sqrt{\det\effgmat_{\play}(\effstrat_{\play})}}
		\sum_{\eff_{\play}=1}^{\nEffs_{\play}} \frac{\pd}{\pd\effstrat_{\play\eff_{\play}}}
			\parens*{\sqrt{\det\effgmat_{\play}(\effstrat_{\play})} \, \repfield_{\play\eff_{\play}}(\effstrat)}
\end{equation}
where, in view of \cref{eq:repfield,eq:payfield-eff}, and in a slight \textendash\ but suggestive \textendash\ abuse of notation, we have set
\begin{align}
\label{eq:repfield-eff}
\repfield_{\play\eff_{\play}}(\effstrat)
	\defeq \effstrat_{\play\eff_{\play}}
		\bracks*{\effpayfield_{\play\eff_{\play}}(\effstrat)
			- \insum_{\effalt_{\play}=1}^{\nEffs_{\play}} \effstrat_{\play\effalt_{\play}} \effpayfield_{\play\effalt_{\play}}(\effstrat)}
\end{align}
with $\effpayfield_{\play\eff_{\play}}(\effstrat) = \payfield_{\play\eff_{\play}}(\strat) - \payfield_{\play,0}(\strat)$ defined as in \eqref{eq:payfield-eff}
for all $\play\in\players$ and all $\eff_{\play} = 1,\dotsc,\nEffs_{\play}$.

With all this in hand, we are finally in a position to define incompressible games:

\begin{restatable}{definition}{DefIncompressible}
\label{def:incompressible}
A finite game $\fingame \equiv \fingamefull$ will be called \define{incompressible} relative to the \shah metric when
\begin{equation}
\label{eq:incompressible}
\diver\repfield(\effstrat)
	\defeq \insum_{\play\in\players} \diver_{\play} \repfield_{\play}(\effstrat)
	= 0.
\end{equation}
\end{restatable}

We should stress here that \cref{def:incompressible} is motivated by purely geometric considerations, and provides a ``complement'' to the class of potential games in a geometric context.
In this regard,
our goal in the sequel will be to use this definition as the basis for a Helmholtz-like decomposition relative to the \shah metric and, in so doing, we understand the dynamic and game-theoretic implications of such a decomposition.
We carry this out in the next section.

%% file: Body/Results.tex

\subsection{A geometric decomposition of games}

To recap, our analysis so far has highlighted the relation between the \shah metric and learning under \eqref{eq:EW}\,/\,\eqref{eq:RD}.
On that account, the first question that we seek to address is whether Helmholtz's theorem can be extended to the present context, and whether such a decomposition resolves the dynamic/strategic disconnect that underlies the ``vanilla'' Helmholtz decomposition.
Our first result below answers this question in the positive.

\begin{restatable}{theorem}{ThDecomposition}
\label{thm:Hodge}
Every finite game $\fingame$ can be decomposed as
\begin{equation}
\label{eq:Hodge}
\fingame
	= \potgame + \incgame
\end{equation}
where
$\potgame$ is potential and $\incgame$ is incompressible.
In particular, at the vector field level, we have
\begin{equation}
\label{eq:Hodge-vec}
\repfield
	= \grad\pot
		+ \incfield
\end{equation}
where $\pot$ is a potential for $\potgame$ and $\incfield$ is incompressible in the sense of \eqref{eq:incompressible}.
\end{restatable}

\cref{thm:Hodge} comes as a consequence of \cref{thm:harmonic}, which relates harmonic to incompressible games, and which we state later in this section.
Because the calculations are fairly lengthy and involved, we defer all relevant details to \cref{app:incompressible}, and we focus here on the game-theoretic implications of \cref{thm:Hodge}.

A first conclusion that can be drawn from \cref{thm:Hodge} is that the decomposition \eqref{eq:Hodge} pinpoints two concrete building blocks of the space of games:
potential games and incompressible games.
With regard to the potential component, \cref{thm:Hodge} resolves the dynamic-strategic disconnect that arose when we applied the standard Helmholtz decomposition to \eqref{eq:RD}:
the component $\potgame$ of \eqref{eq:Hodge} also admits a \shah potential in the sense of \eqref{eq:grad-Shah}, so there is no longer any mismatch between the two viewpoints.

The role of the incompressible component is less transparent, but it is clarified by the striking equivalence below:

\begin{restatable}{theorem}{ThHarmonicIncompressible}
\label{thm:harmonic}
A finite game is harmonic if and only if it is incompressible.
In particular, up to strategic equivalence, the decompositions \eqref{eq:Candogan} and \eqref{eq:Hodge} coincide.
\end{restatable}

This result hinges on a series of geometric calculations involving the explicit coordinate expression of the \shah divergence operator \eqref{eq:div-Shah};
we defer this calculation to \cref{app:incompressible}, where we discuss all relevant details.
What is more important for our purposes is that \cref{thm:harmonic} provides a fairly unexpected \textendash\ and operationally significant \textendash\ interpretation of incompressible games:
even though incompressible games were introduced solely based on their relation with the \shah metric \textendash\ and, through that, to the learning dynamics \eqref{eq:EW} \textendash\ they are characterized by the same ``negative externalities'' property \eqref{eq:harm} which states that the net incentive to deviate toward and/or away from any pure strategy profile is zero.
As we shall see below, this strategic ``conservation of incentives'' is mirrored in the evolution of learning in incompressible / harmonic games under \eqref{eq:EW}.

\subsection{Dynamic considerations}

We now turn to our paper's second major objective:
understanding the behavior of learning under \eqref{eq:EW} in the class of harmonic\,/\,incompressible games.

The first thing to note here is that, as in the Euclidean case, incompressibility is inherently tied to volume preservation.
However, in contrast to the Euclidean case, volumes must now be measured relative to the \shah metric.
The relevant device in our Riemannian setting is the notion of the \define{\shah volume form}, defined on the (open) unit simplex $\intsimplex$ of $\R^{\nEffs+1}$ as
\begin{equation}
\label{eq:vol}
\vol(\borel)
	= \int_{\push(\borel)}
		\!\sqrt{\det\effgmat(\effstrat)} \dd\effstrat_{1}\dotsm \!\dd\effstrat_{\nEffs}
\end{equation}
where $\borel$ is an open subset of $\intsimplex$ 
and $\effgmat(\effstrat)$ is the coordinate representation of the \shah metric in the ``corner-of-cube'' coordinates $\effstrat = \push(\strat)$ of \cref{sec:incompressible}.

As we discuss in \cref{app:geometry}, the Riemannian version of Liouville's theorem states that, if the vector field $\repfield(\strat)$ is incompressible, the dynamics $\dot\strat = \repfield(\strat)$ are \define{volume-preserving} in the sense that
\begin{equation}
\label{eq:preserve}
\vol(\borel_{\time})
	= \vol(\borel_{\tstart})
\end{equation}
where $\borel_{\tstart} \subseteq \intsimplex$ is an open set of initial conditions and $\borel_{\time}$ is the image of $\borel_{\tstart}$ after following the flow of $\repfield$ for time $\time$.
We thus get the following result:

\begin{restatable}{proposition}{PropVolumeConservation}
\label{prop:preserve}
If $\fingame$ is incompressible, \eqref{eq:RD} is volume-preserving under the \shah volume form \eqref{eq:vol}.
\end{restatable}

This result (which we prove and discuss in detail in \cref{app:incompressible}) suggests that \eqref{eq:EW} is unlikely to converge in the class of incompressible \textendash\ and therefore \emph{harmonic} \textendash\ games.
In particular, \cref{prop:preserve}  should be contrasted to a result of \citet{FVGL+20}, who showed that \eqref{eq:EW} is volume-preserving for \emph{every game} relative to the Euclidean volume form on the ``dual'' space of the score variables $\score_{\play}$.
We stress here however that the volume-preservation result of \cite{FVGL+20} applies to \emph{every game}, a property which plays a crucial role in showing that any asymptotically stable state of \eqref{eq:EW}\,/\,\eqref{eq:RD} must be a pure strategy profile (in fact, a strict \acl{NE}) \cite{FVGL+20,Wei95}.
By contrast, \cref{prop:preserve} \emph{does not} apply to all games and essentially, is an equivalence:
if a game is not incompressible, the Riemannian version of Liouville's formula (which we state formally in \cref{app:geometry}) shows that \eqref{eq:RD} is expanding (resp.~contracting) in areas of positive (resp.~negative) divergence, and is not volume-preserving overall.

In this sense, the \shah volume form is more descriptive, and allows for a finer understanding of the flow of \eqref{eq:RD}.
In fact, as we show in \cref{app:incompressible}, incompressibility under the \shah metric induces a further striking structural property:

\begin{restatable}{theorem}{ThmIncompressibleConstantOfMotion}
\label{thm:constant}
If $\fingame$ is incompressible, the induced dynamics \eqref{eq:EW}\,/\,\eqref{eq:RD} admit a constant of motion.
Specifically, there exists a function $\energy\from\intstrats\to\R$ such that $\energy(\stratof{\time}) = \energy(\stratof{\tstart})$ for every initial condition $\stratof{\tstart}\in\intstrats$.
\end{restatable}

This result is surprising because it ties together two drastically different \textendash\ and, to a certain extent, distinctly independent \textendash\ properties:
volume preservation on the one hand, and the existence of conserved quantities on the other.
In the context of learning under \eqref{eq:EW}\,/\,\eqref{eq:RD}, the existence of constants of motion has only been established for very special classes of games, namely two-player zero-sum games with an interior equilibrium \citep[p.~75]{HS98}, positive affine transformations or polymatrix/network versions of the above
\citep{MPP18,
nagarajanChaosOrderSymmetry2020,
PS14},
and certain other games with a min-max structure.
In this regard, \cref{thm:constant} serves to identify a much wider class of $\nPlayers$-player games, not necessarily with a min-max structure, where the no-regret dynamics \eqref{eq:EW} are conservative.


\begin{figure}[tbp]
\centering
\footnotesize
\includegraphics[height=32ex]{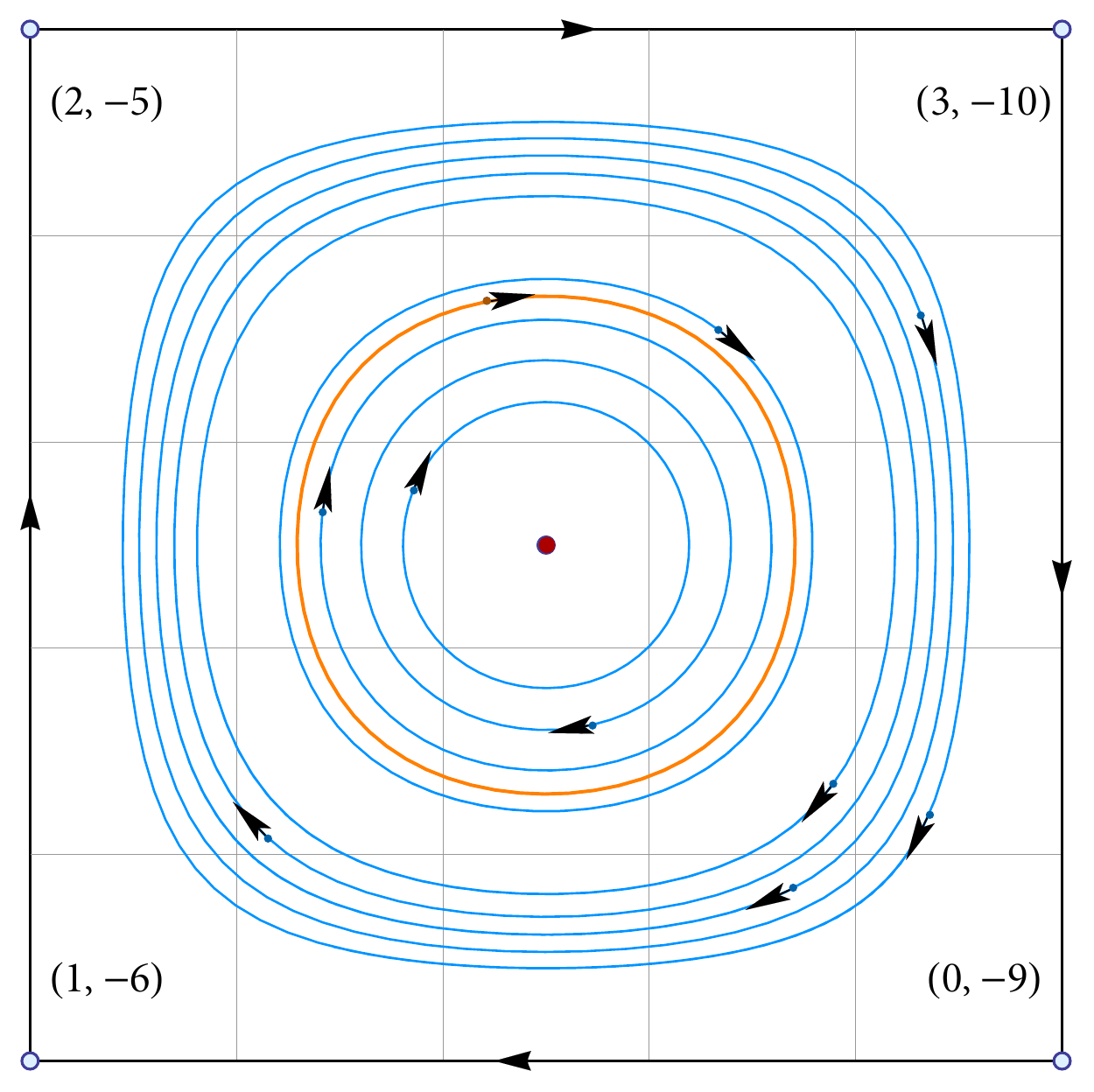}
\quad
\includegraphics[height=32ex]{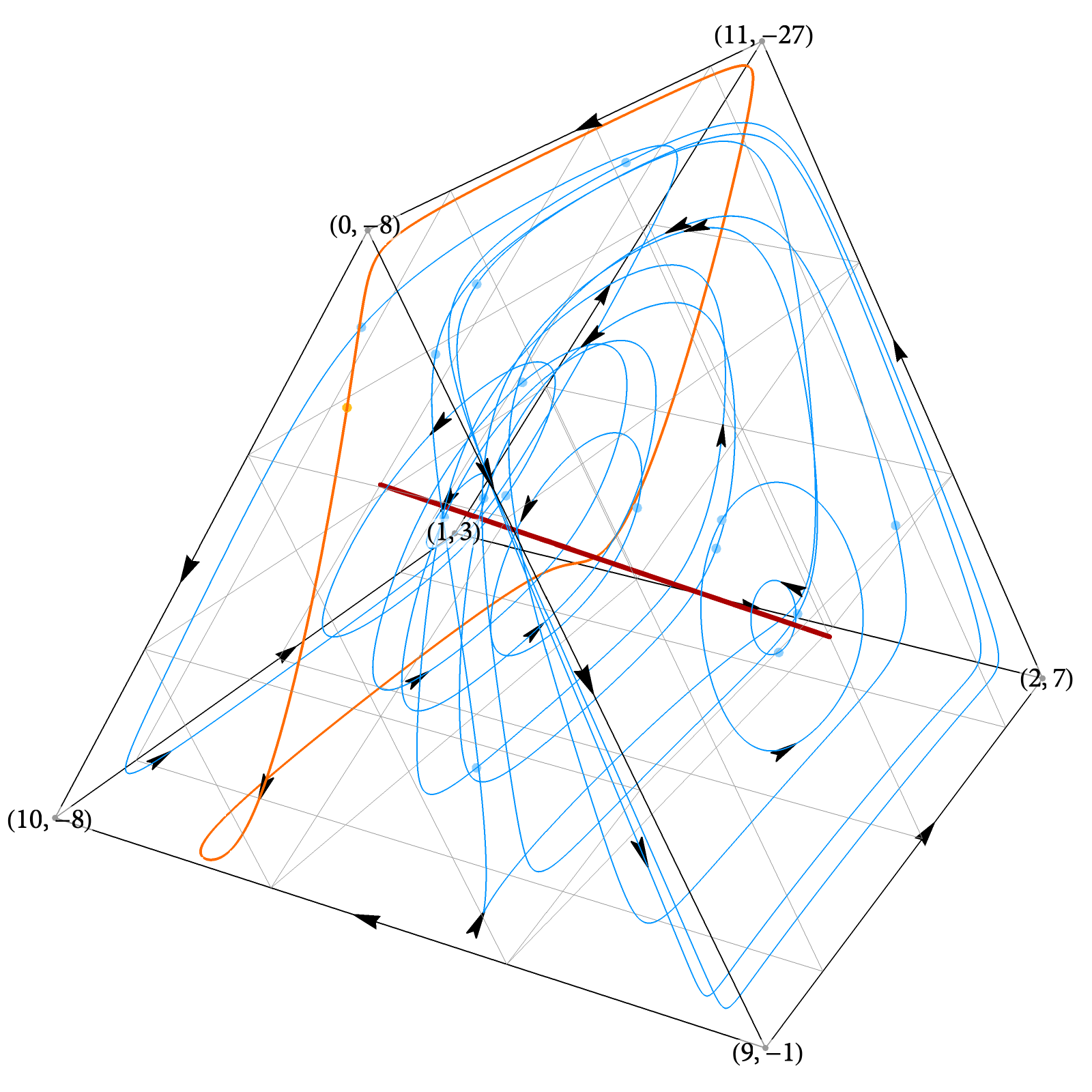}
\quad
\includegraphics[height=32ex]{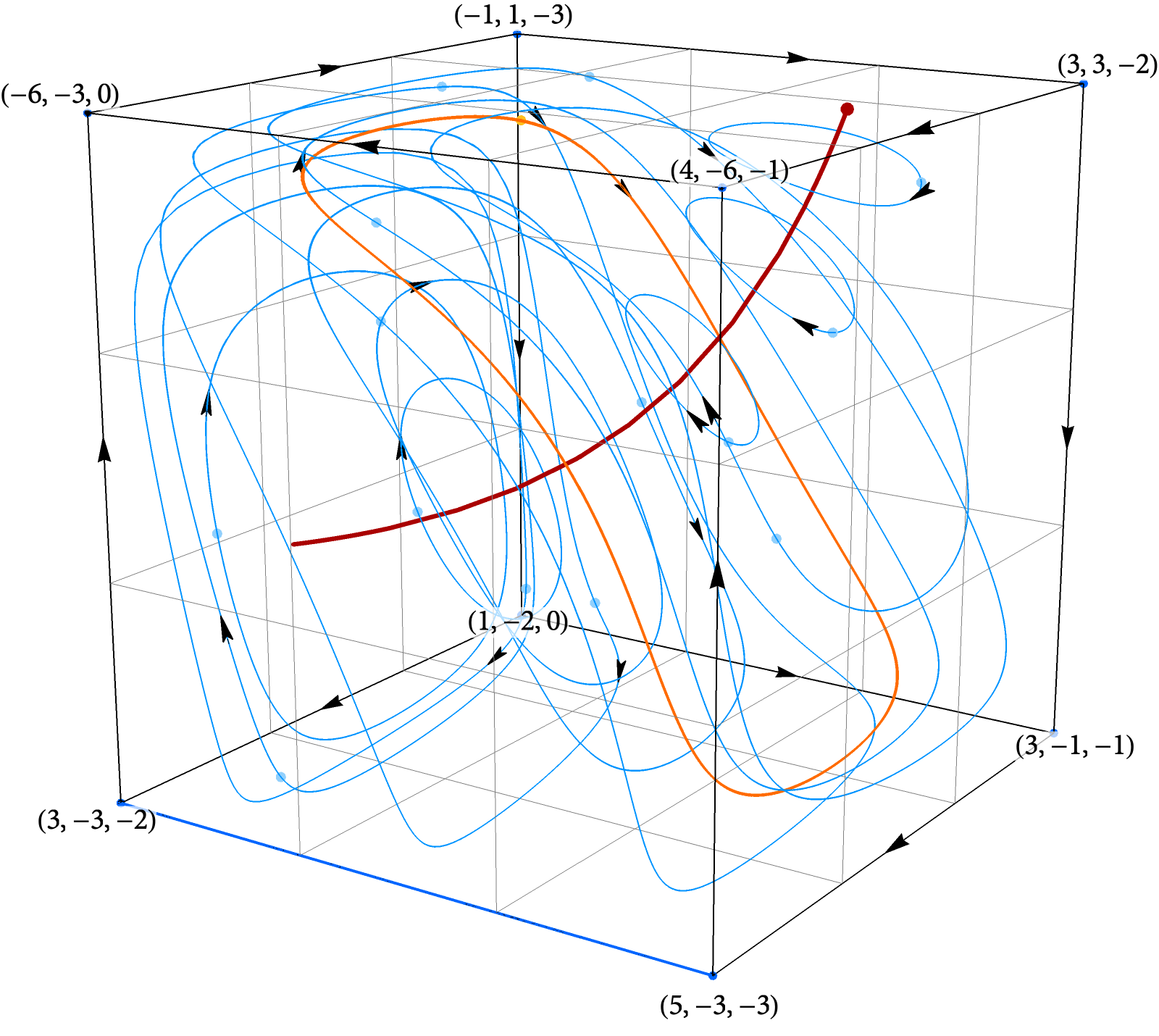}
\caption{The evolution of \eqref{eq:EW} in three randomly generated harmonic games with random initial conditions
(left: $2\times2$;
middle: $3\times2$;
right: $2\times2\times2$).
In all cases, the dynamics are Poincaré recurrent and cycle the game's \aclp{NE} (depicted in red). None of these games is zero-sum; the left and right games are strategically equivalent to a \ac{ZSG}, while the middle is not \textendash\ \cf \cref{rem:harmonic-vs-zero-sum}.
For visual clarity, we have highlighted a randomly selected orbit in each case, and the arrows indicate the direction in which orbits are traversed. In \cref{app:harmonic-potential-trajectories} we include a series of trajectories of \eqref{eq:EW}/\eqref{eq:RD} for a convex combination of potential and a harmonic game, showing how Poincaré recurrence breaks down as the relative magnitude of the potential component increases.}
\label{fig:portraits}
\end{figure}


A concrete consequence of the above is that, in view of \cref{thm:constant}, the trajectories of \eqref{eq:EW} in incrompressible games are constrained to move on the level sets of a certain function.
In fact, as we show in \cref{app:poincare-recurrence-primal-dual}, this function is convex, so its level sets are concentric topological spheres (as boundaries of convex sets, namely the function's sublevel sets).
In turn, this means that the phase space of the dynamics foliates into an ensemble of spheres, each of which constrains the evolution of \eqref{eq:EW}.
In a certain sense, this is the closest that one can get to proving periodic behavior in general dynamics \textendash\ and, in fact, by the Poincaré-Bendixson theorem, it is easy to see that the dynamics are periodic when \eqref{eq:RD} has effective dimension $2$ or $3$ (so in the case of $2\times2$, $2\times3$ and $2\times2\times2$ games, \cf \cref{fig:portraits}).

This brings us to our final result on the long-run behavior of the dynamics \eqref{eq:EW}\,/\,\eqref{eq:RD}:
Even when periodicity fails, it only fails by an arbitrarily small amount.

\begin{restatable}{theorem}{ThIncompressibleRecurrent}
\label{thm:recurrent}
If $\fingame$ is harmonic, the dynamics \eqref{eq:EW}\,/\,\eqref{eq:RD} are Poincaré recurrent.
Specifically, for almost every initialization $\stratof{\tstart}\in\intstrats$, the induced trajectory $\stratof{\time}$ returns arbitrarily close to $\stratof{\tstart}$ infinitely often:
there exists an increasing sequence of times $\curr[\time] \uparrow \infty$ such that $\stratof{\curr[\time]} \to \stratof{\tstart}$.
\end{restatable}

\cref{thm:recurrent} (which we prove in \cref{app:incompressible}) shows that the behavior of \eqref{eq:EW}/\eqref{eq:RD} in harmonic games is orthogonal to their behavior in potential games:
in the latter, every orbit eventually converges to \acl{NE};
in the former, the system's orbits cycle back in an almost-periodic manner to (a neighborhood of) their starting points infinitely often (\cf \cref{fig:portraits}).
This provides a partial negative answer to an open question of \citet{CMOP11} regarding the convergence of the \acl{RD} in harmonic games, and shows that potential and harmonic games are orthogonal to each other also in the sense of learning.

%% file: Body/Discussion.tex

We find the equivalence between harmonic and incompressible games particularly intriguing as it links four otherwise distinct and independent notions:
\begin{enumerate*}
[(\itshape i\hspace*{.5pt}\upshape)]
\item
a standard class of game-theoretic learning schemes (which lead to no-regret, so players become more efficient over time);
\item
the existence of anti-aligned incentives (encoded by the notion of a harmonic game);
and,
through the surprising property of volume preservation,
\item
the existence of a constant of motion
and
\item
Poincaré recurrence (the prototypical manifestation of non-convergent, quasi-periodic behavior).
\end{enumerate*}
\Cref{thm:recurrent} in particular shows that the interplay between these notions is significantly more intricate than what the strong no-regret properties of \eqref{eq:EW} might suggest:
in harmonic games, the players' long-run behavior under the dynamics of \acl{EW} is Poincaré recurrent and fails to converge, even though the empirical distribution of play converges to the game's set of \aclp{CCE}.

Under this light, the decomposition of a game into a potential and a harmonic/incompressible component is strongly reminiscent of Conley's decomposition theorem \cite{Con78} which states that any dynamical system can be decomposed into an attracting, convergent part, and a chain-recurrent part.
Of course, Conley's theorem concerns a decomposition of the game's state space, not the flow itself;
nonetheless, this alignment between the dynamic and strategic components of a game hints at a much deeper connection which opens up many directions for further research.

One such direction concerns the general class of \acf{FTRL} dynamics \citep{SS11,SSS06}, of which \eqref{eq:EW}/\eqref{eq:RD} is a special case:
concretely, we conjecture here that Poincaré recurrence holds for harmonic games in the \emph{entire} class of \ac{FTRL} dynamics.
While close in spirit, the techniques presented in this paper do not extend to the class of \ac{FTRL} dynamics because the analogue of \cref{thm:harmonic} fails to hold for dynamics other than \eqref{eq:EW}/\eqref{eq:RD}, so there is no longer a clear link between incompressibility and harmonicity.
We leave this question open for the future.

%% file: Acks.tex
%
%
This research was supported in part by 
the French National Research Agency (ANR) in the framework of
the PEPR IA FOUNDRY project (ANR-23-PEIA-0003),
the ``Investissements d'avenir'' program (ANR-15-IDEX-02),
the LabEx PERSYVAL (ANR-11-LABX-0025-01),
MIAI@Grenoble Alpes (ANR-19-P3IA-0003),
and
project MIS 5154714 of the National Recovery and Resilience Plan Greece 2.0 funded by the European Union under the NextGenerationEU Program.

%% file: Appendix/App-Geometry.tex

Riemannian geometry, a cornerstone of differential geometry, provides a powerful framework for analyzing spaces endowed with a metric structure.
At its core lies the notion of a \define{Riemannian manifold}, a smooth manifold endowed with a smoothly varying inner product structure on its tangent spaces.
This structure allows for the definition of gradient directions of smooth functions, lengths of curves, angles between tangent vectors, and various notions of curvature, enabling the study and analysis of geometric properties of the manifold, and of dynamical systems defined thereon.
In what follows, we provide a quick dictionary of some of the basic notions that we use throughout our paper.

\subsection{Riemannian manifolds}
A \define{smooth manifold} $\mfld$ of dimension $\dims$ is a Hausdorff, second countable topological space equipped with a collection of \define{local charts} $(U_\alpha, \chartalt_\alpha)$, where each $U_\alpha$ is an open subset of $M$ and each $\chartalt_\alpha: U_\alpha \rightarrow \mathbb{R}^\dims$ is a homeomorphism onto an open subset of $\mathbb{R}^\dims$.
\footnote{The charts are required to satisfy also the compatibility condition that the \define{transition maps} $\chartalt_\alpha \circ \chartalt_\beta^{-1}: \chartalt_\beta(U_\alpha \cap U_\beta) \rightarrow \chartalt_\alpha(U_\alpha \cap U_\beta)$ are smooth whenever $U_\alpha \cap U_\beta \neq \varnothing$.} Points in the co-domain of each local chart are called \define{local coordinates}.

The \define{tangent space} $\tspace_\point \mfld$ of a smooth manifold $\mfld$ at a point $\point \in \mfld$ is the vector space of all derivations on the space of smooth functions defined on an open neighborhood of $\point$; the dimension of $\tspace_\point \mfld$ is the same as the dimension of $\mfld$.

Intuitively, a smooth manifold can be thought of as a topological space that \textit{locally} resembles $\R^{\dims}$; typical example of smooth manifolds include smooth surfaces in the Euclidean space. Building on this intuition, the tangent space $\tspace_\point \mfld$ of $\mfld$ a can be thought of as the space of all possible \quotes{directions} or \quotes{velocities} one can move in from the point $\point$ without leaving the surface. For instance, if $\mfld$ is the Earth's surface, $\tspace_\point \mfld$ at a point $\point$ would be the plane including vectors representing north, south, east, west \textendash\ but not upwards.

\begin{remark}
\label{rem:tangent-space-euclidean-manifold}
Another fundamental example of smooth manifold is the Euclidean space $\R^{\dims}$ itself, with the identity map as a chart. Its tangent space at any point $\point$ is an $\dims$-dimensional vector space, hence isomorphic to $\R^{\dims}$; in the following we will identify the tangent space $\tspace_{\point}\R^{\dims}$ with $\R^{\dims}$ itself.
\endenv
\end{remark}

A \define{Riemannian manifold} $(\mfld, \g)$ is a smooth manifold $\mfld$ together with a smoothly varying positive definite symmetric bilinear form $\g_\point : \tspace_\point \mfld \times \tspace_\point \mfld \rightarrow \R$ defined on each tangent space $\tspace_\point \mfld$ with $\point \in \mfld$. This form, called the \define{Riemannian metric}, assigns to each pair of tangent vectors $\tvec, \tvecalt $ at a point $\point$ a real number $\g_\point(\tvec, \tvecalt)$, satisfying:

\begin{enumerate}
    \item Smoothness: The map $\point \mapsto \g_\point$ is smooth.
    \item Positive definiteness: For all $\point \in \mfld$ and $\tvec \in \tspace_\point \mfld$, $\g_\point(\tvec, \tvec) \geq 0$ with equality iff $\tvec = 0$.
    \item Symmetry: For all $\point \in \mfld$ and $\tvec, \tvecalt \in \tspace_\point \mfld$, $\g_\point(\tvec, \tvecalt) = \g_\point(\tvecalt,  \tvec)$.
\end{enumerate}
Throughout this work use equivalently the notations $\g_{\point}(\tvec, \tvecalt) \equiv \inner{\tvec}{\tvecalt}_{\point} \equiv \tvec^{T} \cdot \g_{\point} \cdot \tvec$ for all $\tvec, \alt\tvec \in\tspace_{\point}\mfld$, where in the third expression $\g_{\point}$ is a $\dim\mfld \times \dim\mfld$-dimensional matrix and $\cdot$ denotes matrix multiplication.

Finally, a \define{vector field $\vfield$ on $\mfld$} is a smooth map $\point \mapsto \vfield\of\point \in \tspace_{\point}\mfld$ that for all points $\point$ on the manifold gives a vector $\tvec = \vfield\of\point$ in the tangent space to $\mfld$ at $\point$. A vector field $\vfield$ on $\mfld$ can be written locally as a linear combination $\vfield = \sum_{\coord = 1}^{\dim\mfld} \vfield^{\coord} \bvec_{\coord}$, where $\vfield^{\coord}\from\mfld\to\R$ is a smooth function on the manifold and $\bvec_{\coord}$ is a \define{basis vector field}, \ie $\bvec_{\coord}\of\point$ is the $\coord$-th element in a basis of $\tspace_{\point}\mfld$, for all $\point \in \mfld$ and all $\coord = 1, \dots, \dim\mfld$. An ordered collection $\{ \bvec_{\coord} \}_{\coord = 1, \dots, \dim\mfld}$ of such basis vector fields is called \define{frame bundle}.

\subsection{Riemannian gradients}
Given a Riemannian manifold $(\mfld, \g)$ and a smooth function $ \fn \from \mfld \to \R$, a Riemannian metric allows to define a special vector field on $\mfld$:
\begin{definition}
The \define{gradient} of $\fn$ is the vector field $\grad\fn$ on $\mfld$ defined by
\begin{equation}
\inner{\grad\fn\of{\point}}{\tvec}_{\point}
	= \dir\fn(\point;\tvec) \in \R \eqcomma
\end{equation}
for all points $\point \in \mfld$ and all tangent vectors $\tvec \in \tspace_{\point}\mfld$.
\end{definition}
The components of $\grad\fn$ can be expressed in any local chart as follows: let $\vfield = \bvec_{\coord}$ be a basis vector field with $\coord \in \setof{1, \dots, \dim\mfld}$, and denote by $ \bracks{\d\fn\of\point}_{\coord} 	\defeq \dir\fn(\point; \bvec_{\coord}\of\point)$ the directional derivative of $\fn$ at $\point$ in the direction of $\bvec_{\coord}\of\point$. 
\begin{lemma}[Components of gradient field]
\label{lemma:sharp-components-inverse-metric}
For all $\point \in \mfld$, the components of $\grad\fn$ are given by the matrix multiplication between the inverse matrix of the metric $\g^{-1}$, and the array of basis directional derivatives $\d\fn$:
\begin{equation}
\label{eq:grad-coords}
\grad\fn(\point) = \g^{-1}(\point) \d\fn(\point) \eqdot
\end{equation}
\end{lemma}
\begin{proof}

Write $\gfield_{\coord}(\point) \defeq \bracks{\grad\fn(\point)}_{\coord}$. By symmetry of $\g$,
\[
\inner{\grad\fn\of{\point}}{\bvec_{\coord}\of{\point}}_{\point}
	= \sum_{\coordalt, \coordaltalt = 1}^{\dim{\mfld}} \g_{\coordalt \coordaltalt}(\point) \, \gfield_{\coordalt}(\point) \, \delta_{\coord \coordaltalt} 
	 = \sum_{\coordalt = 1}^{\dim{\mfld}} \g_{\coordalt \coord}(\point) \, \gfield_{\coordalt}(\point)
	 = \bracks{\g(\point)\gfield(\point)}_{\coord} \eqdot
\]
By definition of gradient this expression is equal to $ \bracks{\d\fn\of\point}_{\coord}$, so we get the matrix equation $\g(\point)\gfield(\point) = \d\fn\of\point$ for all $\point \in \mfld$. Since $\g\of\point$ is positive-definite for all $\point \in \mfld$ we can multiply from the left both sides of this equations by the inverse matrix $\g^{-1}(\point)$ to get~\eqref{eq:grad-coords}.
\end{proof}

\begin{remark}[Euclidean \vs non-Euclidean gradient]
\label{rem:euclidean-vs-riemannian-gradient}
The Euclidean metric in $\vecspace$ is represented by the identity matrix $\g_{\coord\coordalt}(\point) = \delta_{\coord\coordalt}$, from which the familiar result that the gradient of a function \textit{is} the array of basis directional derivatives, or \define{differential}, of $\fn$: $\grad\fn(\point) = \bracks{\pd\fn\of{\point; \bvec_{\coord}}}_{\coord = 1}^{\vdim} = \d\fn(\point)$. In a non-Euclidean setting, the difference between the gradient and the differential of a function is given by the (inverse) metric tensor. \endenv
\end{remark}

\paragraph{Gradients give directions of maximal rate of change}
Given a smooth function $\fn \from \mfld \to \R$ on a Riemannian manifold $(\mfld, \g)$, the gradient of $\fn$ at $\point$ gives the direction of maximal rate of change of $\fn$ at $\point$; we make this precise with the following lemma. 

\begin{lemma}
\label{lemma:gradient-direction-maximal-change}
On a Riemannian manifold $(\mfld, \g)$ let $\dirspace_{\point} \subset \tspace_{\point}\mfld \defeq \setdef{\tvec \in \tspace_{\point}\mfld}{ \norm{\tvec}_{\point} = 1}$ be the set of \define{directions at $\point$}, \ie the set of tangent vectors at $\point$ of unitary norm, where $\norm{\tvec}_{\point} \defeq \sqrt{\inner{\tvec}{\tvec}_{\point}}$ is the the norm induced by the Riemannian inner product. Then for all smooth functions $\fn \from \mfld \to \R$ and for all $\point \in \mfld$,
\begin{equation}
\frac{\grad\fn(\point)}{\norm{\grad\fn(\point)}_{\point}} = \argmax_{\tvec \in \dirspace_{\point}}\setof{\pd\fn(\point; \tvec)} \eqdot
\end{equation}
\end{lemma}
\begin{proof}
By definition of gradient, $\pd\fn(\point; \tvec) = \inner{\grad\fn\of{\point}}{z}_{\point}$. By Cauchy-Schwarz inequality, we have $\abs{\inner{\tvec}{\alt\tvec}_{\point}} \leq  \norm{\tvec}_{\point} \, \norm{\alt\tvec}_{\point}$ for all $\tvec,\alt\tvec \in \tspace_{\point}\mfld$ with equality iff $\tvec \propto \alt\tvec$, thus $\pd\fn(\point; \tvec)$ is maximized in the direction $\tvec \propto \grad\fn(\point)$.
\end{proof}

\subsection{Riemannian divergence}

\label{app:div-cod-product}
In vector calculus, the \define{divergence} is a differential operator mapping a vector field $\vfield$ to a function:
\begin{equation}
\label{eq:euclidean-divergence}
\div{\vfield} = \sum_{\coord = 1}^{\dim\mfld} \frac{\pd}{\pd\point_{\coord}}   \vfield^{\coord} \quad \text{[Euclidean]} \eqdot
\end{equation}
This operator captures how much the field is locally spreading out or converging at a given point: loosely speaking, if the divergence is positive in the neighborhood of a point, the vector field is locally spreading out (\eg the outward-radial field $\vfield(x,y) = (x,y)$); if it is negative, the vector field is locally converging (\eg the inward-radial field  $\vfield(x,y) = (-x,-y)$); and if it is zero, the field is neither locally spreading out nor converging (\eg the hyperbolic field  $\vfield(x,y) = (x,-y)$ or the spherical field $\vfield(x,y) = (y,-x)$).

Let $\vfields$ and  $\funcs$ denote respectively the space of vector fields and smooth functions on a Riemannian manifold $(\mfld, \g)$. To generalize the divergence operator to this setting one must take into account how the \textit{volume element} of the metric $\g$ changes from point to point, and this in turn depends on the \textit{determinant} $\det{\g}$ \textendash\ \cf \cref{eq:riemannian-volume}. With this idea in mind, we give the following generalization of the Euclidean divergence operator to a Riemannian setting~\citep{leeIntroductionSmoothManifolds2012, jostRiemannianGeometryGeometric2017, leeIntroductionRiemannianManifolds2018}.
\begin{definition}
\label{def:riemannian-divergence}
The \define{divergence operator} $\div\from\vfields\to\funcs$ on a Riemannian manifold $(\mfld, \g)$ is defined by
\begin{equation}
\label{eq:riemannian-divergence}
\div{\vfield} =  \frac{1}{\sqrt{\detg}}\sum_{\coord = 1}^{\dim\mfld} \frac{\pd}{\pd\point_{\coord}} \ll \sqrt{\detg} \vfield^{\coord}\rr \quad \text{for all } \vfield \in \vfields \eqcomma
\end{equation}
where $\detg \defeq \det{\g}$.
\endenv
\end{definition}

\begin{remark}
If $(\mfld, \g)$ is the standard Euclidean space equipped with the Euclidean metric then the constant term $\sqrt\detg$ is not affected by the partial derivatives and cancels out with $(\sqrt\detg)^{-1}$, giving back the familiar \cref{eq:euclidean-divergence}.
\endenv
\end{remark}

Note that \cref{eq:riemannian-divergence} can be rewritten by product rule as
\begin{equation}
\label{eq:riemannian-divergence-two-terms}
\div{\vfield} =  \sum_{\coord = 1}^{\dim\mfld} \ll \frac{\de_{\coord}\sqrt{\detg}}{\sqrt{\detg}} \rr \vfield^{\coord} + \sum_{\coord = 1}^{\dim\mfld} \de_{\coord}\vfield^{\coord} \eqcomma
\end{equation}
where $\pd_{\coord}$ is a shorthand for $\frac{\pd}{\pd\point_{\coord}}$.
\begin{example}[Divergence on the sphere]
The determinant of the Euclidean metric in $\mathbb{R}^3$ induced on the unit sphere in standard spherical coordinates fulfills $\sqrt{\det{g}} = \sin\theta$, so by \cref{eq:riemannian-divergence-two-terms} the divergence of the vector field $\vfield(\theta, \phi)= (\vfield^\theta, \vfield^\phi)$ is
\[
\div{\vfield} = \partial_\theta \vfield^\theta + \partial_\phi \vfield^\phi + \frac{\vfield^\theta}{\tan{\theta}} \eqdot
\]

In particular the divergence of a longitudinal vector field $\vfield = (1,0)$ is $\frac{1}{\tan\theta}$, which diverges to infinity a the north pole, is zero at the equator, and diverges to minus infinity at the south pole. This captures the fact that a small set of initial conditions starting close to the north pole and evolving along flow of $\vfield$ quickly expands moving towards the equator; the rate of expansion decreases until the equator is crossed, after which the flow lines converge at increasingly higher rate towards the south pole. 

Conversely, the divergence of a latitudinal vector field $\vfield = (0,1)$ is $0$, capturing the fact that the volume of a small set of initial conditions remains constant along the flowlines parallel to the equator.
\endenv
\end{example}
\paragraph{Riemannian divergence on product manifold}
In this section we show that the divergence operators on two Riemannian manifolds naturally induce a divergence operator on the product manifold.

Let $(\mfld, \g_{\mfld})$ and $(\alt\mfld, \g_{\alt\mfld})$ be Riemannian manifolds with coordinates $\point$ and $\pointalt$ respectively, with $\g_{\mfld}$ represented by the matrix $(\g_{\mfld})_{ij}$ for $i, j  = 1, \dots,  \dim\mfld$, and $\g_{\alt\mfld}$ represented by the matrix $(\g_{\alt\mfld})_{hk}$ for $h, k  = 1, \dots,  \dim\alt\mfld$. Consider the \define{product manifold} $P = \mfld \times \alt\mfld$ with $\dim{P} = \dim{\mfld} + \dim{\alt\mfld}$, coordinates $(\point,\pointalt)$, and metric $\g(\point,\pointalt) \defeq \g_{\mfld}(\point) + \g_{\alt\mfld}(\pointalt)$. The matrix representing the metric $\g$ is the block diagonal matrix with the matrices of $\g_{\mfld}$ and $\g_{\alt\mfld}$ on the diagonal, so $\det{\g} = \det{\g_{\mfld}} \det{\g_{\alt\mfld}}$. In the following we denote $\detg \defeq \det{\g} = \detg_{\mfld}\detg_{\alt\mfld}$.

A vector field $\vfield$ on $\mfld$ locally written as $\vfield(\point) = \insum_{i = 1}^{\dim\mfld} \vfield^{i}(\point) \bvec_{\point_{i}}$ can be naturally seen as a vector field $\vfield$ on $P$, given by $\vfield(\point, \alt\point) = \insum_{A = 1}^{\dim{P}} \vfield^{A}(\point, \alt\point) \bvec_{A} = \insum_{i = 1}^{\dim{\mfld}} \vfield^{i}(\point) \bvec_{\point_{i}} + \insum_{j = 1}^{\dim{\alt\mfld}} 0 \,  \bvec_{\pointalt_{j}} $. With this identification in mind, two vector fields $\vfield$ on $\mfld$ and $\alt\vfield$ on $\alt\mfld$ naturally give a vector field $Z$ on $P$ by $Z(\point,\pointalt) \defeq \vfield(\point)+\alt\vfield(\pointalt) = \insum_{i = 1}^{\dim\mfld} \vfield^{i}(\point) \bvec_{\point_{i}} + \insum_{j = 1}^{\dim\alt\mfld} \alt\vfield^{j}(\pointalt)\bvec_{\pointalt_{j}} $.

\begin{lemma}[The divergence operator does not mix coordinates]
\label{lemma:div-does-not-mix}
Let $\vfield$ be a vector field on a Riemannian manifold $(\mfld,\g)$ and $\alt\vfield$ a vector field on a Riemannian manifold $(\alt\mfld,\alt\g)$. Then the vector field $Z = \vfield + \alt\vfield$ on the product manifold $(P = \mfld \times \alt\mfld, \g = \g_{\mfld} + \g_{\alt\mfld})$ fulfills
\begin{equation}
\div_{\g}{Z} = \div_{\g_{\mfld}}{\vfield} + \div_{\g_{\alt\mfld}}{\alt\vfield} \eqdot
\end{equation}
\end{lemma}
\begin{proof} The $\div{}$ operator acts linearly on vector fields~\citep{leeIntroductionSmoothManifolds2012}, so
\[
\div_{\g}{Z} = \div_{\g}{\vfield} + \div_{\g}{\alt\vfield} \quad \text{by linearity of } \div{} \eqdot
\]
So if we show that $\div_{\g}{\vfield} = \div_{\g_{\mfld}}{\vfield}$ we are done. This is true since the coordinates of the two manifolds remain decoupled under the Cartesian product operation, so they are acted upon only by derivatives of the corresponding type:
\[
\begin{split}
\div_{\g}{\vfield} & =  \sum_{A = 1}^{\dim{P}} \ll \frac{\de_{A} \sqrt{\detg}}{\sqrt{\detg}} \rr \vfield^{A} + \sum_{A = 1}^{\dim{P}} \de_{A}\vfield^{A}   \\
& = \sum_{i = 1}^{ \dim{\mfld}}  \ll \frac{\de_{\point_{i}} \ll\sqrt{\detg_{\mfld}}\sqrt{\detg_{\alt\mfld}}\rr}{\sqrt{\detg_{\mfld}}\sqrt{\detg_{\alt\mfld}}} \rr \vfield^{i} + \sum_{i = 1}^{\dim{\mfld}} \de_{\point_{i}}\vfield^{i} + \sum_{j = 1}^{\dim{\alt\mfld}}0 \eqdot  \\
\end{split}
\]
Now $\de_{\point_{i}} \ll\sqrt{\detg_{\mfld}}\sqrt{\detg_{\alt\mfld}}\rr = \de_{\point_{i}} \ll\sqrt{\detg_{\mfld}}\rr\sqrt{\detg_{\alt\mfld}}$, so the terms containing $\sqrt{\detg_{\alt\mfld}}$ simplify:
\[
\div_{\g}{\vfield}  = \sum_{i = 1}^{ \dim{\mfld}}  \ll \frac{\de_{\point_{i}} \ll\sqrt{\detg_{\mfld}}\rr}{\sqrt{\detg_{\mfld}}} \rr \vfield^{i} + \sum_{i = 1}^{\dim{\mfld}} \de_{\point_{i}}\vfield^{i}   = \div_{\g_{\mfld}}{\vfield} \eqdot \qedhere
\]
\end{proof}

\subsection{Flows on manifolds}
\label{sec:flows-manifolds}

We recall here a few concepts from the theory of dynamical systems on manifolds. We refer the reader to \citet[Ch.~9,16]{leeIntroductionSmoothManifolds2012} for the general theory and to \citet{FVGL+20} for a concise treatment in the context of no-regret learning. For a detailed account of the theory of ordinary differential equations and deterministic dynamical systems in continuous time in the context of multipopulation evolutionary dynamics we refer the reader to the excellent introduction by \citet[Ch. 6]{Wei95}, and in particular to Section 6.6 for a general discussion on the Euclidean version of Liouville's theorem, and to sections 5.2.2 and 5.8.2 for relevant applications. 

Given a smooth vector field $\vfield \in \vfields$ on a smooth manifold $\mfld$, a \define{smooth global integral curve of $\vfield$} is a smooth curve $\curve\from\R\to\mfld$ such that $\dot{\curve}(\time) = \vfield(\curve(\time))$ for all $\time \in \R$. The point $\curve(0)$ is called \define{starting point of $\curve$}. If a smooth global integral curve $\curve$ of $\vfield$ with starting point $\point$ exists, then it is the unique maximal solution to the initial value problem
\begin{equation}
\tag{IVP}
\label{eq:dynamical-system}
\dot{\curve} = \vfield(\curve), \quad \curve(0) = \point \eqdot
\end{equation}

Given a smooth manifold $\mfld$, a \define{smooth global flow} on $\mfld$ is a smooth map $\flow\from \R \times \mfld \to\mfld$ such that for all $\time, \timealt \in \R$ and $\point \in \mfld$, $\flow(\time, \flow(\timealt, \point)) = \flow(\time+\timealt, \point)$ and $\flow(0, \point) = \point$. Given a smooth global flow, fixing $\time \in \R$ one can define the \define{orbit map} $\flowt\from\mfld\to\mfld$ by $\flowt(\point) = \flow(\time, \point)$; the orbit map of a smooth global flow can be shown to be a diffeomorphism of $\mfld$ onto itself with inverse $\ll \flowt \rr^{-1} = \flow_{-\time}$. Similarly, by fixing $\point \in \mfld$ one can define the curve $\flowp \from \R \to \mfld$ by $\flowp(\time) = \flow(\time, \point)$.

Given a smooth global flow $\flow$ and a smooth vector field $\vfield \in \vfields$ on a smooth manifold $\mfld$, we say that \textit{$\flow$ is the flow of $\vfield$} if $\vfield \ll \flowp(\time)  \rr = \dot{\flowp}(\time)$ for all $\time \in \R$ and $\point \in \mfld$. If $\vfield$ admits a global flow, then $\flowp\from\R\to\mfld$ is an integral curve of $\vfield$ with starting point $\point$, hence a solution to the initial value problem \eqref{eq:dynamical-system} \textendash\ equivalently, the orbit map $\flowt\from\mfld\to\mfld$ maps any initial condition $\point \in \mfld$ to the point $\curve\of\time$, where $\curve$ is the maximal solution to the initial value problem \eqref{eq:dynamical-system}.

A vector field may not always admit a global flow, since it may not always be the case that every integral curve is defined for all time. The \define{Fundamental Theorem of Flows} \citep[Th. 9.12]{leeIntroductionSmoothManifolds2012} asserts that every smooth vector field on a smooth manifold determines a unique \textit{local maximal} smooth flow;\footnote{The definition of local maximal flow is analogue to that of global flow, restricting the domain to a suitable open subset $\flowD \subseteq \R \times \mfld$.} the proof is an application of the existence, uniqueness, and smoothness theorem for solutions of ordinary differential equations. For the scope of this work note that the trajectories of \eqref{eq:RD} on $\intstrats$ are defined for all $\time \in \R$, so the replicator vector field $\repfield$ defines a smooth global flow on $\intstrats$.

Next, we look at the relation between the Riemannian divergence of a vector field defined in \cref{app:div-cod-product}, and the \define{volume} of a set of initial conditions evolving along the flow of such vector field. We warm up in an Euclidean setting, before moving to a Riemannian one. 

\paragraph{The Euclidean Liouville’s theorem}
Consider a vector field $\vfield \in \vfields$ on $ \mfld = \R^{\dims}$ that admits a global flow $\flow \from \R \times \R^{\dims} \to \R^{\dims}$.
For any open set $\open \subseteq \R^{\dims}$ and any $\time \in \R$ denote by $\open_{\time}$ the image of $\open$ under the orbit map $\flowt\from\R^{\dims}\to\R^{\dims}$:
\begin{equation}
\label{eq:open_time}
\open_{\time} \defeq \flowt\of\open = \setdef{\flowt(\point)}{\point \in \open} \subseteq \R^{\dims}  \eqdot
\end{equation}
Note that $\open_{\time = 0} = \open$, since the orbit map $\flow_0$ is the identity map on $\R^{\dims}$.

A fundamental result of classical mechanics known as \define{Liouville’s theorem}~\cite{Arn89} relates the Euclidean divergence \eqref{eq:euclidean-divergence} of the vector field $\vfield$, which is a function $\div{\vfield}\from\R^{\dims}\to\R$, with the Euclidean volume of an open set of initial conditions evolving along the flow of $\vfield$:
\begin{theorem*}[Euclidean Liouville's theorem]
\label{th:liouville-euclidean}
Given a smooth vector field $\vfield$ in $\R^{\dims}$ and an open set $\open \subseteq \R^{\dims}$,
\begin{equation}
\label{eq:liouville-euclidean}
\frac{d}{d\time} \vol\of{\open_{\time}}  = \int_{\open_{\time}} \div{\vfield} \d\point \eqcomma
\end{equation}
for all $\time \in \R$ such that the flow of $\vfield$ is defined.
\end{theorem*}
\begin{proof}
See \eg \citet[Ch. 3]{Arn89}.
\end{proof}
If a map $\phi\from\R^{\dims}\to\R^{\dims}$ fulfills $\vol(\open) = \vol(\phi\open)$ for all open subsets $\open \subseteq \mfld$ we say that the map is \define{volume-preserving.} An immediate corollary of Liouville's theorem is that the orbit maps of vector fields with zero divergence are volume-preserving:
\begin{corollary*}[Conservation of Euclidean volume]
If a vector field $\vfield$ in $\R^{\dims}$ fulfills $\div\vfield = 0$, then
\begin{equation}
\vol\of{\open_{\time}} = \vol\of{\open}
\end{equation}
for all open sets $\open \subseteq \R^{\dims}$ and all $\time \in \R$ such that the flow of $\vfield$ is defined.
\end{corollary*}
\begin{proof}
If $\div\vfield = 0$ the right hand side of \cref{eq:liouville-euclidean} vanishes, hence $\vol{\open_{\time}}$ is constant whenever the flow of $\vfield$ is defined.
\end{proof}
\paragraph{The Riemannian Liouville’s theorem}
The constructions of the previous paragraph generalize to the more general setting of Riemannian manifolds~\citep[Ch. 16]{leeIntroductionSmoothManifolds2012}.
Given a smooth vector field $\vfield \in \vfields$ that admits a global flow $\flow$ on a Riemannian manifold $(\mfld, \gmat)$ of dimension $\dims$, let
$
\open_{\time} = \flowt\of\open
$ be the image of any open subset $\open \subseteq \mfld$ under the orbit map $\flowt \from \mfld \to \mfld$, as in \cref{eq:open_time}. To generalize the Euclidean Liouville's theorem to this Riemannian setting we need the appropriate notions of \define{divergence of a vector field} and \define{volume of an open set} on a Riemannian manifold. The appropriate generalization of the divergence operator is given by \cref{eq:riemannian-divergence}; the appropriate notion of volume on a Riemannian manifold is the following~\citep[Ch. 16]{leeIntroductionSmoothManifolds2012}: If $\open$ is an open subset completely contained in the domain of a single smooth chart $(\mathcal{V}, \chartalt)$ of $\mfld$, then its \define{Riemannian volume} is\footnote{The definition extends to arbitrary open subsets of $\mfld$ by a partition of unity argument.} 
\begin{equation}
\label{eq:riemannian-volume}
\vol{\open} = \int_{\chartalt\of\open} \sqrt{\det\tilde{\g}(\tilde{\point})}  \, d\tilde{\point} \eqcomma
\end{equation}
where $\chartalt\from\open\to\R^{\dims}$ is an homeomorphism onto an open subset $\tilde{\open} \defeq \chartalt(\open)$ of $\R^\dims$ mapping $\point \in \mfld$ to $\tilde{\point} \in \tilde{\open}$; and $\tilde{\g}$ is the \define{effective representation} of the Riemannian metric $\g$ on $\tilde{\open}$ (\cf \cref{sec:eff-rep-metric}).

With these definitions at hand we can state the Riemannian version of Liouville’s theorem:
\begin{theorem*}[Riemannian Liouville's theorem]
\label{th:liouville-riemannian}
Given a vector field $\vfield \in \vfields$ on a Riemannian manifold $(\mfld, \g)$ and an open set $\open \subseteq \mfld$,
\begin{equation}
\label{eq:liouville-riemannian}
\frac{d}{d\time} \vol\of{\open_{\time}}  = \int_{\chartalt(\open_{\time})} \div{\vfield} \d\point 
\end{equation}
for all $\time \in \R$ such that the flow $\flow$ of $\vfield$ is defined, where $\open_{\time} = \flowt\of\open$ and $\chart$ is a chart whose domain contains $\open_{\time}$\footnote{As discussed in \citet{leeIntroductionSmoothManifolds2012} the result does not depend on the choice of smooth chart whose domain contains $\open_{\time}$.}.
\end{theorem*}
\begin{proof}
See \eg \citet[Ch. 16]{leeIntroductionSmoothManifolds2012}. 
\end{proof}
As in the Euclidean case, if a map $\phi\from\mfld\to\mfld$ fulfills $\vol(\open) = \vol(\phi\open)$ for all open subsets $\open \subseteq \mfld$ we say that the map is \define{volume-preserving}; and the orbit maps of vector fields with zero Riemannian divergence are volume-preserving.
\begin{corollary*}[Conservation of Riemannian volume]
If a vector field $\vfield \in \vfields$ on a Riemannian manifold fulfills $\div\vfield = 0$ then
\begin{equation}
\label{eq:zero-riemannian-divergence-volume-conservation}
\vol\of{\open_{\time}} = \vol\of{\open} \eqcomma
\end{equation}
where $\open_{\time} = \flowt\of\open$, for all open sets $\open \subseteq \mfld$ and all $\time \in \R$ such that the flow $\flow$ of $\vfield$ is defined.
\end{corollary*}
\begin{proof}
The proof is identical to the one for the Euclidean counterpart.
\end{proof}
\paragraph{Poincaré recurrence}
The last notion we need is that of \define{Poincaré recurrence}, a property of volume-preserving maps on sets of finite volume.
We present a measure-theoretic version of Poincaré's classical recurrence theorem, and adapt it to our Riemannian framework.

Given a measure space\footnote{Recall that a measure space $(\ms,\me)$ is a set $\ms$ with a countable-addictive function $\me$ from the sigma-algebra $\Sigma$ of $\ms$ into the nonnegative real numbers (including infinity) such that $\ms(\emptyset) = 0$, and that any element in $\Sigma$ is called a \define{measurable} subset of $\ms$.} $(\ms,\me)$, we say that $(\ms,\me)$ is \define{finite} if $\me(\ms) < \infty$, and that a map $\phi\from\ms\to\ms$ is \define{measure preserving} if $\me(\phi \open) = \me(\open)$ for all measurable subsets $\open \subseteq \ms$. Given a finite measure that is invariant under some map one has the following theorem~\citep{bekkaErgodicTheoryTopological2000}:
\begin{theorem*}[Poincaré \textendash\ Measure setting]
Let $(\ms, \me)$ be a finite measure space, and let $\phi \from \ms \to \ms$ be a measure preserving mapping. Let $\open$ be a measurable subset of $\ms$. Then almost every point $\point \in \open$ is infinitely recurrent with respect to $\open$, that is, the set $\setdef{n \in \N}{\phi^{n}\point \in \open}$ is infinite.
\end{theorem*}
\begin{proof}
See \eg \citet[Th. 1.7]{bekkaErgodicTheoryTopological2000}.
\end{proof}
\begin{remark}
\label{rem:riemannian-mfld-is-measure-space}
The Riemannian volume \cref{eq:riemannian-volume} on a Riemannian manifold $(\mfld, \g)$ defines a measure $\me$ on the Borel sigma-algebra of $\mfld$ by $\me(\open) = \vol{\open} $, hence a Riemannian manifold is in particular a measure space~\citep{spivakCalculusManifoldsModern1965}, on which Poincaré's theorem applies. Furthermore every Riemannian manifold is a separable metric space,\footnote{The distance between two points being the infimum of the lengths of piecewise geodesics joining them \citep{munkresTopology1999, leeIntroductionRiemannianManifolds2018}.} so one can formulate a Riemannian version of Poincaré's theorem: given a Riemannian manifold $(\mfld, \g)$ of finite volume and a volume-preserving map $\phi\from\mfld\to\mfld$, almost every point $\point \in \mfld$ is $\phi$-recurrent, that is, there is a strictly increasing sequence of integers $\curr[\time] \uparrow \infty$ such that $\lim_{\run \to \infty}\phi^{\curr[\time]}\point \to \point$; see \citet[Corollary 1.8]{bekkaErgodicTheoryTopological2000}.
\end{remark}
The tools presented in this appendix will be used in \cref{app:incompressible} to prove some of the main results of this paper, namely that a game is incompressible if and only if it is harmonic (via a Riemannian divergence operator); that \acl{RD} on incompressible games are volume-preserving with respect to a non-Euclidean Riemannian structure (via Liouville's theorem); and that \acl{RD} on incompressible games exhibit Poincaré recurrence (via Poincaré's theorem).

%% file: Appendix/App-Reduction.tex

\begin{figure}
\centering
\begin{tikzpicture}[scale=2]
\draw[->] (-0.5,0) -- (1.5,0) node[below] {$\effstrat_1$}; 
\draw[->] (0,-0.5) -- (0,1.5) node[left] {$\effstrat_2$}; 
\node[blue] at (0.35,0.35) {$\effstrat$}; 
\fill[blue, opacity=0.3] (0,0) -- (0,1) -- (1,0) -- cycle; 
\node at (0.5,0.8) {$\intcorcube$}; 
\fill[red, opacity=0.1] (-0.7,-0.7) -- (-0.7,+1.3) -- (+1.2,+1.3) -- (+1.2,-0.7) -- cycle; 
\node at (-0.3,1.0) {$\tspace\intcorcube$}; 
\node at (1.0,1.2) {$\R^{\nEffs}$}; 
\draw[red, ->] (0,0) -- (1,0) node[above] {$\effbvec_{1}$}; 
\draw[red, ->] (0,0) -- (0,1) node[above] {$\effbvec_{2}$}; 

\begin{scope}[shift={(3,0)}]
\draw[->] (0,0,0) -- (1.5,0,0) node[below] {$\strat_0$}; 
\draw[->] (0,0,0) -- (0,1.5,0) node[left] {$\strat_1$};  
\draw[->] (0,0,0) -- (0,0,1.5) node[above] {$\strat_2$}; 
\node[blue] at (0.5,0.5,0.5) {$\strat$}; 

\draw[red, ->] (0,0,0) -- (-0.4,0.4,0) node[above] {$\effbvec_{1}$}; 
\draw[red, ->] (0,0,0) -- (-0.4,0,0.4) node[above] {$\effbvec_{2}$}; 

\fill[red, opacity=0.1]  (-1.3/3,-1.3/3,2.6/3) -- (-1.3/3,2.6/3,-1.3/3) -- (2.6/3,-1.3/3,-1.3/3) -- cycle; 
\node at (-0.4,1.0) {$\tanplane$}; 

\fill[blue, opacity=0.3] (0,0,1) -- (0,1,0) -- (1,0,0) -- cycle; 
\node at (0.5,0.8) {$\intstrats$}; 
\node at (1.0,1.2) {$\R^{\nEffs+1}$}; 
\end{scope}

\draw[->, thick] (1.6,0.75) to[bend left] node[above] {$\incl$} (2.4,0.75);
\draw[<-, thick] (1.6,0.25) to[bend right] node[above] {$\chart$} (2.4,0.25);
\end{tikzpicture}
\caption{Maps \eqref{eq:simplex-parametrization} and \eqref{eq:simplex-chart} between the open corner of cube $\intcorcube$ in $\R^{\nEffs}$ and the open simplex $\intstrats$ in $\R^{\nEffs+1}$ for $\nEffs = 2$. In light red are the tangent spaces $\tspace\intcorcube = \R^{2}$ and $\tspace\intstrats = \tanplane$, where $\tanplane$ is the hyperplane $\tanplane = \setdef{(\strat_{0}, \strat_{1}, \strat_{2}) \in \R^{3}}{\strat_{0} + \strat_{1} + \strat_{2} = 0}$. The basis vectors $\effbvec_{1} = (1,0)$ and $\effbvec_{2} = (0,1)$ of $\tspace\intcorcube$ are mapped by \cref{eq:eff-rep-basis-vectors} to the vectors $\effbvec_{1} = \bvec_{1} - \bvec_{0} = (-1, 1, 0) $ and $\effbvec_{2} = \bvec_{2} - \bvec_{0} = (-1, 0, 1)$ in $\tanplane$.}
\label{fig:simplex-parametrization}
\end{figure}
Given the finite normal form game $\fingame = \fingamefull$ let $\nPures_{\play} \equiv \dimStrats_{\play}+1$ be the number of pure strategies of player $\play \in \players$, and denote the set of \aposs{} pure strategies as $\pures_{\play} = \setof{0_{\play}, 1_{\play}, \dots, \dimStrats_{\play}}$. Define $\effPures_{\play} \defeq \setof{ 1_{\play}, \dots, \dimStrats_{\play}}$; in the following the index $\pure_{\play} \in \pures_{\play}$ runs from $0_{\play}$ to $\dimStrats_{\play}$, and the index $\eff_{\play} \in \effPures_{\play}$ runs from $1_{\play}$ to $\dimStrats_{\play}$, unless otherwise specified.

Finite games in this form carry two intrinsic redundancies. First, $\nEffs_{\play}$ out of the $\nEffs_{\play}+1$ components of the mixed strategy $\strat_{\play} \in \strats_{\play}$ of player $\play$ are sufficient to completely specify it, since the remaining one is constrained by $\sum_{\pureplay}\strat_{\findex} = 1$. Second, two strategically equivalent games, albeit having different payoff functions, effectively represent the same game, since they display the same strategical and dynamical properties.\footnote{As discussed in \citet{CMOP11} strategically equivalent games have the same set of equilibria, but in general different efficiency (\eg Pareto optimality).} For this reason it is desiderable to introduce a \textit{reduced} or \textit{effective} representation of a game, in which \begin{enumerate*}
\item
the mixed strategy of each player is represented by an $\nEffs_{\play}$-dimensional object, and
\item strategically equivalent games are \quotes{clearly} the same, in a sense to be made precise.
\end{enumerate*}

To this end consider for each player the coordinates transformation between \aposs{} strategy space $\strats_{\play} = \simplex(\pures_{\play}) = \setdef{\strat_{\play} \in \clorthantplayalt }{\sum_{\pureplay \in \pures_{\play}}\strat_{\findex} = 1}$ and the \define{corner of cube simplex} $\corcube_{\play} = \setdef{\effstrat_{\play}\in\clorthantplayeff}{\sum_{\effplay \in \effPures_{\play} } \effstrat_{\eff} \leq 1}$ given by
\begin{equation}
\label{eq:simplex-parametrization}
\map{\incl}{\corcube_{\play}}{\strats_{\play}}{\effstrat_{\play}}{\strat_{\play}}
\quad \text{such that} \quad 
\begin{cases}
\strat_{\zindex} = 1 - \sum_{\eff_{\play} = 1}^{^{\nEffs_{\play}}} \effstrat_{\effindex} \\
\strat_{\effindex} = \effstrat_{\effindex}
\quad
\text{for all } \effplay \in \effrangeplay \eqdot
\end{cases}
\end{equation}
This map is visualized in \cref{fig:simplex-parametrization} with its obvious inverse\footnote{The maps $\incl$ and $\chart$ are labeled by $_0$ to denote the fact the the we express $\strat_{\zindex}$ as a function of the remaining coordinates; this choice comes without any loss of generality, \ie it would be equivalent to consider the map $\iota_{\pure}$ such that $\strat_{\play \pure_{\play}} = 1 - \sum _{\eff_{\play} \neq \pure_{\play} \effstrat_{\effindex}}$, and $\strat_{\effindex} = \effstrat_{\effindex}$ for all $\effplay \neq \pureplay$.}
\begin{equation}
\label{eq:simplex-chart}
\map{\chart}{\strats_{\play}}{\corcube_{\play}}{\strat_{\play}}{\effstrat_{\play}}
\quad \text{such that} \quad 
\effstrat_{\effindex} = \strat_{\effindex} \quad \text{for all } \effplay \in \effrangeplay \eqdot
\end{equation}
This standard reduction technique goes back at least to \citet[p.~4]{ritzbergerNashField1990}, and is employed
in many other works \citep{hofbauerEvolutionaryDynamicsBimatrix1996,MM10,LM13,LM15,CGM15,BM17,Wei95,HS98,TJ78}.

In the following we consider only \define{interior strategies} by restricting $\incl$ to $\incl|_{\intcorcube_{\play}}\from\intcorcube_{\play}\to\intstrats_{\play}$ (and we will denote $\incl|_{\intcorcube}$ just by $\incl$). \textit{Geometrically}, the reason to consider the relative interior of the strategy space is that $\intstrats_{\play}$ is a smooth manifold of dimension $\nEffs_{\play}$ with a global chart $\chart$ onto the open corner of cube $\intcorcube_{\play}$, which is an open subset of $\R^{\nEffs_{\play}}$; on the other hand, $\strats_{\play}$ is not a smooth manifold (\cf \cref{app:geometry}). For a \textit{dynamical} justification of the restriction to the interior of $\strats_{\play}$, \cf \cref{eq:intstrats} and the surrounding discussion.

Under the maps $\incl$ and its inverse $\chart$ the open corner of cube and the open simplex are fundamentally the same object; the corner of cube representation retains all the information existing on the simplex in a more efficient way, getting rid of the redundant degree of freedom. Thus, all the objects and structures defined on the open simplex $\intstrats_{\play}$ as a subspace of $\orthantplayalt$ \textendash\ such as payoff functions and payoff fields, vector fields and metrics \textendash\ must admit via \cref{eq:simplex-parametrization,eq:simplex-chart} an equivalent representation on the open corner of cube $\intcorcube_{\play}$, that we'll refer to interchangeably as \define{reduced} or \define{effective}. As opposed to effective, we will refer to objects defined on $\intstrats_{\play}$ as \define{full}.

In the next paragraphs we will present for each open simplex $\intstrats_{\play}$ and its corresponding open corner of cube $\intcorcube_{\play}$ the effective representation of payoff functions and payoff fields, \acl{RD}, tangent vectors, and metric tensors. The end result of this reduction procedure is the \define{effective representation} of the mixed extension of a finite game $\fingamefull$, in which all the relevant objects are define on (the interior of) the product corner of cube $\corcube = \prod_{\play \in \players} \corcube_{\play}$, rather than on the \quotes{redundant} original strategy space $\strats = \prod_{\play \in \players} \strats_{\play}$.

\subsection{Effective representation of payoff functions and payoff fields}
\begin{figure}
\centering
\begin{tikzpicture}[scale=1.5]

        \draw[blue, thick] (0,1) -- (1,0);
        \node at (0.6, 0.8) {$\intstrats_{\play}$};
        
        \draw[red, thick] (-1,1) -- (1,-1);
        \draw[red, dotted] (-1.2,1.2) -- (1.2,-1.2);
        \node at (-0.8, 0.4)  {$\tspace\intstrats_{\play}$};
        
        \draw[->] (-1.5,0) -- (1.5,0) node[right] {$\strat_{i,0}$};
        \draw[->] (0,-1.5) -- (0,1.5) node[above] {$\strat_{i,1}$};

\begin{scope}[shift={(4,0)}]

        \draw[blue, dashed] (0,0) -- (1,0);
        \draw[blue, dashed] (0,1) -- (1,1);
        \draw[blue, dashed] (0,0) -- (0,1);
        \draw[blue, dashed] (1,0) -- (1,1);
        \fill[blue, opacity=0.3] (0,0) -- (1,0) -- (1,1) -- (0,1) -- cycle;
        \node at (0.5, 0.7) {\small{$\intcorcube_{1} \times \intcorcube_{2}$}};
        
         \fill[red, opacity=0.1] (-1.2,-1.2) -- (1.2,-1.2) -- (1.2,1.2) -- (-1.2,1.2) -- cycle;
         \node at (-0.8, 1.3)  {$\tspace\of{\intcorcube_{1} \times \intcorcube_{2}}$};
        
        \draw[->] (-1.5,0) -- (1.5,0) node[right] {$\effstrat_{1}$};
        \draw[->] (0,-1.5) -- (0,1.5) node[above] {$\effstrat_{2}$};
        
        \end{scope}
    \end{tikzpicture}
\caption{A $(2 \times 2)$ game.~\textit{Left}: The strategy space of each player $\play \in \setof{1,2}$ in a $(2 \times 2)$ game is the $1$-dimensional open simplex $\intstrats_{\play}$ as a subspace of $\R^{2}$; the tangent space $\tspace\intstrats_{\play}$ is the line $\strat_{0} + \strat_{1} = 0$. \textit{Right}: The strategy space $ \intstrats_{1} \times\intstrats_{2}$ of a $2 \times 2$ game is a subset of $\R^{4}$, so we represent the open corner of cube $ \intcorcube =   \intcorcube_{1} \times \intcorcube_{2} = \setdef{\parens{\effstrat_{1}, \, \effstrat_{2}}}{ \effstrat_{1} > 0, \, \effstrat_{2} > 0, \, \effstrat_{1} < 1, \, \effstrat_{2} < 1 }$ as an open subset of $\R^{2}$. Its tangent space $\tspace\intcorcube$ is the whole $\R^{2}$.}
\label{fig:simplices}
\end{figure}
The effective representation of mixed strategies is given precisely by \cref{eq:simplex-chart}. Since payoff functions are scalar functions of these strategies, the effective representation $\effpay_{\play}\from\intcorcube_{\play}\to\R$ of the payoff function $\pay_{\play}\from\intstrats_{\play}\to\R$ is obtained as the restriction of $\pay_{\play}$ to $\intcorcube_{\play}$, \ie
\begin{equation}
\label{eq:eff-rep-pay}
\effpay_{\play}\of\effstrat = \pay_{\play}\of\strat
\end{equation}
for all $\play \in \players$, and all $\effstrat \in \effstrats, \strat \in \strats$ related by \cref{eq:simplex-parametrization}.

Just like the full payoff field is obtained differentiating the full payoff functions, the reduced payoff field is obtained differentiating the reduced payoff functions:\footnote{In more geometrical terms, $\effpay_{\play}$ (resp. $\effpayfield_{\play}$) is the \textit{pull-back} of $\pay_{\play}$ (resp. $\payfield_{\play}$) along $\incl$.}
\begin{equation}
\label{eq:eff-payfield-is-eff-individual-differential}
\payfield_{\findex}(\strat)
= \pay_{\play}(\pure_{\play};\strat_{-\play})
= \frac{\pd\pay_{\play}}{\pd\strat_{\findex}} \of\strat
\implies
\effpayfield_{\effindex}(\effstrat)
\defeq  \frac{\pd\effpay_{\play}}{\pd\effstrat_{\effindex}} \of\effstrat \eqdot
\end{equation}
\begin{remark}[Individual differential]
\label{rem:individual-differential}
\cref{eq:eff-payfield-is-eff-individual-differential} says that the components of the full (resp. reduced) payoff field  $\payfield_{\play}$ (resp. $\effpayfield_{\play}$) are obtained by partial differentiation of the payoff function $\pay_{\play}$ (resp. $\effpay_{\play}$) of player $\play \in \players$ with respect to \aposs{} mixed strategies $\strat_{\play}$ (resp. $\effstrat_{\play}$). As mentioned in \cref{rem:euclidean-vs-riemannian-gradient} we refer to the array of partial derivatives of a function as \define{differential} of the function; since we are differentiating each payoff function only with respect to the variables relative to one player, we say that the full (resp. reduced ) payoff field of a player is the \define{individual differential} of the full (resp. reduced) payoff function of the player, and for each $\play \in \players$ we write
\begin{equation}
\label{eq:individual-differential}
\payfield_{\play} = d_{\play} \pay_{\play} \quad \text{and} \quad \effpayfield_{\play} = \Eff{d}_{\play}\effpay_{\play} \eqdot
\end{equation}
\end{remark}
We have the following useful lemma to compute partial derivatives in effective coordinates:
\begin{lemma}
\label{eq:eff-chain-rule}Let $\fn\from\intstrats\to\R$ be a differentiable function and $\Eff{\fn}\from\intcorcube\to\R$ its effective representation. Then
\begin{equation}
\frac{\de }{ \de \effstrat_{\effindex} } \Eff{\fn}   \of \effstrat
= \ll \frac{\de}{ \strat_{\effindex} } - \frac{\de}{\strat_{\zindex} } \rr \fn \of \strat \eqcomma
\end{equation}
for all $\play \in \players, \eff_{\play} \in \effPures_{\play}$, and all $\strat \in \intstrats$ and $\effstrat \in \intcorcube$ related by \cref{eq:simplex-parametrization}.
\end{lemma}
\begin{proof}
Fix $\play \in \players$ and $\eff_{\play} \in \effPures_{\play}$. By the chain rule,
$
\frac{\de }{ \de \effstrat_{\effindex} } \Eff{\fn}   \of \effstrat
 =
 \sum_{\playalt \in \players}
 \sum_{\pure_{\playalt} = 0 _{\playalt}}^{\nEffs_{\playalt}}
 \frac{\de \strat_{\playalt \pure_{\playalt}}}{\de \effstrat_{\effindex}}
 \frac{\de}{\de{\strat_{\playalt \pure_{\playalt}}}} \fn\of\strat
$.
By \cref{eq:simplex-parametrization}, $\frac{\pd\strat_{\playalt 0_{\playalt}}}{\pd\effstrat_{\effindex}} = -\delta_{\play\playalt} $ and $\frac{\pd\strat_{\playalt\effalt_{\playalt}}}{\pd\effstrat_{\effindex}} = \delta_{\play\playalt}\delta_{\effplay\effalt_{\play}} $ for all $ \effalt_{\playalt} \in \setof{1_{\playalt}, \dots, \nEffs_{\playalt}}$, and we conclude expanding the sum and substituting.
\end{proof}
Applying the previous lemma to \cref{eq:eff-payfield-is-eff-individual-differential} we get the reduced expression of the payoff field:
\begin{equation}
\label{eq:eff-rep-payfield}
\effpayfield_{\effindex}(\effstrat) =  \payfield_{\effindex}(\strat) - \payfield_{\zindex}(\strat) \eqcomma
\end{equation}
with $\strat = \incl\of\effstrat$ for all $\effstrat \in \effstrats$ and all $\play \in \players, \, \eff_{\play}=1_{\play},\dotsc,\nEffs_{\play}$, in agreement with \cref{eq:payfield-eff} in the main text. The first order version of this equation gives an important relation between the Jacobian matrices of the full and reduced effective fields:
\begin{lemma}
\label{lemma:chain-derivative-eff-payfield}
The components of the Jacobian matrix of the effective payoff field are given by
\begin{equation}
\label{eq:chain-derivative-eff-payfield}
\frac{\de \v{\play}{\effalt}}{\dey{\playalt}{\eff}} \,  \of \effstrat
= \ll \frac{\de}{\dex{\playalt}{\eff}} - \frac{\de}{\dex{\playalt}{0}} \rr \,\ll \V{\play}{\effalt} - \V{\play}{0}\rr \of \strat
\end{equation}
with $\strat = \incl\of\effstrat$ for all $\effstrat \in \intcorcube$, $\play, \playalt \in \players$, $\effaltplay \in \effrangeplay$ and $\eff_\playalt \in \setof{1, \dots, \nEffs_{\playalt}}$ .
\end{lemma}
\begin{proof}
Immediate by \cref{eq:eff-chain-rule} and \cref{eq:eff-rep-payfield}.
\end{proof}
Next is a simple but important property of payoff fields:
\begin{lemma}
\label{lemma:payfield-independent-strat}
For every player $\play \in \players$ and for every pure strategy $\pureplay \in \pures_{i}$, the component $\payfield_{\findex}$ of the payoff field $\payfield$ does \textit{not} depend on the mixed strategy of player $\play$:
\begin{equation}
\frac{\de \V{\play}{\pure}}{\de \x{\play}{\purealt}}  \equiv 0
\end{equation}
for all players $\play \in \players$ and all $\pureplay, \purealtplay \in \pures_{\play}$. Analogously, for the reduced payoff field,
\begin{equation}
\frac{\de \v{\play}{\eff}}{\de \effstrat_{\effindexalt}}  \equiv 0
\end{equation}
for all players $\play \in \players$ and all $\effplay, \effaltplay \in \effPures_{\play}$.
\end{lemma}
\begin{proof}

The first statement is immediate by the fact that $\payfield_{\findex}\of\strat$ is the partial derivative with respect to $\strat_{\findex}$ of the multilinear function $\pay_{\play}(\strat)$; the second follows from \cref{lemma:chain-derivative-eff-payfield}.
\end{proof}
This property will be crucial in the proof of \cref{th:shah-divergence-repfield-result} in \cref{app:incompressible}.

\begin{example}[$2 \times 2$ game \textendash\ Effective payoff ]
\label{ex:22-game}
Consider a $2 \times 2$ game with $\pures_{1} = \pures_{2} = \setof{0,1}$ and $\effPures_{1} = \effPures_{2} = \setof{1}$.  The mixed strategies in full and effective representations are respectively
\begin{subequations}
\begin{align}
\strat & = \parens[\big]{
\parens{\stratEx{1}{0}, \stratEx{1}{1}} , 
\parens{\stratEx{2}{0}, \stratEx{2}{1}}} \in \strats_{1} \times \strats_{2} \\
\effstrat  &= \parens{\effstratEx{1}{1}, \effstratEx{2}{1}} \equiv (\effstrat_{1}, \effstrat_{2}) \in \corcube_{1} \times \corcube_{2} \eqcomma
\end{align}
\end{subequations}
where in the second line we drop an index since $\effPures_{i}$ is a singleton; \cf \cref{fig:simplices}. The full and effective payoff functions are
\begin{subequations}
\begin{align}
\pay_{\play}(\strat) 
	&= \stratEx{1}{0} \,  \stratEx{2}{0} \,  \pay_{\play}(0, 0) + \stratEx{1}{0} \,  \stratEx{2}{1} \,  \pay_{\play}(0, 1) + \stratEx{1}{1} \,  \stratEx{2}{0} \,  \pay_{\play}(1, 0) + \stratEx{1}{1} \,  \stratEx{2}{1} \,  \pay_{\play}(1, 1)
	\notag\\
\effpay_{\play}(\effstrat)
	&= \effstrat_{1}\effstrat_{2}\bracks[\big]{\pay_{\play}(0, 0) - \pay_{\play}(0, 1) - \pay_{\play}(1, 0) + \pay_{\play}(1, 1) }
		+ \effstrat_{1}\bracks[\big]{- \pay_{\play}(0, 0) + \pay_{\play}(1, 0)}
	\notag\\
	&\quad
		+ \effstrat_{2}\bracks[\big]{- \pay_{\play}(0, 0) + \pay_{\play}(0, 1)}
		+ \pay_{\play}(0, 0) \eqdot
\end{align}
\end{subequations}
Note that the full payoffs are polynomials of degree $2$ with each term of the same degree, while the reduced payoffs are polynomials of degree $2$ with terms of all possible degrees.

The two full payoff fields, with two components each, are
\begin{equation}
\payfield_{1}(\strat)  =
\begin{pmatrix}
\stratEx{2}{0} \, \pay_1(0, 0) + \stratEx{2}{1} \, \pay_1(0, 1)   \\ 
\stratEx{2}{0} \, \pay_1(1, 0) + \stratEx{2}{1} \, \pay_1(1, 1)
\end{pmatrix}
\quad \quad
\payfield_{2}(\strat)  =
\begin{pmatrix}
\stratEx{1}{0} \, \pay_2(0, 0) + \stratEx{1}{1} \, \pay_2(1, 0)  \\
\stratEx{1}{0} \, \pay_2(0, 1) + \stratEx{1}{1} \, \pay_2(1, 1)
\end{pmatrix} \eqcomma
\end{equation}
whereas the two reduced payoff fields, with one component each, are
\begin{subequations}
\begin{align}
\effpayfield_{1}(\effstrat)  = \effstrat_{2} \bracks[\big]{ \pay_1(0, 0) -  \pay_1(0, 1) -  \pay_1(1, 0) +  \pay_1(1, 1)} -  \pay_1(0, 0) +  \pay_1(1, 0)
\\
\effpayfield_{2}(\effstrat)  = \effstrat_{1} \bracks[\big]{ \pay_2(0, 0) -  \pay_2(0, 1) -  \pay_2(1, 0) +  \pay_2(1, 1)} -  \pay_2(0, 0) +  \pay_2(0, 1)\eqdot
\end{align}
\end{subequations}
Note that $\payfield_{\play}\of\strat$ does \textit{not} depend on $\strat_{\play}$, and $\effpayfield_{\play}\of\effstrat$ does \textit{not} depend on $\effstrat_{\play}$, as expected by \cref{lemma:payfield-independent-strat}.
\endenv
\end{example}

\paragraph{Effective payoff field and strategical equivalence} The expression for the effective payoff field can be used to show that two games are strategically equivalent if and only if they are described by the same effective payoff field.  Before making this precise, we give here a simple but powerful lemma that we will use often in the following:
\begin{lemma}[Vanishing of multilinear extension]
\label{lemma:multilinear-extension-zero} Given a finite game $\fingame = \fingamefull$ let $F\from\pures\to\R$ be a real function of pure strategy profiles and $\bar{F}\from\strats\to\R$ its multilinear extension, \ie
\begin{equation}
\bar{F}(\strat)
= \exwrt{\pure\sim\strat}{F(\pure)}
= \sum_{\pure \in \pures} F(\pure) \prod_{\play \in \players}\strat_{\findex}
\equiv  \sum_{\pure \in \pures} F(\pure) \,  \strat_{\pure}
\end{equation}
Then $F(\pure) = 0$ for all $\pure \in \pures$ if and only if $\bar{F}(\strat) = 0$ for all $\strat \in \strats$.
\end{lemma}
\begin{proof}
If $F(\pure) = 0$ for all $\pure \in \pures$ the conclusion is immediate. Conversely, assume that $\bar{F}(\strat) = 0$ for all $\strat \in \strats$. In particular, for any $\pure \in \pures$, this holds true for the mixed strategy $\strat_{\findexalt} = \delta_{\pureplay \purealtplay}$ in which each player $\play$ assigns all weight to $\pureplay$, \ie $0 = \bar{F}(x) = \sum_{\purealt \in \pures} F(\purealt) \prod_{\play \in \players} \delta_{\pureplay \purealtplay} = F(\pure)$.
\end{proof}

We now move on to show that two games are strategically equivalent if and only if they are described by the same effective payoff field. Recall from \cref{eq:strat-equiv} in the main text that two finite games $\fingamefull$ and $\fingamefullalt$ are strategically equivalent if
\begin{equation}
\tag{\ref{eq:strat-equiv}}
\alt\pay_{\play}(\purealt_{\play};\pure_{-\play}) - \alt\pay_{\play}(\pure_{\play};\pure_{-\play})
	= \pay_{\play}(\purealt_{\play};\pure_{-\play}) - \pay_{\play}(\pure_{\play};\pure_{-\play})
\end{equation}
for all $\play \in \players$ and all $\pure, \purealt \in \pures$. If two games $\fingame$ and $\alt\fingame$ are strategically equivalent, we write $\fingame \sim \alt\fingame$.

\begin{proposition}
\label{prop:strategical-equivalence}
Two finite games are strategically equivalent if and only if they have the same effective payoff field.
\end{proposition}
The proof of this result is better broken into steps.

Firstly, following \citet{CMOP11} we give the following definition and lemma:
\begin{definition}[Non-strategic game]
A finite normal form game $\fingamefullns$ is called \define{non-strategic} if all players are indifferent between all of their choices:
    \begin{equation}
    \label{eq:non-strategic-game}
        \payns_\play(\purealt_\play, \pure_{\others}) =  \payns_\play(\pure_\play, \pure_{\others})
    \end{equation}
        for all $\play \in \players$, all $\pure_{\others} \in \pures_{\others}$, and all $\pure_\play, \purealt_\play \in \pures_\play$.
\end{definition}
\begin{lemma}
\label{lemma:strat-equiv-iff-difference-non-strategic}
Two finite games $\fingamefull, \fingamefullalt$ are strategically equivalent if and only if their difference is a non-strategic game.
\end{lemma}
\begin{proof}
Let $\fingame - \alt\fingame$ be non-strategic; then $\payns \defeq \alt\pay - \pay$ fulfills \cref{eq:non-strategic-game}, and rearranging the terms we immediately get \cref{eq:strat-equiv}. Conversely let $\fingame$ and $\alt\fingame$ be strategically equivalent, and set $\payns \defeq \alt\pay - \pay$; again rearrange the terms in \cref{eq:strat-equiv} to immediately conclude that $\payns$ fulfills \cref{eq:non-strategic-game}.
\end{proof}
Secondly, we give the following characterization of non-strategic games:
\begin{proposition}
\label{prop:non-strategic-iff-vanishing-eff-payfield}
A finite game $\fingamefullns$ is non-strategic if and only if its effective payoff field $\effpayfield$ vanishes identically.
\end{proposition}
\begin{proof}
Let $\fingame$ be non-strategic. Then its effective payoff field fulfills
\begin{equation}
\begin{split}
\effpayfield_{\effindex} \of\effstrat
	& = \payfield_{\effindex} \of\strat - \payfield_{\zindex} \of\strat
	= \payns_{\play} (\effplay, \strat_{\others}) - \payns_{\play} (0_{\play}, \strat_{\others}) \\
	& = \sum_{\pure_{\others}} \strat_{\findexothers} \, \bracks[\big]{ \payns_{\play} (\effplay, \pure_{\others}) - \payns_{\play} (0_{\play}, \pure_{\others})} 
	= 0
\end{split}
\end{equation}
for all $\effstrat \in \corcube$, all $\play \in \players$, and all $\effplay \in \effrangeplay$, where the last equality holds by definition of non-strategic game.

Conversely, let the effective payoff field $\effpayfield\of\effstrat$ of $\fingame$ be identically zero for all $\effstrat \in \corcube$. Then by \cref{eq:eff-rep-payfield} all the components $\payfield_{\findex}(\strat) = \payns_{\play}(\pureplay, \strat_{\others})$ of the full payoff field are equal to each other, \ie $\payns_{\play}(\pureplay, \strat_{\others}) = \payns_{\play}(\purealtplay, \strat_{\others})$ for all $\strat \in \strats$ and $\pureplay, \purealtplay \in \pures_{\play}$, which in turn implies that $\payns_\play(\pure_\play, \pure_{\others}) =  \payns_\play(\purealt_\play, \pure_{\others})$ by \cref{lemma:multilinear-extension-zero}. 
\end{proof}
Finally, by \cref{lemma:strat-equiv-iff-difference-non-strategic}, \cref{prop:non-strategic-iff-vanishing-eff-payfield} is equivalent to \cref{prop:strategical-equivalence}:
\begin{proof}[Proof of \cref{prop:strategical-equivalence}]
Given two finite games $\fingamefull$, $\fingamefullalt$ we have the following implications:
\begin{equation}
\fingame \sim \alt\fingame \iff \fingame - \alt\fingame \text{ is non-strategic } \iff \effpayfield - \alt\effpayfield = 0 \iff \effpayfield = \alt\effpayfield 
\end{equation}
which concludes the proof.
\end{proof}
\subsection{Effective representation of the replicator dynamics}

We report here for ease of reference the full \acl{RD} \cref{eq:RD,eq:repfield} and its effective representation \cref{eq:RD-eff,eq:repfield-eff}, already derived in the main text:
\begin{equation}
\tag{\ref{eq:repfield}}
\dot{\strat}_{\findex} = 
\repfield_{\play\pure_{\play}}(\strat)
= \strat_{\play\pure_{\play}}
		\bracks*{
			\payfield_{\play\pure_{\play}}(\strat)
			- \insum_{\purealt_{\play} = 0}^{\nEffs_{\play}}
				\strat_{\play\purealt_{\play}} \payfield_{\play\purealt_{\play}}(\strat)} \quad \text{for all } \play \in \players \text{ and } \pureplay \in \rangeplay \eqcomma
\end{equation}
\begin{equation}
\tag{\ref{eq:repfield-eff}}
\dot{\effstrat}_{\effindex} = 
\repfield_{\play\eff_{\play}}(\effstrat)
	= \effstrat_{\play\eff_{\play}}
		\bracks*{\effpayfield_{\play\eff_{\play}}(\effstrat)
			- \insum_{\effalt_{\play}=1}^{\nEffs_{\play}} \effstrat_{\play\effalt_{\play}} \effpayfield_{\play\effalt_{\play}}(\effstrat)}  \quad \text{for all } \play \in \players \text{ and } \effplay \in \effrangeplay \eqdot
\end{equation}
\begin{remark}
\label{rem:notation-eff-repfield}
We prefer the notation $\repfield_{\effindex}$ over the more correct $\effrepfield_{\effindex}$ for the \textit{effective} replicator vector field for notational simplicity when there is no risk of ambiguity, and reinstate the tilde $\tilde{\cdot}$ when we want to stress the difference between the full and effective representations of the replicator field.
\endenv
\end{remark}

\subsection{Effective representation of tangent vectors}
As we saw, the map $\incl$ can be used to obtain the effective representation of \textit{points} in the strategy space (\ie mixed strategies) and of \textit{scalar functions} of these points (\ie functions of mixed strategies); the effective representation of \textit{vectors} that are tangent to the strategy space requires more care.

The open corner of cube $\intcorcube_{\play}$ is an open subset of $\orthantplayeff$ in its own right, so its tangent space is the whole euclidean space, that we denote by $\tspace\intcorcube_{\play} \equiv \R^{\nEffs_{\play}}$. On the other hand $\intstrats_{\play}$ is an open subset of an affine hyperplane in $\orthantplayalt$, and its tangent space is give by the hyperplane $\tanplane_{\play}$, the linear subspace in $\R^{\nEffs_{\play} + 1}$ of vectors whose components add up to zero, \cf \cref{lemma:tangent-space-simplex} and \cref{fig:simplex-parametrization}.

A basis vector $\effbvec_{\effindex}$ of $\R^{\nEffs_{\play}}$ must correspond via $\incl$ to a vector tangent to the simplex, \ie a vector in $\tanplane_{\play}$, that we want to determine. Since $\tanplane_{\play}$ is a linear subspace of $\R^{\nEffs_{\play}+1}$, it must be possible to express via $\incl$ this sought after vector as a linear combination of basis vectors $\setof{\bvec_{\findex}}_{\pureplay \in \rangeplay}$ of $\R^{\nEffs_{\play}+1}$, for all $\effplay \in \effrangeplay$.

The way to obtain this identification comes from an important result in differential geometry~\citep{leeIntroductionSmoothManifolds2012}: given a smooth map between two spaces such as $\incl\from\intcorcube_{\play}\to\intstrats_{\play}$, its differential induces a linear map $\d\incl\from\R^{\nEffs_{\play}}\to\tanplane_{\play}$ between the tangent spaces to the two spaces. The matrix representing this differential is the Jacobian matrix $\jac$ of $\incl$, so (dropping temporarily the player index $\play$ for notational simplicity) a basis vector $\effbvec_{\eff} \equiv \effbvec_{\effindex}$ of $\R^{\nEffs} \equiv \R^{\nEffs_{\play}}$ is mapped to the vector $\d\incl\of{\effbvec_{\eff}} \in \tanplane \subset \R^{\nEffs + 1}$ of component $\bracks*{\d\incl\of{\effbvec_{\eff}}}_{\pure} \equiv \bracks*{\d\incl\of{\effbvec_{\effindex}}}_{\findex}$ given by
\begin{equation}
\bracks*{\d\incl\of{\effbvec_{\eff}}}_{\pure}  = \sum_{\effalt = 1}^{\nEffs} \jac_{\pure \effalt} \bracks{\effbvec_{\eff}}_{\effalt} 
\end{equation}
for all $\pure \in \setof{0, \dots, \nEffs}$. Since $\effbvec_{\eff}$ is a basis vector its $\effalt$-th component is given by the Kronecker delta $\delta_{\eff\effalt}$, so
\begin{equation}
\begin{split}
\d\incl\of{\effbvec_{\eff}} 
& = \sum_{\pure = 0}^{\nEffs} \bracks*{\d\incl\of{\effbvec_{\eff}}}_{\pure} \, \bvec_{\pure}
= \sum_{\pure = 0}^{\nEffs} \jac_{\pure \eff} \, \bvec_{\pure}
= \sum_{\pure = 0}^{\nEffs} \frac{\pd\strat_{\pure}}{\pd\effstrat_{\eff}} \, \bvec_{\pure}
\end{split}
\end{equation}
for all $\eff \in \setof{1, \dots, \nEffs}$. Again by \cref{eq:simplex-parametrization} we have $\frac{\pd\strat_{0}}{\pd\effstrat_{\eff}} = -1 $ and $\frac{\pd\strat_{\effalt}}{\pd\effstrat_{\eff}} = \delta_{\eff\effalt} $ for all $\eff, \effalt \in \setof{1, \dots, \nEffs}$, so after reinserting the player index we get $\d\incl\of{\effbvec_{\rindex}} = \bvec_{\rindex} - \bvec_{\zindex}$. For notational simplicity in the following we drop the differential of the $\incl$ map and denote $\d\incl\of{\effbvec_{\rindex}}$ just by $\effbvec_{\rindex}$, so that in conclusion
\begin{equation}
\label{eq:eff-rep-basis-vectors}
\effbvec_{\rindex} = \bvec_{\rindex} - \bvec_{\zindex}
\end{equation}
for all $\play \in \players$ and all $\effplay \in \effrangeplay$; \cf \cref{fig:simplex-parametrization} for a  visual example.
\subsection{\shah metric}
\label{sec:eff-rep-metric}

As discussed in \cref{sec:decomposition}, the Euclidean metric is not attuned with the dynamical properties of \acl{RD}. For this reason \citet{Sha79} introduced a metric that \quotes{\textit{[...] turns out to be surprisingly effective in clarifying the dynamics of the [replicator dynamical] system}}. In this section we present this metric in its full and effective representations, along with some of its geometrical properties.

A remark on notation: we include where needed the player index $\play \in \players$ for ease of comparison with the other sections of this work; all expressions hold true with exactly the same form if it is omitted.
\begin{definition}
\label{def:sha-i-orhant}
The \shah metric on the positive orthant $\orthantplayalt$ is the smoothly varying positive definite symmetric bilinear form $\gmat_{\strat_{\play}} \from \orthantplayalt \times \orthantplayalt \to \R $ represented by the $(\nEffs_{\play} + 1) \times (\nEffs_{\play} + 1)$ matrix\footnote{Recall that $\tspace_{\strat}\orthantplayalt \cong \orthantplayalt $ by \cref{rem:tangent-space-euclidean-manifold}.}
\begin{equation}
\label{eq:sha-i-orhant}
\gmat_{\shindex}(\strat_{\play})
	\defeq \frac{\delta_{\pureplay \purealtplay }}{ \strat_{\findex}}
\end{equation}
for all $\strat_{\play}\in \orthantplay$ and $\pureplay, \purealtplay \in \rangeplay$.
\end{definition}
\paragraph{Effective \shah metric}
The components of the effective metric tensors $\effgmat_{\effstrat_\play}$ on $\intcorcube_{\play}$ are obtained by \cref{eq:gmat} as the inner product between effective tangent vectors, that by \cref{eq:eff-rep-basis-vectors} is
\begin{align}
\effgmat_{\effshindex}\of\effstrat
	&= \inner{\effbvec_{\effindex}}{\effbvec_{\effindexalt}}_{\effstrat}
	= \inner{\bvec_{\effindex} - \bvec_{\zindex}}{\bvec_{\effindexalt} - \bvec_{\zindex}}_{\strat}
	\notag\\
	&= \sum_{\pureplay\purealtplay} \frac{\delta_{\pureplay\purealtplay}}{\strat_{\findex}} \, (\delta_{\pureplay \effplay} - \delta_{\pureplay 0_{\play}})(\delta_{\purealtplay \effaltplay} - \delta_{\purealtplay 0_{\play}})
	\notag\\
	&= \frac{\delta_{\effplay\effaltplay}}{\strat_{\effindex}} - \frac{\delta_{\effindex 0_{\play}}}{\strat_{\effindex}} - \frac{\delta_{\effindexalt 0_{\play}}}{\strat_{\effindexalt}} + \frac{1}{\strat_{\zindex}}
	\eqdot
\end{align}
The second and third terms vanish, so in conclusion
\begin{equation}
\label{eq:eff-rep-shah}
 \effgmat_{\effshindex}\of{\effstrat_{\play}} =  \frac{\delta_{\effplay\effaltplay}}{\effstrat_{\effindex}}  + \frac{1}{1-\sum_{\effaltalt_{\play} = 1}^{\nEffs_{\play}}\effstrat_{\play\effaltalt_{\play}}}
\end{equation}
for all $\play \in \players$, all $\effplay, \effaltplay \in \effrangeplay$, and all $\effstrat_{\play} \in \intcorcube_{\play}$, in agreement with \cref{eq:Shah-eff}.
\paragraph{Determinant of the \shah metric}
Since the matrix representing the full \shah metric is diagonal its determinant is immediately given by
\begin{equation}
\label{eq:det-sha-full}
\det{\gmat_{\play}}(\strat_{\play}) = \frac{1}{\prod_{\pureplay = 0}^{\nEffs_{\play}}\strat_{\findex}} \eqdot
\end{equation}
The determinant of the \shah metric $\effgmat_{\play}$ in its effective representation then follows after a standard calculation based on the matrix determinant lemma, \viz
\begin{equation}
\label{eq:det-sha-eff}
\det{\effgmat}_{\play}(\effstrat_{\play}) 
= \frac{1}{\left(1-\sum_{\effplay = 1}^{\nEffs_{\play}} \effstrat_{\effindex}\right)\prod_{\effaltplay = 1}^{\nEffs_{\play}} \effstrat_{\effindexalt}}
= \frac{1}{ \strat_{\zindex} \prod_{\effaltplay = 1}^{\nEffs_{\play}} \effstrat_{\effindexalt}} \eqdot
\end{equation}
The fact that the determinant of the full and reduced metric formally agree is a particularity of the \shah metric, and is in general \textit{not} true for Riemannian metrics.

\paragraph{\shah unitary spheres}

As discussed in \cref{sec:decomposition} in the main text, the \shah unit sphere $\sphere_{\strat_{\play}} \defeq \setdef{\tvec\in \orthantplayalt }{\gmat_{\strat_{\play}}(\tvec, \tvec) = 1}$ at $\strat_{\play} \in \orthantplayalt$ becomes increasingly flattened along the $\strat_{\findex}$-axis as $\strat_{\findex}\to 0$, as depicted in \cref{fig:balls}. Indeed (omitting for a second the player index $\play$) the \shah unit sphere at $\strat$,
\begin{equation}
\label{eq:sha-ball}
\Big\{\tvec \in \R_{++}^{\nEffs+1} \text{ such that } 
\gmat_{\strat}(\tvec, \tvec)
= \sum_{\pure \purealt} \frac{\delta_{\pure\purealt}}{\strat_{\pure}} \tvec_{\pure} \tvec_{\purealt}
= \sum_{\pure} \frac{\tvec_{\pure}^{2}}{\strat_{\pure}} = 1 \Big\} \eqcomma
\end{equation}
is an hyper-ellipse with the size of the $\pure$-th axis going to zero as $\strat_{\pure} \to 0$.

\begin{remark}
\label{rem:sha-boundary-1}
As discussed in \cref{ex:RD-22-logistic}, the behavior of the \shah metric as the boundary is approached, responsible for the shrinking of unit spheres describe above, is also the key feature that confines the \acl{RD} to the interior of the strategy space.
\endenv
\end{remark}

This concludes our (by no means exhaustive) treatment of the geometrical properties of the \shah metric; for an in-depth treatment we refer the reader to \citet{akinGeometryPopulationGenetics1979} and \citet{Sha79}. In \cref{app:replicator-geometry} we turn at some of its \textit{dynamical} properties and at its deep connection with the \acl{RD}; on this matter see also
\citep{HS98,
mertikopoulosRiemannianGameDynamics2018,
harperInformationGeometryEvolutionary2009,
LM15,
mladenovicGeneralizedNaturalGradient2021}.

%% file: Appendix/App-RD.tex
In this appendix we discuss the relation between the \shah metric and the \acl{RD}.

Given a finite normal form game $\fingamefull$, the evolution of the players' mixed strategies $\strat_{\play} \in \strats_{\play} = \simplex(\pures_{\play}) = \setdef{\strat_{\play} \in \clorthantplay }{\sum_{\pure_{\play} \in \pures_{\play}}\strat_{\findex} = 1}$ under the exponential weights learning scheme evolves according to the replicator dynamical system, that is
\begin{equation}
\tag{\ref{eq:RD}}
\dot\strat_{\play\pure_{\play}}
	= \strat_{\play\pure_{\play}}
		\bracks{\pay_{\play}(\pure_{\play};\strat_{-\play}) - \pay_{\play}(\strat)}
		= \strat_{\play\pure_{\play}}
		\bracks*{
			\payfield_{\play\pure_{\play}}(\strat)
			- \insum_{\purealt_{\play}\in\pures_{\play}}
				\strat_{\play\purealt_{\play}} \payfield_{\play\purealt_{\play}}(\strat)} 
\end{equation}
for all $\play \in \players$ and $\pureplay \in \pures_{\play}$.
We define $\RD_{\findex} (\strat) 
	\defeq \strat_{\play\pure_{\play}}
		\bracks*{
			\payfield_{\play\pure_{\play}}(\strat)
			- \insum_{\purealt_{\play}\in\pures_{\play}}
				\strat_{\play\purealt_{\play}} \payfield_{\play\purealt_{\play}}(\strat)}$ as in \cref{eq:repfield} in the main text, and  write the replicator system more compactly as\footnote{The notation $\RD$ is reminiscent of the \define{sharp} operator, of common use in differential geometry~\citep{leeIntroductionSmoothManifolds2012}. This is no coincidence: the payoff field $\payfield$ can be seen as a \define{dual vector field} or \define{$1$-form} in $\R^{\pures}$, and \acl{RD} follow the flow lines of the the primal-vector field obtained as the \shah sharp of the reduced payoff dual-vector field, as briefly discussed in \cref{app:geometric-tour}. This primal-dual interpretation is discussed in detail in \cite{mertikopoulosRiemannianGameDynamics2018}.}
\begin{equation}
\label{eq:RD-dynamical-system}
\dot{\strat}_{\play} = \RD_{\play}(\strat) \eqdot
\end{equation}
\paragraph{Interior and parallelism}
For~\eqref{eq:RD-dynamical-system} to make sense $\RD_{\play}(\strat)$ must point in a direction parallel to $\strats_{\play}$ for all $\strat$ along the trajectory. The notion of \quotes{parallelism} breaks down at the boundary of $\strats_{\play}$, but for each player $\play$ the interior of $\strats_{\play}$ is invariant under $\RD_{\play}$\footnote{If $\strat_{\findex}(t_{0}) = 0$ then $\strat_{\findex}(t) = 0 $ for all $t$ since $\dot{\strat}_{\findex} \propto \strat_{\findex}$, \cf \citet{HS98}. As shown in \cref{ex:RD-22-logistic}, this key dynamical feature bounding the \acl{RD} to the interior of the strategy space, is a consequence of the functional form of the \shah metric.}. So by restricting our attention to dynamics with initial conditions $\strat_{\play}(t_{0})$ in the \define{open mixed strategies space},
\begin{equation}
\label{eq:intstrats}
\intstrats_{\play} = \setdef{\strat_{\play} \in \orthantplay }{\sum_{\pure_{\play} \in \pures_{\play}}\strat_{\findex} = 1} \eqcomma
\end{equation}
we are sure that $\strat_{\play}(t) \in \intstrats_{\play}$ for all times $t$ and all players, avoiding boundary issues. This means that under \eqref{eq:RD} each pure strategy of each player has a non-zero probability to be played at all times, \ie $\strat_{\findex} \neq 0$ for all $\play \in \players$ and all $\pure_{\play} \in \pures_{\play}$.

Now unburdened from boundary issues we can give a precise notion of parallelism: a vector is parallel to $\intstrats_{\play}$ if its components sum to zero.
\begin{lemma}[Tangent space to open simplex]
\label{lemma:tangent-space-simplex}
The tangent space to the open simplex $\intstrats_{\play} \subset \orthantplay$ for any $\strat_{\play} \in \intstrats_{\play}$ is the hyperplane in $\R^{\pures_{\play}}$ of vectors whose components sum up to zero: 
\begin{equation}
\label{eq:tangent-space-simplex}
\tanplane_{\play} \defeq \tspace_{\strat_{\play}}\intstrats_{\play} = \setdef{\tvec_{\play} \in \R^{\pures_{\play}}}{\sum_{\pureplay}\tvec_{\findex} = 0} \eqdot
\end{equation}
\end{lemma}
\begin{proof} The open simplex $\intstrats_{\play}$ is the level set of value $1$ of the smooth function $\sub_{\play}\from\orthant[\R]^{\pures_{\play}}\to\R$, $\sub_{\play}(\strat_{\play}) = \sum_{\pureplay} \strat_{\findex}$. The differential $\d{\sub_{\play}} = \1_{\play} \defeq (1, \dots, 1)$ of $\sub_{\play}$ is a surjective linear map from $\R^{\pures_{\play}}$ to $\R$ that does not depend on $\strat_{\play}$, so by a standard theorem~\citep[Prop.~5.38]{leeIntroductionSmoothManifolds2012} the tangent space to $\intstrats_{\play}$ at any point is the kernel of $\d\sub_{\play}$, that is $\setdef{\tvec_{\play} \in \R^{\pures_{\play}}}{\dualp{\1_{\play}}{\tvec_{\play}}=0}$.
\end{proof}
As a sanity check note that $\RD_{\play}$ evaluated at any $\strat \in \strats$ is parallel to $\strats_{\play}$, since
$
\sum_{\pure_{\play}} \RD_{\findex}(x) =  \sum_{\pure_{\play}}  \strat_{\play\pure_{\play}}
		\bracks{\pay_{\play}(\pure_{\play};\strat_{-\play}) - \pay_{\play}(\strat)} = 0
$
for all $\play \in \players$ and all $\strat \in \strats$.
\subsection{Full \shah metric and individual gradient}
As discussed in \cref{app:reduction}, the objects defining a game and a learning dynamics admit a \textit{full}, redundant representation; and an \textit{effective} one. In this section we use objects in the full representation to prove \cref{prop:rep-Shah} from the main text, stating that that \acl{RD} are equivalent to the steepest individual payoff ascent dynamics with respect to the \shah metric. In the next section, making use of the reduced representation, we present an alternative and more concise proof of the same result.

Given a finite game in normal form $\fingamefull$ let the mixed strategy of each player evolve according to \eqref{eq:RD}. For each player $\play \in \players$ endow the positive orthant $\orthantplay$ with the \shah metric given by \cref{eq:Shah} in the main text, and discussed in further detail in \cref{sec:eff-rep-metric}.
Using this metric we can give the following
\begin{definition}
\label{def:individual-gradient}
The \define{individual payoff gradient} $\grad_{\play}\pay_{\play}$ of the payoff function
$\pay_{\play}$
is the vector field on $\orthantplay$ that is
\begin{enumerate*}[(\itshape a\upshape)]
\item
\label{item:gradient-is-parallel}
parallel to $\intstrats_{\play}$, \ie
$\sum_{\pure_{\play}} \bracks{\grad_{\play}\pay_{\play}(\strat)}_{\pure_{\play}} = 0 $
for all $\strat \in \intstrats$; and
\item
\label{item:gradient-on-parallel}
fulfills
\end{enumerate*}
\begin{equation}
\label{eq:individual-gradient-directional-derivative}
\tag{\ref{eq:paygrad-Shah}}
\inner{\grad_{\play}\pay_{\play}(\strat)}{\tvec_{\play}}
	= \dir\pay_{\play}(\strat;\tvec_{\play})
\end{equation}
for all $\play \in \players$, $\strat \in \intstrats$, and $\tvec_{\play} \in \R^{\pures_{\play}}$ that are tangent to $\intstrats_{\play}$ (that is, $\sum_{\pure_{\play}\in\pures_{\play}} \tvec_{\play\pure_{\play}} = 0$). 
\end{definition}
\begin{remark}
\label{rem:gradient-gauge-fixing}
The definition is well-posed. From the material in \cref{app:geometry}, we we can use a Riemannian metric on $\orthantplay$ to define the gradient of a function by specifying the value of its inner product at all points with \textit{all} vectors $\tvec_{\play} \in \R^{\pures_{\play}}$. Condition \ref{item:gradient-on-parallel} specifies this value for vectors in the hyperplane \textit{parallel} to the simplex, leaving a degree of freedom to be specified -  namely, the value of the inner product between the gradient and vectors that are \textit{normal} to the simplex. Condition \ref{item:gradient-is-parallel} fixes this gauge by requiring the gradient itself to be parallel to the simplex, thus giving zero inner product with normal vectors. This gauge-fixing procedure will be crucial in the proof of \cref{prop:rep-Shah}. 
\endenv
\end{remark}
We are now in position to prove \cref{prop:rep-Shah} from the main body of the article, that we restate here for ease of reference:
\RepShah*
\begin{proof} Write $\gfield_{\play} \defeq \grad_{\play}\pay_{\play}$; we have to show that
\begin{equation}
\gfield_{\findex}\of\strat = \strat_{\play\pure_{\play}}
		\bracks{\pay_{\play}(\pure_{\play};\strat_{-\play}) - \pay_{\play}(\strat)}
\end{equation}
for all $\strat \in \intstrats, \play \in \players$, and $\pure_{\play} \in \pures_{\play}$. Let $\tvec_{\play} \in \tspace_{\strat_{\play}}\intstrats_{\play}$ be a tangent vector; by \cref{eq:sha-i-orhant} its \shah inner product with $\gfield_{\play}\of\strat$ is
\begin{equation}
\inner{ \gfield_{\play}\of\strat}{ \tvec_{\play} }
= \sum_{\pure_{\play} \purealt_{\play}} \frac{\delta_{\pure_{\play}\purealt_{\play}}}{\strat_{\findex}} \, \gfield_{\findex} \tvec_{\findexalt}
= \sum_{\pure_{\play}} \frac{\gfield_{\findex}}{\strat_{\findex}} \, \tvec_{\findex}
\end{equation}
for all $\strat \in \intstrats$ (note in particular that $\strat_{\findex} \neq 0$). By condition \ref{item:gradient-on-parallel} in the individual payoff gradient's definition, this inner product is also equal to
$
\inner{ \gfield_{\play}\of\strat}{ \tvec_{\play} }
= \pd\pay_{\play}\of{\strat, \tvec_{\play}}
= \sum_{\pure_{\play}} \tvec_{\findex} \frac{\pd{\pay_{\play}}\of\strat}{\pd\strat_{\findex}}
$; equating the two expressions and rearranging the sums one gets
\begin{equation}
\label{eq:dummy}
\tag{$\ast$}
\sum_{\pure_{\play}} \tvec_{\findex} \parens*{ \frac{\gfield_{\findex}\of\strat}{\strat_{\findex}} -  \frac{\pd{\pay_{\play}}\of\strat}{\pd\strat_{\findex}} } = 0 \eqdot
\end{equation}
Denote the term in brackets by
$
\dummyB_{\findex}\of\strat \defeq \frac{\gfield_{\findex}\of\strat}{\strat_{\findex}} -  \frac{\pd{\pay_{\play}}\of\strat}{\pd\strat_{\findex}}
$
and let $\dummyB_{\play}\of\strat = (\dummyB_{\findex}\of\strat)_{\pure_{\play} \in \pures_{\play}}$; \cref{eq:dummy} then reads $\tvec_{\play} \cdot \dummyB_{\play}\of\strat = 0$.

If we required this to hold for all $\tvec_{\play} \in \R^{\pures_{\play}}$  then $\dummyB_{\play}\of\strat$ should vanish identically; but since
$\tvec_{\play} \in \tspace_{\strat_{\play}}\intstrats_{\play}$ it follows that $\dummyB_{\play}\of\strat$ must belong to the annihilator of $\tspace_{\strat_{\play}}\intstrats_{\play}$, that is $\Span{\1_{\play}}$. In other words, the components of $\dummyB_{\play}\of\strat$ must all be the same:
\begin{equation}
\dummyB_{\findex}\of\strat = \dummyB_{\findexalt}\of\strat = \const_{\play}\of\strat \quad \forall \pure_{\play}, \purealt_{\play} \in \pures_{\play}, \quad \forall \strat \in \intstrats \eqcomma
\end{equation}
for some function $\const_{\play}\from\intstrats\to\R$.

We can get to this result more explicitly also by expanding the sum in \eqref{eq:dummy} and eliminating one of the components of $\tvec_{\play}$. Let $\dimStrats_{\play}+1$ be the number of pure strategies of player $\play$, and denote the set of \aposs{} pure strategies by $\pures_{\play} = \setof{0_{\play}, 1_{\play}, \dots, \dimStrats_{\play}}$. Then letting the index $\pure_{\play}$ run from $0_{\play}$ to $\dimStrats_{\play}$ and the index $\eff_{\play}$ run from $1_{\play}$ to $\dimStrats_{\play}$ we get
\begin{equation}
\begin{split}
0 & = \sum_{\pure_{\play}} \tvec_{\findex} \dummyB_{\findex}\of\strat =  \tvec_{\zindex} \dummyB_{\zindex}\of\strat + \sum_{\eff_{\play}} \tvec_{\rindex} \dummyB_{\rindex}\of\strat = \\
& -\sum_{\eff_{\play}}\tvec_{\rindex} \dummyB_{\zindex}\of\strat + \sum_{\eff_{\play}} \tvec_{\rindex} \dummyB_{\rindex}\of\strat = \sum_{\eff_{\play}} \tvec_{\rindex} \parens*{ \dummyB_{\rindex}\of\strat - \dummyB_{\zindex}\of\strat } 
\eqdot\end{split} 
\end{equation}
This time the $\tvec_{\rindex}$-s are $\dimStrats_{\play}$ unconstrained numbers, so the term in bracket must vanish identically, and we recover the sought after result that the components of $\dummyB_{\play}\of\strat$ must all be the same. Plugging this fact in the definition of $\dummyB_{\findex}$ and solving for $\gfield_{\findex}$ we get
\begin{equation}
\gfield_{\findex}\of\strat = \strat_{\findex}  \parens*{  \frac{\pd{\pay_{\play}}\of\strat}{\pd\strat_{\findex}} + \const_{\play}\of\strat } \eqdot
\end{equation}
So far we only used condition \ref{item:gradient-on-parallel} in \cref{def:individual-gradient} of individual gradient. The last step consists of invoking condition \ref{item:gradient-is-parallel} to fix the value of $\const_{\play}$ (this is the gauge-fixing procedure mentioned in \cref{rem:gradient-gauge-fixing}):
\begin{equation}
0 = \sum_{\pure_{\play}}\gfield_{\findex}\of\strat
=  \ll \sum_{\pure_{\play}} \strat_{\findex}  \frac{\pd{\pay_{\play}}\of\strat}{\pd\strat_{\findex}} \rr + \const_{\play}\of\strat \eqdot
\end{equation}
To conclude, $\pd\pay_{\play} / \pd\strat_{\play\pure_{\play}} = \pay_{\play}(\pure_{\play};\strat_{-\play})$ by \cref{eq:pay-lin}, so $\const_{\play}\of\strat
 = - \pay_{\play}\of\strat$ and $\gfield_{\findex}\of\strat = \strat_{\play\pure_{\play}}
		\bracks{\pay_{\play}(\pure_{\play};\strat_{-\play}) - \pay_{\play}(\strat)}$.
\end{proof}

\subsection{Effective \shah metric and individual gradient}
\label{sec:effective-individual-gradient}
In the previous section we work with objects in the \textit{full} representation, showing that the individual gradients of the payoff functions give to the full expression \eqref{eq:RD} of \acl{RD}. 
By working with object in the \textit{effective} representation we can provide an alternative and more concise proof for this fact: in such representation the parallelism condition~\ref{item:gradient-is-parallel} in \cref{def:individual-gradient} is automatically fulfilled, and one can use \cref{lemma:sharp-components-inverse-metric} to verify that the individual gradients with respect to the effective \shah metric \eqref{eq:eff-rep-shah} of the effective payoff functions give the effective replicator field \eqref{eq:RD-eff}.
\begin{proof}[Alternative proof of \cref{prop:rep-Shah}]
We need to verify that the matrix identity
\begin{equation}
\effrepfield_{\effindex}\of\effstrat
= \bracks{\effgmat_{\play}^{-1} \, \Eff{d}_{ \play}\effpay_{\play}}_{\effplay}
= \bracks{\effgmat_{\play}^{-1} \, \effpayfield_{\play}}_{\effplay}\of\effstrat
\end{equation}
holds true for all $\play \in \players$, all $\effplay \in \effrangeplay$, and all $\effstrat \in \intcorcube$. The effective replicator field $\effrepfield_{\effindex}\of\effstrat$ on the left hand side is given by \cref{eq:repfield-eff} (\cf \cref{rem:notation-eff-repfield} for the notation); the first equality comes from \cref{lemma:sharp-components-inverse-metric} for the components of the gradient field; and the second equality holds true because the individual differential of the effective payoff function is the effective payoff field by \cref{eq:individual-differential}. 

The inverse matrix $\effgmat^{-1}_{\play}$ of the effective \shah metric \eqref{eq:eff-rep-shah} can be computed by the Sherman\textendash Morrison formula:
\begin{equation}
\label{eq:eff-rep-inv-shah}
\effgmat^{-1}_{\effshindex} (\effstrat) 
=  \delta_{\effplay\effaltplay} \effstrat_{\effindex} - \effstrat_{\effindex}\effstrat_{\effindexalt}
\end{equation}
for all $\play \in \players, \effplay, \effaltplay \in \effrangeplay, \effstrat \in \intcorcube$; \cf \citep[Sec.~2]{LM15}. By matrix multiplication one then gets the sought after expression for the effective replicator field:
\begin{equation}
\sum_{\effaltplay} \effgmat_{\effshindex}^{-1} \effpayfield_{\effindexalt} \of{\effstrat}
 = \sum_{\effaltplay} \parens*{ \delta_{\effplay\effaltplay} \effstrat_{\effindex} - \effstrat_{\effindex}\effstrat_{\effindexalt} }\, \effpayfield_{\effindexalt}
= \effstrat_{\effindex} \bracks*{ \effpayfield_{\effindex}\of\effstrat - \sum_{\effaltplay} \effstrat_{\effindexalt}\effpayfield_{\effindexalt}\of\effstrat} = \effrepfield_{\effindex} \of{\effstrat} \eqdot \qedhere
\end{equation}
\end{proof}
In conclusion, the effective \acl{RD} are given by individual steepest ascent on effective payoff functions with respect to the effective \shah metric.
It is illustrative to see in a simple example how the functional form of the effective \shah metric guarantees that such dynamics remain confined to the \textit{interior} of the strategy space.
\begin{example}[\shah metric and \acl{RD} in a $2 \times 2$ game]
\label{ex:RD-22-logistic}
\begin{figure}
\centering
\includegraphics[width=0.6\textwidth]{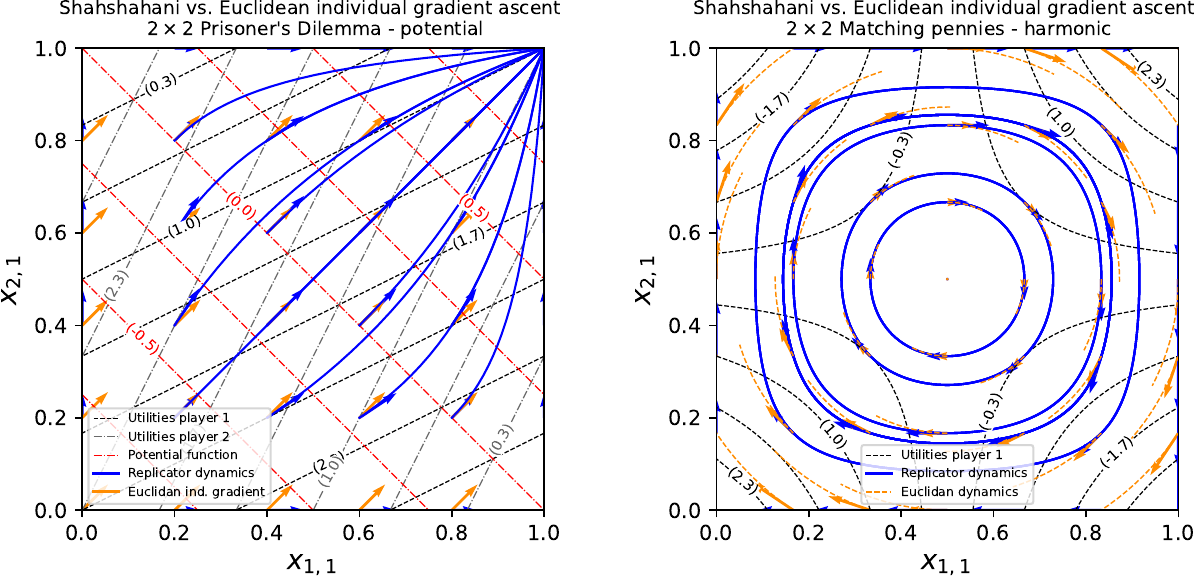}
\caption{Euclidean (orange) \vs \shah (blue) individual payoff gradients in a $2 \times 2$ potential and harmonic game (left and right respectively). Dark dotted lines represent payoff contours, and red dotted lines (left figure only) represent contours of the potential function. The replicator dynamical system \eqref{eq:logistic-system} is equivalent to individual \shah gradient ascent; the figure shows how the functional form of the inverse Shahshahani metric given by \cref{eq:functional-inverse-sha-2-2}, decaying to zero as the boundary is approached, is the key feature that confines \acl{RD} (blue) to the interior of the strategy space, whereas Euclidean steepest individual payoff ascent (orange) leads to hitting the boundary in finite time. The payoffs used in these examples are $u_{1} =(2, 0, 3, 1), u_2 = (2, 3, 0, 1)$ in Prisoner's Dilemma, that is potential with potential function $\pot = (-1, 0, 0, 1)$; and $u_1 = (3, -3, -3, 3), u_2 = (-3, 3, 3, -3)$ in rescaled Matching Pennies, that is harmonic.
}
\label{fig:RD-contours}
\end{figure}
In a $2 \times 2$ game the effective strategy space of each player is one dimensional, \ie $\effplay$ and $\effaltplay$ belong to the singleton $\setof{1_{\play}}$ for $\play \in \setof{1,2}$. As discussed in \cref{ex:22-game} the two reduced payoff fields, with one component each, are
\begin{gather}
\effpayfield_{1}(\effstrat_{1}, \effstrat_{2})
= \effstrat_{2} \bracks[\big]{ \pay_1(0, 0) -  \pay_1(0, 1) -  \pay_1(1, 0) +  \pay_1(1, 1)} -  \pay_1(0, 0) +  \pay_1(1, 0)
\\
\effpayfield_{2}(\effstrat_{1}, \effstrat_{2})
= \effstrat_{1} \bracks[\big]{ \pay_2(0, 0) -  \pay_2(0, 1) -  \pay_2(1, 0) +  \pay_2(1, 1)} -  \pay_2(0, 0) +  \pay_2(0, 1)
\end{gather}
By \cref{eq:eff-rep-inv-shah} the effective inverse \shah metric on each $1$-dimensional open corner of cube $\intcorcube = (0,1) \subset \R$ has only one component given by
\begin{equation}
\label{eq:functional-inverse-sha-2-2}
\begin{split}
\effgmat_{1}^{-1}\of{\effstrat_{1}, \effstrat_{2}}
 & = \effstrat_{1} (1-\effstrat_{1})
\\
\effgmat_{2}^{-1}\of{\effstrat_{1}, \effstrat_{2}}
 & = \effstrat_{2} (1-\effstrat_{2})
\end{split}
\end{equation}
This functional form of the inverse \shah metric is the key feature alluded at in \cref{rem:sha-boundary-1} that confines the \acl{RD} to the interior of the strategy space $\corcube$, as opposed to steepest individual payoff ascent with respect to the Euclidean metric, that leads to hitting the boundary in finite time. Expanding $\effgmat^{-1}_{i}$ around $\effstrat_{\play} = 0.5$ we see that $\effgmat^{-1}_{i}(\effstrat) \approx 0.25$, \ie toward the middle of the open corner of cube $\intcorcube$ the \shah metric is just a rescaled version of the Euclidean metric  $\effgmat^{-1}_{i} = 1$. On the other hand $\lim_{\effstrat_{\play} \to 0_{+}} \effgmat^{-1}_{i}(\effstrat) = \lim_{\effstrat_{\play} \to 1_{-}} \effgmat^{-1}_{i}(\effstrat) = 0 $, \ie toward the boundary of $\intcorcube$  the inverse \shah metric goes to zero, dampening the dynamics and bounding it to the interior of $\corcube$:  
\begin{equation}
\label{eq:logistic-system}
\begin{cases}
\effrepfield_{1}\of{\effstrat}
= \dot{\effstrat}_{1}
= \effg_{1}^{-1} \effpayfield_{1}
= \effstrat_{1} (1-\effstrat_{1}) \effpayfield_{1}\of{\effstrat}
=\effstrat_{1} \bracks*{\effpayfield_{1}\of{\effstrat} - \effstrat_{1} \, \effpayfield_{1}\of{\effstrat}}
\\
\effrepfield_{2}\of{\effstrat}
= \dot{\effstrat}_{2}
= \effg_{2}^{-1} \effpayfield_{2}
= \effstrat_{2} (1-\effstrat_{2}) \effpayfield_{2}\of{\effstrat}
= \effstrat_{2} \bracks*{\effpayfield_{2}\of{\effstrat} - \effstrat_{2} \, \effpayfield_{2}\of{\effstrat}}
\end{cases}
\end{equation}
which are \cref{eq:repfield-eff} for a $2 \times 2$ game (\cf\cref{rem:notation-eff-repfield} for the notation). The orbits of this dynamical system are plotted in \cref{fig:RD-contours} for Prisoner's Dilemma (potential) and Matching Pennies (harmonic).
\endenv
\end{example}

\subsection{Differential characterizations of potential games}
\label{sec:app-potential-games}

We conclude this appendix with two characterizations of \define{exact potential} games in the sense of \citet{MS96}.

\begin{definition}
The effective payoff field $\effpayfield$ of a finite game $\fingamefull$ is called \define{exact} if it is the differential of a function, namely if there exist a function $\pot\from\pures\to\R$, called \define{potential}, such that $\effpayfield_{\play}(\effstrat) = \Eff{d}_{\play}\Eff{\pot}(\effstrat)$ for all $\effstrat \in \corcube$ and all $\play \in \players$.
\end{definition}
\begin{remark}
We denote a function $\pot\from\pures\to\R$ and its multilinear extension $\pot\from\strats\to\R$,  $\pot(\strat) = \sum_{\pure}\strat_{\pure}\pot(\pure)$ by the same symbol, and it is understood that it is the multilinear extension undergoing differentiation. As in \cref{rem:individual-differential}, the differential $\d{f}$ of a differentiable function $f$ is the array of partial derivatives of the function. A tilde denotes as usual effective objects, and $\Eff{d}_{\play}$ denotes the array of partial derivatives with respect to the $\effstrat_{\play}$ coordinates.
\endenv
\end{remark}
Recall from the main text that a finite game $\fingame$ is exact-potential in the sense of \citet{MS96} if it admits a potential function $\pot\from\strats\to\R$ such that
\begin{equation}
\tag{\ref{eq:pot}}
\pay_{\play}(\purealt_{\play};\pure_{-\play}) - \pay_{\play}(\pure_{\play};\pure_{-\play})
	= \pot(\purealt_{\play};\pure_{-\play}) - \pot(\pure_{\play};\pure_{-\play})
\end{equation}
for all $\pure,\purealt\in\pures$ and all $\play\in\players$. The notion of exactness for the effective payoff field allows for an equivalent characterization:
\begin{lemma}
\label{lemma:potential-iff-effective-payfield-exact}
A finite game $\fingamefull$ is exact-potential in the sense of \citet{MS96} if and only if its effective payoff field $\effpayfield$ is exact.
\end{lemma}
\begin{proof}
Assume that $\effpayfield$ is exact with potential $\pot$. Then by chain rule $\effpayfield_{\effindex} \of\effstrat = \Eff{\pd}_{\effindex} \Eff{\pot}\of\effstrat = (\pd_{\effindex} - \pd_{\zindex})\pot \of\strat$, for all $\play \in \players$, $\effplay \in \effrangeplay$, and $\effstrat \in \corcube$. Since $\effpayfield_{\effindex}\of\effstrat =  \payfield_{\effindex}\of\strat - \payfield_{\zindex}\of\strat$ we conclude that $\pay_{\play}(\pure_{\play}, \strat_{\others}) - \pay_{\play}(0_{\play}, \strat_{\others}) = \pot(\pure_{\play}, \strat_{\others}) - \pot(0_{\play}, \strat_{\others})$ for all $\play \in \players$, $\pureplay \in \rangeplay$, and $\strat \in \strats$, implying in turn that the game is potential by \cref{lemma:multilinear-extension-zero}. Conversely, assume that $\fingamefull$ is potential; then the sequence of arguments is reversed without any change, showing that $\effpayfield$ is exact.
\end{proof}
In the context of population games this result can be found in \citet[Chapter 3.2.2]{San10}. The interpretation of this lemma is the following: Recall by \cref{eq:individual-differential} that the the effective payoff field $\effpayfield$ of a game is the array of individual differentials of the effective payoff functions, \ie $\effpayfield = \ll \effd_{\play} \effpay_{\play} \rr_{\play \in \players}$. As such, in general $\effpayfield$ is \textit{not} the differential of a function; \cref{lemma:potential-iff-effective-payfield-exact} says that a game is potential precisely when this is the case, namely when $\effpayfield = \ll \effd_{\play} \effpot \rr_{\play \in \players} = \effd\effpot$.

The differential characterization of potential games given above, relying on the notion of exact effective payoff field, is intrinsic of a game and \textit{independent} of any choice of metric. Yet, if a Riemannian metric on $\intstrats$ is available, the metric's non-degeneracy allows for another characterization of potential games, in a sense dual to the one given above.

Given a Riemannian metric $\gmat$ on $\intstrats$, and leveraging \cref{lemma:sharp-components-inverse-metric} relating the gradient and the differential of a function via the matrix of the metric, we say that a vector field $\vecfield$ on $\intstrats$ is the gradient of a function $\pot$ if $\vecfield = \g^{-1} \cdot d\pot$, where $\cdot$ denotes matrix multiplication. With this in mind we can give the following
\begin{lemma}
\label{prop:potential-iff-gradient}
Given a Riemannian metric $\g$ on $\intstrats$, a finite game $\fingamefull$ is an exact potential game if and only if its field of effective individual payoff gradients, denoted by $\effrepfield$, is the gradient of a function, namely if there exists a function $\pot\from\pures\to\R$ such that $\effrepfield = \effg^{-1} \cdot \effd\effpot$.
\end{lemma}
\begin{proof}
The field of effective individual payoff gradients of a game is
\begin{equation}
\effrepfield
= \ll \effgmat_{\play}^{-1} \cdot \tilde{d}_{\play} \effpay_{\play} \rr_{\play \in \players}
= \ll \effgmat_{\play}^{-1} \cdot \effpayfield_{\play} \rr_{\play \in \players} \eqdot
\end{equation}
This field is a gradient if there exists a function $\pot\from\pures\to\R$ such that 
\begin{equation}
\effrepfield
= \effgmat^{-1} \cdot \effd \effpot = \ll \effgmat_{\play}^{-1} \cdot \effd_{\play}\effpot \rr_{\play \in \players} \eqcomma
\end{equation}
where $\Eff{\pot}\from\corcube\to\R$ is the effective representation of the multilinear extension of $\pot\from\pures\to\R$. By non degeneracy of the Riemannian metric $\g$, the effective field of individual gradients is a gradient if and only if $\effpayfield_{\play} = \Eff{d}_{\play} \Eff{\pot}$ for all $\play \in \players$, \ie if the effective payoff field is exact, which is equivalent to $\fingamefull$ being an exact potential game by \cref{lemma:potential-iff-effective-payfield-exact}.
\end{proof}

An example of this result is given by \emph{Kimura's maximum principle} \cite{Kim58,Sha79}, which states that, in potential games, \acl{RD} (that is, the dynamics given by the field of individual payoff \shah gradients) is a \shah gradient system. \cref{prop:potential-iff-gradient} provides a broad generalization of this principle: in a potential game, given \textit{any} metric $\g$, the dynamics given by the field of individual payoff $\g$-gradients is a $\g$-gradient dynamics in its own right.

\paragraph{Full potential games} By \cref{lemma:potential-iff-effective-payfield-exact}, a finite game in normal form is potential if and only if its \textit{effective} payoff field is exact. In the context of population games, \citet[Chapter 3.1]{San10} calls a game such that the \textit{full} payoff field is exact, \ie  $\payfield = \d\pot$, a \define{full potential game}. One can show that if the full payoff field of the game is exact then the effective payoff field is exact too, but the converse needs not be true. In other words, $\payfield = \d\pot$ is a sufficient but not necessary condition for a game to be potential. We illustrate this fact with a simple example; for further details on the relation between potential games and full potential games we refer the reader to \citet[Chapter 3.2.3]{San10}.
\begin{example}
Consider the $2 \times 2$ potential game with payoffs $\pay$ and potential function $\pot$ given by
\begin{equation} \pay = 
\begin{pmatrix}
(2,2) & (0,3) \\
(3,0) & (1,1) \\
\end{pmatrix},
\quad
\pot = 
\begin{pmatrix}
-1 & 0 \\
0 & 1 \\
\end{pmatrix} \eqdot
\end{equation}
The corresponding \acl{RD} are plotted in \cref{fig:RD-contours}, left. Replacing for notational simplicity $(x_{1, 0}, x_{1,1}) \to (x_{0}, x_{1})$ and $(x_{2, 0}, x_{2,1}) \to (y_{0}, y_{1})$ we have from \cref{ex:22-game} that the full payoff field is
\begin{equation}
\payfield (x_0, x_1, y_0, y_1)  =
\ll
2  y_{0}   					, \, 
3  y_{0}  + y_{1} , \, 
2  x_{0}   					, \, 
3  x_{0}  + x_{1}
\rr \eqdot
\end{equation}
A simple check shows that
$
\partial_{y_0} v_{x_1} = 3 \neq 0 = \partial_{x_1} v_{y_0}
$
which, by a routine application of Poincaré's lemma \cite{hubbardVectorCalculusLinear2015}, implies that $\payfield$ is \emph{not} the differential of a function. On the other hand, the reduced payoff field is
\begin{equation}
\effpayfield( \Eff{x}_1, \Eff{y}_1 ) = (1,1) \eqcomma
\end{equation}
which is exact, since it is the differential of any function $\effpot(\Eff{x}_1, \Eff{y}_1) = \Eff{x}_1 + \Eff{y}_1 + \text{const}$. Note that the multilinear extension of the potential function of the game is $\pot( x_0, x_1, y_0, y_1 ) = - x_0 y_0 + x_1 y_1$, and its reduced form is $\pot( \Eff{x}_1, \Eff{y}_1 ) = - (1 - \Eff{x}_1) (1 - \Eff{y}_1) + \Eff{x}_1 \Eff{y}_1 = \Eff{x}_1 + \Eff{y}_1 - 1 $, showing that, albeit the full payoff field is not exact, the reduced payoff field is the differential of the reduced multilinear extension of the potential function of the game.
\end{example}

%% file: Appendix/App-Incompressible.tex

The results presented here heavily rely on the mathematical tools of \cref{app:geometry}, which we therefore recommend reading first.
\subsection{\shah divergence}
In this section we prove the results given in \cref{sec:results}. They all revolve around the expression of the \shah divergence of the effective replicator field that \textendash\ as discussed in more detail in \cref{app:div-cod-product} \textendash\ is given by \cref{eq:div-payfield} from the main text:
\begin{equation}
\tag{\ref{eq:div-payfield}}
\diver_{\play} \repfield_{\play}(\effstrat)
	= \frac{1}{\sqrt{\det\effgmat_{\play}(\effstrat_{\play})}}
		\sum_{\eff_{\play}=1}^{\nEffs_{\play}} \frac{\pd}{\pd\effstrat_{\play\eff_{\play}}}
			\parens*{\sqrt{\det\effgmat_{\play}(\effstrat_{\play})} \, \repfield_{\play\eff_{\play}}(\effstrat)} \quad \text{for all } \play \in \players \eqdot
\end{equation}
By \cref{lemma:div-does-not-mix}, the divergence operator on a product manifold - which in our case is the product $\intcorcube = \prod_{\play} \intcorcube_{\play}$ of open corner of cubes, where the effective game described in \cref{app:reduction} lives - is given by the sum of the divergence operators on the factor manifolds, which justifies \cref{def:incompressible} of incompressible games given in \cref{sec:decomposition}:
\DefIncompressible*
We devote the rest of this section to the proof of the following result:
\begin{proposition}
\label{th:shah-divergence-repfield-result}
The \shah divergence of the effective replicator field of a finite game $\fingamefull$ is
\begin{equation}
\label{eq:shah-divergence-repfield-result}
\div\repfield\of\effstrat = +\frac{1}{2} \sum_{\play\in\players}\sum_{\pureplay \in \pures_{\play}} \ll \payfield_{\findex}\of\strat - \pay_{\play}\of\strat \rr \eqdot
\end{equation}
\end{proposition}
\begin{proof}
To prove this result we will first compute  $\div_{\play}{\repfield_{\play}}\of\effstrat$ for some $\play \in \players$, dropping the index $\play$ for notational simplicity. Recall by \cref{eq:riemannian-divergence-two-terms} that the expression for the divergence can be rewritten by product rule as the sum of two terms, one metric-dependent and one metric-independent:
\begin{equation}
\div_{\play}{\repfield_{\play}}\of\effstrat
\overbrace{\equiv}^{\text{Drop } \play \text{ index}}
\div{\repfield}\of\effstrat
= \sum_{\eff = 1}^{\nEffs} \ll \frac{\de_{\effstrat_{\eff}}\sqrt{\det{\effgmat}}}{\sqrt{\det{\effgmat}}} \rr \repfield_{\eff}(\effstrat)
+ \sum_{\eff = 1}^{\nEffs} \de_{\effstrat_{\eff}}\repfield_{\eff} (\effstrat) \eqcomma
\end{equation}
where $\de_{\effstrat_{\eff}}$ is a shorthand for $\frac{\pd}{\pd\effstrat_{\eff}}$, and $\det{\effgmat}$ is given by \cref{eq:det-sha-eff}:
\begin{equation}
\tag{\ref{eq:det-sha-eff}}
\det{\effgmat}(\effstrat) 
= \frac{1}{\left(1-\sum_{\eff = 1}^{\nEffs} \effstrat_{\eff}\right)\prod_{\effalt = 1}^{\nEffs} \effstrat_{\effalt}}
= \frac{1}{ \strat_{0} \prod_{\effalt = 1}^{\nEffs} \effstrat_{\effalt}} \eqdot
\end{equation}

\paragraph{Metric-independent term of the divergence}
Denote by  $\bar{\payfield}\of\strat \defeq \sum_{\pure} \strat_{\pure}\payfield_{\pure}\of\strat$, and by $\bar{\effpayfield} \of\effstrat \defeq  \sum_{\eff}\effstrat_{\eff}\effpayfield_{\eff}\of\effstrat$, so that the full and effective \acl{RD} (with player index suppressed) read
\begin{subequations}
\begin{align}
\dot{\strat}_{\pure} = 
\repfield_{\pure}(\strat)
= \strat_{\pure}
		\bracks*{
			\payfield_{\pure}(\strat)
			- \bar{\payfield}(\strat)} \quad \text{ for all } \pure = 0, \dots, \nEffs \eqcomma
\\
\dot{\effstrat}_{\eff} = 
\repfield_{\eff}(\effstrat)
	= \effstrat_{\eff}
		\bracks*{\effpayfield_{\eff}(\effstrat)
			- \bar{\effpayfield}(\effstrat)}  \quad  \text{ for all } \eff = 1, \dots, \nEffs \eqdot
\end{align}
\end{subequations}

\begin{remark}
Reinserting for the scope of this remark the player index, a simple check shows that $\bar{\effpayfield}_{\play}(\effstrat) = \bar{\payfield}_{\play}(\strat) - \payfield_{\zindex}(\strat)$; note that $\bar{\payfield}_{\play}\of\strat$ is precisely the full payoff function $\pay_{\play}\of\strat$, while $\bar{\effpayfield}_{\play}\of\effstrat$ is \textit{not} the effective payoff function $\effpay_{\play}\of\effstrat$, the difference between the two being precisely $\payfield_{\zindex}\of\strat$.
\endenv
\end{remark}
As usual, the indices $\eff$ and $\effalt$ appearing in the following expressions are understood to run over $1, \dots, \nEffs$; and the indices $\pure$ and $\purealt$ are understood to run over $0, \dots, \nEffs$. With these notational caveats in place, the metric-independent term of the divergence reads
\begin{equation}
\label{eq:div-metric-independent-full}
\begin{split}
\sum_{\eff}\de_{\effstrat_{\eff}}\repfield_{\eff}(\effstrat) & =  \sum_{\eff}\de_{\effstrat_{\eff}} \Big( \effstrat_{\eff}(\effpayfield_{\eff}(\effstrat) - \bar{\effpayfield}(\effstrat)) \Big) \\
& =  \sum_{\eff}(\effpayfield_{\eff}-\bar{\effpayfield}) +\sum_{\eff}\effstrat_{\eff}\de_{\effstrat_{\eff}}\effpayfield_{\eff}-\bar{\effpayfield} - \sum_{\eff\effalt}\effstrat_{\eff}\effstrat_{\effalt}\de_{\effstrat_{\eff}}\effpayfield_{\effalt} \eqdot \\
\end{split}
\end{equation}
By \cref{lemma:payfield-independent-strat} the second and fourth terms vanish: re-inserting the player index $\play \in \players$, 
$
\frac{\pd \effpayfield_{\effindexalt}}{\pd\effstrat_{\effindex}}\of\effstrat \equiv 0
$
since the components of the reduced payoff field of player $\play \in \players$ do not depend on the mixed coordinates of player $\play$. This leads to
\begin{equation}
\sum_{\eff}\de_{\effstrat_{\eff}}\repfield_{\eff}(\effstrat) 
= \sum_{\eff} \ll \payfield_{\eff} - \payfield_{0} - \bar{\payfield} + \payfield_{0} \rr
- \bar{\payfield} + \payfield_{0}
= \sum_{\pure} (\payfield_{\pure} - \bar{\payfield}) \eqcomma
\end{equation}
and in conclusion the metric-independent term of the divergence $\div{\repfield}\of\effstrat$ is
\begin{equation}
\label{eq:divergence-metric-independent}
\sum_{\eff}\de_{\effstrat_{\eff}}\repfield_{\eff}(\effstrat) =  \sum_{\pure} \ll \payfield_{\pure}\of\strat - \bar{\payfield}\of\strat \rr  \eqcomma
\end{equation}
where it is understood that effective objects (\ie containing the effective payoff field $\effpayfield$) are evaluated at $\effstrat$, and full objects (\ie containing the full payoff field $\payfield$) are evaluated at $\strat$, with $\strat$ and $\effstrat$ related as usual by the map $\chart$ of \cref{eq:simplex-chart}.

In settings with non-multilinear utility functions, such as games with continuous action sets and differentiable payoff functions \citep{letcherDifferentiableGameMechanics2019,MZ19,Ros65}, the terms of \cref{eq:div-metric-independent-full} containing derivatives of the payoff field do \textit{not} necessarily vanish; as a reference for possible applications to these settings we provide below the version of \cref{eq:divergence-metric-independent} including such terms. To this end denote by $B = \jac{\effpayfield}$ the Jacobian matrix of the effective payoff field, \ie $B_{\effalt\eff} = \de_{\effstrat_{\eff}}\effpayfield_{\effalt}$, and  recall that $\effpayfield_{\eff} = \payfield_{\eff} - \payfield_{0}$ to rewrite \cref{eq:div-metric-independent-full} as
\begin{equation}
\sum_{\eff}\de_{\effstrat_{\eff}}\repfield_{\eff}(\effstrat)
= \sum_{\pure} (\payfield_{\pure} - \bar{\payfield})
+ \ll \sum_{\eff}\effstrat_{\eff}B_{\eff\eff} - \sum_{\eff\effalt}\effstrat_{\eff}\effstrat_{\effalt}B_{\eff\effalt} \rr \eqdot
\end{equation}
The idea now is that of re-arranging the terms that contain effective coordinates $\effstrat$ and combine them into objects expressed in terms of full coordinates $\strat$. After a straightforward but tedious computation leveraging \cref{lemma:chain-derivative-eff-payfield} the second term on the right hand side can be expressed in terms of full objects as
\begin{equation}
 \sum_{\eff}\effstrat_{\eff}B_{\eff\eff} - \sum_{\eff\effalt}\effstrat_{\eff}\effstrat_{\effalt}B_{\eff\effalt} 
= \sum_{\pure}\strat_{\pure}\de_{\pure}\payfield_{\pure} - \sum_{\pure\purealt}\strat_{\pure}\strat_{\purealt}\de_{\pure}\payfield_{\purealt}
\eqdot
\end{equation}

\paragraph{Metric-dependent term of the divergence}
A direct computation shows that
\begin{equation}
\frac{\de_{\effstrat_{\eff}}\sqrt{\det{\effgmat}}}{\sqrt{\det{\effgmat}}} = -\frac{1}{2}\frac{\effstrat_{0}-\effstrat_{\eff}}{\effstrat_{0}\effstrat_{\eff}} \eqdot
\end{equation}
Again, one needs to manipulate the terms containing effective coordinates to express them in terms of full coordinates. After an elementary but lengthy calculation, we obtain that the metric-dependent term of the divergence $\div{\repfield}\of\effstrat$ is
\begin{equation}
\label{eq:divergence-metric-dependent}
\sum_{\eff}\repfield_{\eff}(\effstrat) \, \frac{\de_{\effstrat_{\eff}}\sqrt{\det{\effgmat}}}{\sqrt{\det{\effgmat}}}\of\effstrat
= -\frac{1}{2} \sum_{\pure} \ll \payfield_{\pure}\of\strat - \bar{\payfield}\of\strat \rr
\eqdot
\end{equation}
\paragraph{Conclusion}
Summing \cref{eq:divergence-metric-dependent,eq:divergence-metric-independent} and re-inserting the player index we finally get
\begin{equation}
\div_{\play}\repfield_{\play}\of\effstrat = +\frac{1}{2} \sum_{\pureplay \in \pures_{\play}} \ll \payfield_{\findex}\of\strat - \pay_{\play}\of\strat \rr
\eqcomma
\end{equation}
and the proof is completed by summing over the player index $\play \in \players$.
\end{proof}
To conclude this section we provide an equivalent expression for the \shah divergence of the effective replicator field. For all $\play \in \players$ define the \define{barycenter} $\b_{\play} \in \strats_{\play} $ as the point with coordinates
\begin{equation}
\triple{\b}{\play}{\pure} \defeq \frac{1}{\nPures_{\play}} \quad \text{for all } \pureplay \in \pures_{\play}
\end{equation}
on the mixed strategy space $\strats_{\play}$ of player $\play \in \players$, where $\nPures_{\play} = \abs{\pures_{\play}}$ is the number of pure strategies of player $\play$. Similarly, define $\1_{\play} \defeq (1, \dots, 1) \in \R^{\nPures_{\play}}$.
Then
\begin{lemma}[Equivalent expression of the \shah divergence]
\label{prop:equivalent-divergence-expressions}
The \shah divergence of the effective replicator field of a finite game $\fingamefull$ fulfills
\begin{equation}
- \div\repfield\of\effstrat
= \frac{1}{2}\sum_{\play \in \players} \sum_{\pure_{\play} \in \pures_{\play}} \ll \pay_{\play} - \V{\play}{\pure} \rr \of \strat 
= \frac{1}{2}\sum_{\play} \nPures_{\play} \, \payfield_{\play}\of \strat  \cdot \ll \strat_{\play} -  \b_{\play}\rr 
\eqdot
\end{equation}
\end{lemma}
\begin{proof}
From \cref{th:shah-divergence-repfield-result} and the application of definitions it follows that
\begin{equation}
\begin{split}
- \div\repfield\of\effstrat
& = \frac{1}{2}\sum_{\play \in \players} \sum_{\pure_{\play} \in \pures_{\play}} \ll \pay_{\play} - \V{\play}{\pure} \rr \of \strat 
 = \frac{1}{2}\sum_{\play} \ll \nPures_{\play}\pay_{\play} - \payfield_{\play} \cdot \1_{\play} \rr \of \strat \\
& = \frac{1}{2}\sum_{\play} \nPures_{\play}\ll\pay_{\play} - \payfield_{\play} \cdot \b_{\play}\rr \of \strat 
 = \frac{1}{2}\sum_{\play} \nPures_{\play} \, \payfield_{\play} \of \strat \cdot \ll \strat_{\play} -  \b_{\play}\rr  \eqdot \qedhere
\end{split}
\end{equation}
\end{proof}
\subsection{Incompressible games are precisely harmonic games}
The expression of the \shah divergence of the field of individual \shah payoff gradients of a game $\div\repfield\of\effstrat$ allows to establish an important connection between incompressible games, \ie those game for which $\div\repfield$ vanishes identically, and the harmonic games introduced by \citet{CMOP11}. Recall by \cref{eq:harm} that a game is harmonic if
\begin{equation}
\tag{\ref{eq:harm}}
\insum_{\play\in\players} \insum_{\purealt_{\play}\in\pures_{\play}}
	\bracks{\pay_{\play}(\purealt_{\play};\pure_{-\play}) - \pay_{\play}(\pure_{\play};\pure_{-\play})}
	= 0
\end{equation}
for all $\pure \in \pures$. With this at hand we can prove \cref{thm:harmonic}:
\ThHarmonicIncompressible*
\begin{proof}
Begin by noting that the \cref{eq:harm} characterizing harmonic games can be recast as
$
F(\alpha) = 0
$
for all $\pure \in \pures$, with  $F \from \pures \to \R$ defined by
\begin{align}
F(\alpha)
	&\defeq \sum_{i \in \players}\sum_{\beta_{i} \in \pures_{\play}} \parens[\Big]{ u_{i}(\beta_{i}, \alpha_{-i}) - u_{i}(\alpha_{i}, \alpha_{-i}) }
	\notag\\
	&= \sum_{i} \bracks*{ \ll \sum_{\beta_{i}}  u_{i}(\beta_{i}, \alpha_{-i}) \rr - \nPures_{i} u_{i}(\alpha) }
	\eqdot
\end{align}
We claim that the \shah divergence of the effective replicator field of a finite game is (up to a factor of $1/2$) the multilinear extension of the function $F$ defined above:
\begin{equation}
\label{eq:shah-divergence-repfield-multilinear}
\div\repfield\of\effstrat = \frac{1}{2} \bar{F}(\strat) \defeq \frac{1}{2} \sum_{\pure \in \pures} \strat_{\pure} F(\pure) \eqcomma
\end{equation}
for all $\effstrat \in \corcube$. If we show this the proof is concluded, since invoking \cref{lemma:multilinear-extension-zero} we have
\begin{align}
\text{harmonic }
	&\iff F \equiv 0 \eqcomma
	\notag\\
\text{incompressible }
	&\iff \bar{F} \equiv 0 \eqcomma
	\notag\\
F \equiv 0
	&\iff \bar{F} \equiv 0 \eqdot
\end{align}
\Cref{eq:shah-divergence-repfield-multilinear} is verified by a standard multilinear calculation:
\begin{align}
\bar{F}(\strat)
	&= \sum_{\pure \in \pures} \strat_{\pure} F(\pure)
	= \sum_{\play} \bracks*{ \ll \sum_{\beta_{i}}  u_{i}(\beta_{i}, \strat_{-i}) \rr - \nPures_{\play}\pay_{\play}(\strat) }
	= \sum_{\play} \ll  \payfield_{\play} \cdot \1_{\play} - \nPures_{\play}\pay_{\play} \rr (\strat)
\end{align}
for all $\strat \in \strats$; so from \cref{prop:equivalent-divergence-expressions} we have that $\div\repfield = 1/2 \bar{F}$, concluding the proof.
\end{proof}
Having established the equivalence between harmonic and incompressible games, the following decomposition result comes as an immediate corollary of the strategic decomposition \eqref{eq:Candogan}:
\ThDecomposition*
\begin{proof}
As shown by \citet{CMOP11} every finite game can be decomposed as the sum of a potential and a harmonic game, which gives the decomposition $\fingame = \potgame + \incgame$ in light of the fact that a game is incompressible if and only if it is harmonic.

As for the second point, the field of individual gradients of the incompressible game $\incgame$ is incompressible by definition; and \cref{prop:potential-iff-gradient} asserts that a game is an exact potential game in the sense of \citet{MS96} if and only if its effective field of individual gradients is the Riemannian gradient of a function $\pot$, \ie the field $\repfield$ of the potential game $\potgame$ can be expressed as $\grad{\pot}$.
\end{proof}
\subsection{Dynamics on incompressible games and Poincaré’s recurrence}
\label{app:poincare-recurrence-primal-dual}
Finally we turn our attention at the dynamical properties of \acl{RD} on incompressible games; our first result is a consequence of  the Riemannian version of Liouville's theorem presented in \cref{sec:flows-manifolds}.
\PropVolumeConservation*
\begin{proof}
By definition, a game is incompressible if and only if the \shah divergence of the field $\repfield\of\effstrat$ vanishes identically. Since \eqref{eq:RD} is the dynamical system on $\intcorcube$ given by $\dot{\effstrat} = \repfield\of\effstrat$, as a consequence of Liouville's theorem (\cf \cref{eq:zero-riemannian-divergence-volume-conservation}) we have that
$
\vol\of{\open_{\time}} = \vol\of{\open}
$
for all open sets $\open \subseteq \intcorcube$ and all $\time \in \R$, meaning that \eqref{eq:RD} is volume-preserving with respect to the \shah volume form.
\end{proof}
Our next result shows that \acl{RD} on incompressible games admit a constant of motion.
\ThmIncompressibleConstantOfMotion*
\begin{proof}
For each $\play \in \players$ consider the Kullback–Leibler divergence with respect to the barycenter $\KL_{\b_{\play}} \from \intstrats_{\play}\to\R$,
\begin{equation}
\KL_{\b_{\play}}(\strat_{\play})
=  \sum_{\pureplay \in \pures_{\play}} \b_{\findex} \log{ \frac{\b_{\findex}}{\strat_{\findex}} } \eqdot
\end{equation}
and scale it by the number $\nPures_{\play}$ of pure strategies of player $\play$:
\begin{equation}
\energy_{\play} (\strat_{\play})
\defeq  \nPures_{\play} \KL_{\b_{\play}} (\strat_{\play})
= - \ll  \nPures_{\play}  \log{\nPures_{\play}} +  \sum_{\pureplay \in \pures_{\play}} \log{\strat_{\findex}} \rr \eqdot
\end{equation}
The derivative of $\energy_{\play}$ along a solution trajectory of \eqref{eq:RD} is 
$
\frac{d}{d\time} \energy_{\play}(\strat_{\play}(\time))  
=   -\sum_{\pureplay \in \pures_{\play}} \frac{\dot{\strat}_{\findex}}{\strat_{\findex}}
=   -\sum_{\pureplay \in \pures_{\play}} \bracks{ \payfield_{\findex}(\strat) - \pay_{\play}(\strat)   }
$, so by \cref{th:shah-divergence-repfield-result},
\begin{equation}
\div\repfield\of{\effstrat\of\time} =  -\frac{1}{2}   \frac{\d}{\d{t}} \sum_{\play \in \players}  \energy_{\play}(\strat_{\play}(\time)) \quad \text{along RD}  \eqdot
\end{equation}
It follows that the $\nPures_{\play}$-weighted sum $\energy\from\intstrats\to\R$ of Kullback–Leibler divergences with respect to the barycenter,
\begin{equation}
\label{eq:constant-of-motion}
\energy\of\strat
\defeq \sum_{\play \in \players} \energy _{\play} \of{\strat_{\play}}
=  \sum_{\play \in \players}  \nPures_{\play} \KL_{\b_{\play}} (\strat_{\play})
= - \ll \sum_{\play \in \players} \sum_{\pureplay \in \pures_{\play}} \log{\strat_{\findex}} + \text{const} \rr \eqcomma
\end{equation}
is a constant of motion for the \eqref{eq:EW}\,/\,\eqref{eq:RD} dynamics on incompressible games, with $\text{const} = \sum_{\play} \nPures_{\play} \log{\nPures_{\play}}$.
\end{proof}
The function $\energy \from \intstrats\to\R$ is a nonnegative weighted sum of convex functions, hence it is convex. Its sublevel sets are then compact convex sets, and as such homeomorphic to closed balls, making in turn their boundaries \textendash{} the level sets of $\energy$ \textendash{} homeomorphic to $\left[(\sum_{\play}\nPures_{\play})-\nPlayers-1\right]$-dimensional spheres. Since \ac{RD} solution trajectories are constrained to the level sets of the integral of motion $\energy$, this shows that $\intstrats$ admits a foliation under \acl{RD} on harmonic games with leaves given by concentric topological spheres, as mentioned in \cref{sec:results} in the main body of the article.
\begin{remark}[Harmonic \vs zero-sum games]
\label{rem:harmonic-vs-zero-sum}
In two-player \acp{ZSG} with an interior Nash equilibrium $\strat^{\ast}$ the sum of Kullback-Leibler divergences with respect to said equilibrium, $\sum_{\play} \KL_{\strat^{\ast}_{\play}} (\strat_{\play})$,  is a constant of motion for \acl{RD} \citep{MPP18}. Harmonic games and \acp{ZSG} have non-trivial intersection; in particular, normalized harmonic games where all players have the same number of strategies are zero-sum \citep{CMOP11}. In this case the $\nPures_{\play}$-s factor out from \cref{eq:constant-of-motion}, and the two constants of motions coincide, up to affine transformations.

Despite having non-trivial intersection, harmonic and \acp{ZSG} exhibit some very different properties. A \ac{ZSG} can also be a potential game, while the only harmonic game which admits a potential is the zero game (up to strategic equivalence)\citep{CMOP11}; Harmonic games always admit a fully mixed equilibrium \citep{CMOP11}, while \acp{ZSG} do not (that is, not always); \acp{ZSG} may admit strict Nash equilibria, while harmonic games never do; \acl{RD} are recurrent in \acp{ZSG} with a \textit{fully mixed} equilibrium \citep{MPP18} but convergent in \acp{ZSG} with a \textit{strict} equilibrium \citep{mertikopoulosLearningGamesReinforcement2016}, while they are always recurrent in harmonic games (\cref{thm:recurrent}). As minimal example setting apart harmonic and zero-sum games consider the $2\times 3$ game with payoffs
\begin{equation}
u_1=
\begin{pmatrix}
a & b & -a-b \\
-a & -b & a+b
\end{pmatrix}
\quad \text{and} \quad
u_2 = -\frac{2}{3}u_1 \eqdot
\end{equation}
This game is is harmonic for every choice of $a,b \in \mathbb{R}$, and never zero-sum (except for the trivial case).  
\endenv
\end{remark}
\paragraph{Poincaré recurrence in mixed strategy space}
We turn now to our last result, namely to the fact that \acl{RD} on harmonic games are \textit{recurrent} in the sense of~\citet{poincareProblemeTroisCorps1889}.

Several works in the literature \citep{MPP18,PS14} discuss Poincaré recurrence in the context of zero-sum games with an interior Nash equilibrium, and positive affine transformations or polymatrix/network versions of the above.  The idea is to transform via a suitable diffeomorphism the replicator system in the interior of the strategy space to a system which is divergence-free \textit{under the Euclidean metric}; and to show that, under this transformations, all the orbits of \ac{RD} in the particular class of games at hand are bounded, \eg by exhibiting a constant of motion with bounded level sets.

We tackle the problem from a different angle: our proof relies on the fact that in harmonic games the replicator system itself \textendash\ without undergoing any transformation \textendash\ is volume-preserving under the \shah metric, and that the \shah volume of the space of mixed strategies is finite; the conclusion then follows from the Riemannian versions of Poincaré's theorem presented in \cref{sec:flows-manifolds}.

\ThIncompressibleRecurrent*

\begin{proof}
By \cref{rem:riemannian-mfld-is-measure-space}, Poincaré's theorem applies to a Riemannian manifold $(\mfld, \g)$ given that \begin{enumerate*}[(\itshape a\upshape)] \item there is a volume preserving map $\phi\from\mfld\to\mfld$, and \item the manifold has finite $\g$-volume
\end{enumerate*}. By \cref{prop:preserve}, the flow of the replicator vector field $\effrepfield$ on incompressible (\ie harmonic) games is volume-preserving with respect to the effective \shah metric on the open corner of cube.\footnote{More precisely, the \textit{orbit map} $\flowt \from \intcorcube \to \intcorcube$ of the replicator vector field is volume-preserving under the \shah metric; \cf \cref{sec:flows-manifolds}.}  Hence, if we show that the \shah volume of the open corner of cube is finite, namely that $\vol_{\text{Sha}}{\intcorcube} < \infty$, by Poincaré's theorem we can conclude that the solution trajectories of \acl{RD} on incompressible games return arbitrarily close to almost every starting point $\effstrat \in \intcorcube$.

To show that the volume of the $\nEffs_{\play}$-dimensional open corner of cube $\intcorcube_{\nEffs_{\play}}$ with respect to the effective \shah metric fulfills $\vol_{\text{Sha}}{\intcorcube_{\nEffs_{\play}}} < \infty$ for all $\play \in \players$ and all natural $\nEffs_{\play}$ we resort to a transformation first introduced by \citet[p. 39]{akinGeometryPopulationGenetics1979} and discussed \eg by \citet[p. 228]{San10} and \citet[Example 3.1]{LM15}. Suppressing for a moment the player index,
 \citet{akinGeometryPopulationGenetics1979} shows that the map $A \from \R^{\nEffs+1}_{>0} \to \R^{\nEffs+1}_{>0}$, $A_\pure(\strat) = 2 \sqrt{\strat_{\pure}}$ for $\pure \in \setof{0, \dots, \nEffs}$, is an isometry between the $\nEffs$-dimensional open simplex $\intstrats$ endowed with the Shahshahani metric and the portion of the radius-$2$ $\nEffs$-hypersphere in the positive hyperoctant of $\R^{\nEffs+1}$ endowed with the Euclidean metric. Isometries between Riemannian manifolds preserve volumes,\footnote{To be precise, \textit{orientation-preserving} isometries between Riemannian manifolds preserve volumes; see \citet{leeIntroductionRiemannianManifolds2018}.} so (reinstating the player index $\play \in \players$) the Shahshahani volume of the $\nEffs_{\play}$-dimensional open simplex is the Euclidean $\nEffs_{\play}$-volume of the portion of the $\nEffs_{\play}$-hypersphere of radius 2 in the positive orthant of $\R^{\nEffs_{\play}+1}$, that is
\begin{equation}
\label{eq:sha-volume-simplex}
\vol_{\text{Sha}}\intstrats_{\play} =\frac{\pi^{\frac{\nEffs_{\play}+1}{2}}}{\Gamma(\frac{\nEffs_{\play}+1}{2})} < \infty \text{ for all } m_{\play} \in \N \text{ and } \play \in \players \eqcomma
\end{equation}
where $\Gamma$ is the gamma function.
By construction, the map $\incl\from\intcorcube_{\play}\to\intstrats_{\play}$ is an isometry between the open corner of cube with the reduced \shah metric and the open simplex with the full \shah metric, so \cref{eq:sha-volume-simplex} gives also the sought-after volume of the open corner of cube, finite as required. Finally, since the volume of a product manifold is the product of the factor manifolds, we have that $\vol_{\text{Sha}} (\intcorcube) = \prod_{\play} \vol_{\text{Sha}}( \intcorcube_{\play}) < \infty$.
\end{proof}
\begin{remark} An alternative formula to compute the volume of the open simplex under the \shah metric is the following. Suppress the player index $\play \in \players$ for notational simplicity, and let as usual $\eff, \effalt \in \setof{1, \dots, \nEffs}$. The determinant of the \shah metric $\effgmat$ in its effective representation is given by \cref{eq:det-sha-eff}, so by \cref{eq:riemannian-volume} the volume of the $\nEffs$-dimensional open corner of cube $\intcorcube_{\nEffs}$ with respect to the effective \shah metric is
\begin{equation}
\begin{split}
\vol_{\text{Sha}}{\intcorcube_{\nEffs}} &= \int_{\intcorcube} \sqrt{\det{\effg}} \,  \d{\effstrat}^1 \dots \d{\effstrat}^{\nEffs}  \\
&  = \int_{\effstrat_1>0, \,  \dots, \,  \effstrat_{\nEffs}>0,\, \sum_{\eff}\effstrat_1<1}\sqrt{\det{\effg}} \, \d{\effstrat}^1 \dots \d{\effstrat}^{\nEffs} \\
& = \int_{\effstrat_1:0}^1\int_{\effstrat_2:0}^{1-\effstrat_1}\cdots\int_{\effstrat_{\nEffs}:0}^{1-\effstrat_1-\cdots-\effstrat_{\nEffs-1}}\frac{1}{\sqrt{\left(1-\sum_{\eff}\effstrat_{\eff}\right)\prod_\effalt\effstrat_\effalt}} \d{\effstrat}^{\nEffs} \cdots \d{\effstrat}^{1} \\
& = \int_{\effstrat_1=0}^{\xi_1}\frac{\d{\effstrat}^1}{\sqrt{\effstrat_1}} \int_{\effstrat_2=0}^{\xi_2}\frac{\d{\effstrat}^2}{\sqrt{\effstrat_2}} \cdots \int_{\effstrat_{\nEffs-1}=0}^{\xi_{\nEffs-1}}\frac{\d{\effstrat}^{\nEffs-1}}{\sqrt{\effstrat_{\nEffs-1}}}\int_{\effstrat_{\nEffs}=0}^{\xi_{\nEffs}}\frac{\d{\effstrat}^{\nEffs}}{\sqrt{\effstrat_{\nEffs}(\xi_{\nEffs}-\effstrat_{\nEffs})}} \eqcomma
\end{split}
\end{equation}
with $\xi_{\nEffs} = 1-\effstrat_1-\cdots-\effstrat_{\nEffs-1}$. One can check for a few values of $\nEffs$ that the integrals are in agreement with \cref{eq:sha-volume-simplex}, for example
$
\vol_{\text{Sha}}\intcorcube_{m = 1} = \pi,
\vol_{\text{Sha}}\intcorcube_{m = 2} = 2\pi,
\vol_{\text{Sha}}\intcorcube_{m = 3} = \pi^2
$. \endenv
\end{remark}

The proof of the fact \acl{RD} exhibit Poincaré recurrence in harmonic games relies on the facts that \begin{enumerate*}[(\itshape a\upshape)] \item the flow of \ac{RD} is {\shah}-incompressible in the space of mixed strategies \textit{precisely} in such games, and \item the \shah volume of the open simplex is finite.\end{enumerate*} Both of these phenomena are quite peculiar, as we discuss in the next section.

\subsection{On the special nature of the \shah geometry}
We conclude this appendix with a discussion of the fact that the \shah metric is intrinsically related to the \acl{RD} (cf. \cref{prop:rep-Shah}), but not at all to the simplicial complex of the game's response graph endowed with the Euclidean metric. In particular, incompressible games are defined \textit{completely independently} of the harmonic games of \citet{CMOP11}, and it is only through the lengthy calculations of this appendix (\cref{th:shah-divergence-repfield-result}, which ultimately leads to \cref{thm:harmonic}) that the two structures are shown to be, in fact, compatible.  

It is this difference in the origin of harmonic and incompressible games \textendash\ Euclidean-combinatorial on one side, dynamical/geometric in a \shah framework on the other \textendash\ which makes the equivalence of harmonic and incompressible games, in our opinion, unexpected. This idea is supported by the following example:
\begin{example}
Since harmonic games are defined relative to the Euclidean metric on the simplicial complex of the game's response graph, it would make sense to consider the \textit{Euclidean} projection dynamics on the simplex (less widely used than the replicator / exponential weights dynamics, but still a valid choice of game dynamics). Following
\citet{Fri91, LS08, mertikopoulosRiemannianGameDynamics2018},
in the interior of the simplex these dynamics take the simple form  
\begin{equation}
\dot x_{i \alpha_i} = v_{i\alpha_i}(x) - \frac{1}{|\mathcal{A}_i|} \sum_{\beta_i \in \mathcal{A}_i} v_{i\beta_i}(x) \quad \text{for all } \strat \in \intstrats, \,  \play \in \players, \pureplay \in \pures_{\play} \eqdot
\end{equation}
The RHS of these dynamics is simply the (Euclidean) projection of $v_i(x)$ onto $\intstrats_{\play}$. In this regard, the definition of the divergence boils down to the standard form from calculus, \cf \cref{eq:euclidean-divergence}. However, since by \cref{lemma:payfield-independent-strat} $v_i(x)$ does not depend on $x_i$, it follows that the Euclidean projection dynamics are incompressible under the Euclidean metric in the space of mixed strategies \textit{for all games}, not only for harmonic games.

Put differently, all games are incompressible under the Euclidean metric, so \textit{the equivalence between harmonic and incompressible games does not hold for the Euclidean metric on the simplex}: pictorically,
\begin{align*}
( \g  = \text{\shah metric} ) & \implies (\g\text{-incompressible game} \iff \text{harmonic game} ) 	\eqcomma \\
( \g  = \text{Euclidean metric} ) & \implies ( \g\text{-incompressible game} \iff \text{any game} ) \eqdot
\end{align*}
\end{example}

%% file: Appendix/App-Additional.tex

\subsection{Comparison with the work by \citet{letcherDifferentiableGameMechanics2019}}
\label{app:letcher}

Our work addresses an open issue raised by \citet{letcherDifferentiableGameMechanics2019}, who state that
\textit{"Candogan et al. derive a Hodge decomposition for games that is closely related in spirit to our generalized Helmholtz decomposition \textendash\ although the details are quite different. Their losses are multilinear, which is easier than our setting, but they have constrained solution sets, which is harder in many ways. Their approach is based on combinatorial Hodge theory \citep{jiangStatisticalRankingCombinatorial2011} rather than differential and symplectic geometry. Finding a best-of-both-worlds approach that encompasses both settings is an open problem."}

The machinery we developed does touch on both worlds above, as it connects the differential-geometric Hodge/Helmholtz decomposition to a constrained setting. However, there is a key difference between the spirit of our approach and that of \citep{letcherDifferentiableGameMechanics2019}: \citet{letcherDifferentiableGameMechanics2019} do not propose a decomposition of games. The authors identify two classes of games with well-understood dynamical properties based on the symmetric and skew-symmetric part of the game's Jacobian matrix, and use the symmetric and skew-symmetric components of this matrix to introduce \textit{Symplectic Gradient Adjustment} (SGA), an algorithm for finding stable fixed points in differentiable game.

In more detail, \citet{letcherDifferentiableGameMechanics2019} consider the individual gradients $\xi$ of a game with differentiable losses, and decompose \textit{the Jacobian matrix} $J$ of $\xi$ into its symmetric and skew-symmetric part, $J = S+A$. They then call a game \textit{potential} if $A=0$, and \textit{Hamiltonian} if $S=0$. Hamiltonian games exhibit non-convergent behavior under standard gradient descent methods which is similar to that of incompressible games under \ac{EW}/\ac{RD}. However, given a game with individual gradient field $\xi$, it is not possible in general to find a potential game $\xi_P$ and a Hamiltonian game $\xi_H$ such that $\xi = \xi_P + \xi_H$, the problem with this approach being that the Jacobian of a vector field is a coordinate-dependent object that does not have an intrinsic geometrical meaning, and its skew-symmetric component is in general not integrable. By contrast, \cref{thm:Hodge} provides precisely such a decomposition for normal form games into a potential and incompressible/harmonic component, resolving the dynamic-strategic disconnect that arises when trying to naively apply the standard Helmholtz decomposition to the field of individual gradients of a game.

\subsection{Combination of potential and harmonic games and dynamics}
\label{app:harmonic-potential-trajectories}
In this section we include a sequence of \acl{RD} trajectories on a convex combination of a harmonic and a potential game, showing how Poincaré recurrence breaks down as the relative magnitude of the potential component increases. More precisely, given the payoff $\pay_{p}$ of a potential game and the payoff $\pay_{h}$ of a harmonic game, we run \ac{RD} on the game with payoff $\pay \defeq \potparam  \pay_{p} + (1-\potparam) \pay_{h}$, with $\potparam \in [0,1]$. \cref{fig:pot-harm-trajectories} shows the resulting trajectories for a $2 \times 2 \times 2$ with $\pay_{p}$ and $\pay_{h}$ given respectively by the following tables:
\begin{align*}
\left[
\begin{aligned}
 \pay_{1}[0, 0, 0] & = -14\\
 \pay_{1}[1, 0, 0] & = -8\\
 \pay_{1}[0, 1, 0] & = -18\\
 \pay_{1}[1, 1, 0] & = -7\\
 \pay_{1}[0, 0, 1] & = 13\\
 \pay_{1}[1, 0, 1] & = 8\\
 \pay_{1}[0, 1, 1] & = -8\\
 \pay_{1}[1, 1, 1] & = 1\\
 \end{aligned}
 \quad
 \begin{aligned}
 \pay_{2}[0, 0, 0] & = -16\\
 \pay_{2}[1, 0, 0] & = -16\\
 \pay_{2}[0, 1, 0] & = 2\\
 \pay_{2}[1, 1, 0] & = 7\\
 \pay_{2}[0, 0, 1] & = 6\\
 \pay_{2}[1, 0, 1] & = 0\\
 \pay_{2}[0, 1, 1] & = -1\\
 \pay_{2}[1, 1, 1] & = 7\\
 \end{aligned}
 \quad
 \begin{aligned}
 \pay_{3}[0, 0, 0] & = -7\\
 \pay_{3}[1, 0, 0] & = 0\\
 \pay_{3}[0, 1, 0] & = 2\\
 \pay_{3}[1, 1, 0] & = 8\\
 \pay_{3}[0, 0, 1] & = 8\\
 \pay_{3}[1, 0, 1] & = 4\\
 \pay_{3}[0, 1, 1] & = -8\\
 \pay_{3}[1, 1, 1] & = -4\\
\end{aligned}
\right]
\\[\medskipamount]
\left[
\begin{aligned}
 \pay_{1}[0, 0, 0] & = 7\\
 \pay_{1}[1, 0, 0] & = 2\\
 \pay_{1}[0, 1, 0] & = 1\\
 \pay_{1}[1, 1, 0] & = 7\\
 \pay_{1}[0, 0, 1] & = -29\\
 \pay_{1}[1, 0, 1] & = -6\\
 \pay_{1}[0, 1, 1] & = 24\\
 \pay_{1}[1, 1, 1] & = 0\\
  \end{aligned}
 \quad
 \begin{aligned}
 \pay_{2}[0, 0, 0] & = -15\\
 \pay_{2}[1, 0, 0] & = -3\\
 \pay_{2}[0, 1, 0] & = -10\\
 \pay_{2}[1, 1, 0] & = 2\\
 \pay_{2}[0, 0, 1] & = 23\\
 \pay_{2}[1, 0, 1] & = -9\\
 \pay_{2}[0, 1, 1] & = 0\\
 \pay_{2}[1, 1, 1] & = 4\\
  \end{aligned}
 \quad
 \begin{aligned}
 \pay_{3}[0, 0, 0] & = -8\\
 \pay_{3}[1, 0, 0] & = 4\\
 \pay_{3}[0, 1, 0] & = 1\\
 \pay_{3}[1, 1, 0] & = -6\\
 \pay_{3}[0, 0, 1] & = -8\\
 \pay_{3}[1, 0, 1] & = -6\\
 \pay_{3}[0, 1, 1] & = 0\\
 \pay_{3}[1, 1, 1] & = 5\\
\end{aligned}
\right]
\end{align*}
As expected, \ac{RD} is recurrent (in particular, periodic) in the harmonic case $\potparam = 0$, and converges to a pure \ac{NE} in the potential case $\potparam = 1$. We leave the rationality properties of no-regret learning schemes in convex combinations of potential and harmonic games as an open direction for future research.
\begin{figure*}[t]
\centering
\foreach \i in {1,...,8}{
\includegraphics[width=.238\textwidth]{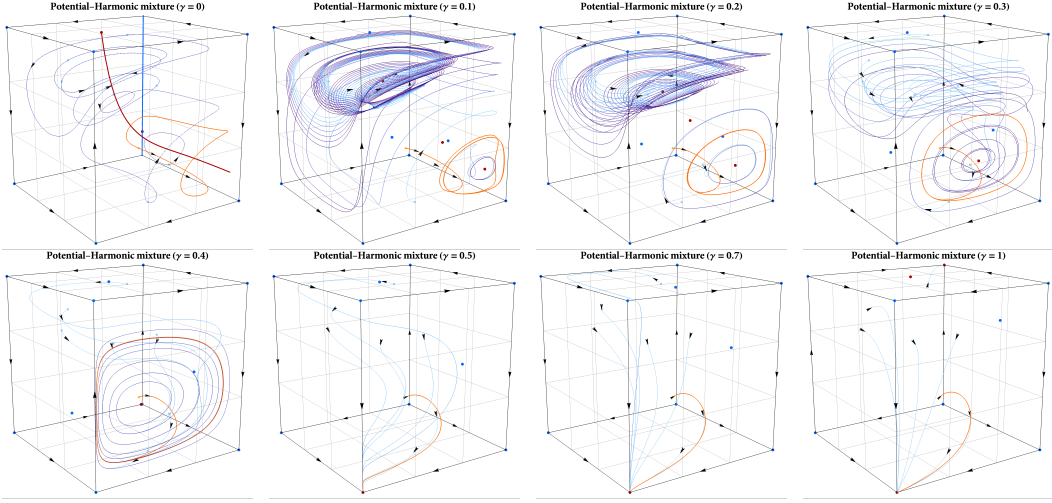}
}
\caption{Replicator trajectories in an ensemble of $2 \times 2 \times 2$ games with payoff $\pay \defeq \potparam  \pay_{p} + (1-\potparam) \pay_{h}$ given by the convex combination of a harmonic and a potential game. The value of the parameter $\potparam$ is shown in the legend of each plot, and $\pay_{p}, \pay_{h}$ are given in \cref{app:harmonic-potential-trajectories}. Each trajectory is color-coded so that deeper shades of blue-purple correspond to later times, with the arrows indicating the direction in which orbits are traversed. Light blue markers represent initial points for the orbits; dark blue markers are stationary points for the \acl{RD}; and dark red points are \aclp{NE}. For visual clarity, we have highlighted in orange one of the plotted orbits. As expected, \ac{RD} is recurrent (in particular, periodic) in the harmonic case $\potparam = 0$, and converges to a pure \ac{NE} in the potential case $\potparam = 1$.}
\label{fig:pot-harm-trajectories}
\end{figure*}

%% file: Appendix/App-GeometricTour.tex
A reader familiar with differential geometry won't have failed to realize that we have been trying to explain some fundamental geometrical concepts in an intuitive, yet not completely rigorous, way. This is an unavoidable price to pay if we wish to present results to an audience not necessarily familiar with the geometrical theory they rely upon.

For that reader, here is a quick tour about what is going on, that can be safely skipped by anyone not interested in the geometrical intricacies underlying the constructions we presented. A notation aside: as usual, for each $\play \in \players$, the index $\pure_{\play} \in \pures_{\play}$ runs from $0_{\play}$ to $\dimStrats_{\play}$ and the index $\eff_{\play} \in \effPures_{\play}$ runs from $1_{\play}$ to $\dimStrats_{\play}$. It is understood that whenever the index $\play$ appears in an equation such equation holds true for all $\play \in \players$, unless otherwise specified.

\begin{itemize}

\item Each open strategy space $\intstrats_{\play} = \intsimplex\of{\pures_{\play}}$ is a smooth submanifold of $\R^{\pures_{\play}}$ of dimension $\nEffs_{\play} = \nPures_{\play} - 1$;
\item $\chart\from\intstrats_{\play}\to\intcorcube_{\play}$ is a global chart and $\incl\from\intcorcube_{\play}\to\intstrats_{\play}$ the corresponding parametrization;
\item $\intstrats = \prod_{\play}\intstrats_{\play}$ has the standard smooth structure of product submanifold in $\R^{\pures} = \prod_{\play}\R^{\pures_{\play}}$;
\item the differential of the parametrization $\d\incl$ is an isomorphism between $\tspace\intcorcube_{\play} = \R^{\nEffs_{\play}}$ and $\tspace\intstrats_{\play} \subset \R^{\nEffs_{\play}+1}$ allowing to express the $\nEffs_{\play}$ basis vectors $\braces{\tilde{\pd}_{\effindex}}_{\effplay \in \effrangeplay}$ of $\tspace\intstrats_{\play}$ in the chart $\chart$ as a linear combination of the $\nEffs_{\play}+1$ basis vectors $\braces{\pd_{\findex}}_{\pureplay \in \rangeplay}$ of $\R^{\nEffs_{\play}+1}$ in the standard Cartesian frame as
\[
\tilde{\pd}_{\effindex} = \pd_{\effindex} - \pd_{\zindex} \quad \text{for all } \effplay \in \effrangeplay \eqdot
\]
\item Analogously, the pull-back of the $\nEffs_{\play} + 1$ basis $1$-forms $\setof{\d\strat^{\findex}}_{\pureplay \in \rangeplay}$ in $\orthantplay$ along $\incl$ in terms of the $\nEffs_{\play}$ basis $1$-forms $\setof{\d\effstrat^{\effindex}}_{\effplay \in \effrangeplay}$ on $\intstrats_{\play}$ gives
\[
\d\strat^{\zindex} =
- \sum_{\effplay} \d\effstrat^{\rindex}
\quad\text{and}\quad
\d\strat^{\rindex} = \d\effstrat^{\rindex}
\quad \text{for all } \effplay \in \effrangeplay \eqdot
\]
\item Payoff functions $\pay_{\play}\from\strats\to\R$ are actually multilinear functions on the whole $\pay_{\play}\from\R^{\pures}\to\R$;
\item each \quotes{effective payoff function} $\effpay_{\play}$ is the pull-back along the inclusion of $\pay_{\play}$, \ie a smooth function on $\intstrats_{\play}$;
\item each payoff field $\payfield_{\play} = \d_{  \play}\pay_{\play}$ is a $1$-form on $\R^{\pures_{\play}}$ and $\payfield = \sum_{\play} \payfield_{\play}$ is a $1$-form on $\R^{\pures}$;
\item the \quotes{effective payoff field} $\effpayfield_{\play}$ is the pull-back along the inclusion of $\payfield_{\play}$, \ie a $1$-form on $\intstrats_{\play}$;
\item the ambient \shah metric on $\R^{\pures}$,
\[
\gmat(\strat)
= \sum_{i \in \players}\sum_{\pureplay\purealtplay} \frac{\delta_{\pureplay \purealtplay}  }{\strat_{\findex}} \, \d\strat^{\findex} \otimes \d\strat^{\play\purealt_{\play}} \eqcomma
\]
is pulled back to the effective \shah metric on $\intstrats$,
\[\tilde\gmat(\effstrat)
= \sum_{i \in \players}\sum_{\effplay\effalt_\play} \parens*{ \frac{\delta_{\effplay \effalt_{\play}}  }{\effstrat_{\effindex}} + \frac{1}{\strat_{\zindex}}}  \, \d\effstrat^{\effindex} \otimes \d\effstrat^{\play\effalt_{\play}} \eqdot
\]
\end{itemize}

It is a standard fact in differential geometry~\citep{leeIntroductionSmoothManifolds2012} that if $\iota: (S,g) \hookrightarrow (M,G)$ is a Riemannian submanifold\footnote{That is, $\iota\from S \to M$ is a smooth injective immersion and $(M,G)$ is a Riemannian manifold.} with metric $g$ induced by the ambient metric $G$, then the sharp isomorphism $\sharp$, the pull-back along the inclusion $\iota^{\ast}$, and the orthogonal projection $P_{G}$ commute, in the sense that $P_{G}\parens{\alpha^{\sharp_{G}}} = \parens{\iota^{\ast}\alpha}^{\sharp_{g}}$ is a vector field on $S$ for all $1$-forms $\alpha$ on $M$. With this in place, each replicator field $\repfield_{\play}(\strat)$ is equivalently
\begin{itemize}
\item the orthogonal projection on $\intstrats_{\play}$ with respect to the ambient \shah metric of the ambient sharp of the $1$-form $\payfield_{\play}$ (this point of view is adopted in \citet{mertikopoulosRiemannianGameDynamics2018});
\item the sharp with respect to the pull-back metric of the pull-back $1$-form $\effpayfield_{\play}$;
\item the individual gradient with respect to the pull-back \shah metric of the pull-back payoff function $\effpay_{\play}$.
\end{itemize}
In our notation $\repfield_{\play}$ is a vector field \textit{parallel} to $\intstrats$, \ie a vector field legitimately defined on the whole ambient space $\R^{\pures}$ that evaluated at $\strat \in \intstrats$ gives a vector in the tangent space $\tspace_{\strat}\intstrats$ as a linear subspace of $\tspace_{\strat}\R^{\pures}$, so $\repfield_{\play}$ has $\nEffs_{\play}+1$ components. On the other hand $\effrepfield_{\play}$ is a vector field \textit{on} $\intstrats$ with $\nEffs_{\play}$ components, well-defined without the need to see $\intstrats$ as an immersed submanifold.

The reason to go through the reduction procedure of the replicator field and the \shah metric is that what matters from a dynamical standpoint is the divergence of $\effrepfield$ as a vector field on $\intstrats$, but in general given a vector field $X$ parallel to a Riemannian submanifold $\iota: (S,g) \hookrightarrow (M,G)$ there is \textit{not} a simple relation between its the divergence with respect to the ambient metric and its divergence with respect to the pull-back metric (in particular, $\diver_{g}X \neq \parens{\diver_{G}X}\circ\iota$); for details see \citet[Ch. 6]{carmoRiemannianGeometry1992}. This ultimately makes it necessary to compute the effective versions (resp. pulled back and projected) of the metric and the vector field at hand.


%% file: Main.bbl
\begin{thebibliography}{86}
\providecommand{\natexlab}[1]{#1}
\providecommand{\url}[1]{\texttt{#1}}
\expandafter\ifx\csname urlstyle\endcsname\relax
  \providecommand{\doi}[1]{doi: #1}\else
  \providecommand{\doi}{doi: \begingroup \urlstyle{rm}\Url}\fi

\bibitem[Abdou et~al.(2022)Abdou, Pnevmatikos, Scarsini, and
  Venel]{abdouDecompositionGamesStrategic2022}
Abdou, J., Pnevmatikos, N., Scarsini, M., and Venel, X.
\newblock Decomposition of {{Games}}: {{Some Strategic Considerations}}.
\newblock \emph{Mathematics of Operations Research}, 47\penalty0 (1):\penalty0
  176--208, February 2022.
\newblock ISSN 0364-765X.
\newblock \doi{10.1287/moor.2021.1123}.

\bibitem[Akin(1979)]{akinGeometryPopulationGenetics1979}
Akin, E.
\newblock \emph{The {{Geometry}} of {{Population Genetics}}}.
\newblock Lecture {{Notes}} in {{Biomathematics}} 31. Springer-Verlag Berlin
  Heidelberg, 1 edition, 1979.
\newblock ISBN 978-3-540-09711-2 978-3-642-93128-4.

\bibitem[Amann \& Hofbauer(1985)Amann and
  Hofbauer]{amannPermanenceLotkaReplicator1985}
Amann, E. and Hofbauer, J.
\newblock Permanence in lotka-volterra and replicator equations.
\newblock In Ebeling, W. and Peschel, M. (eds.), \emph{Lotka-Volterra Approach
  to Cooperation and Competition in Dynamic Systems}. Akademie-Verlag, Berlin,
  1985.

\bibitem[Amari \& Nagaoka(2000)Amari and
  Nagaoka]{amariMethodsInformationGeometry2000}
Amari, S.-i. and Nagaoka, H.
\newblock \emph{Methods of {{Information Geometry}}}.
\newblock American Mathematical Soc., 2000.
\newblock ISBN 978-0-8218-4302-4.

\bibitem[Arnold(1989)]{Arn89}
Arnold, V.~I.
\newblock \emph{Mathematical Methods of Classical Mechanics}.
\newblock Number~60 in Graduate Texts in Mathematics. Springer, New York, NY, 2
  edition, 1989.

\bibitem[Arora et~al.(2012)Arora, Hazan, and Kale]{AHK12}
Arora, S., Hazan, E., and Kale, S.
\newblock The multiplicative weights update method: A meta-algorithm and
  applications.
\newblock \emph{Theory of Computing}, 8\penalty0 (1):\penalty0 121--164, 2012.

\bibitem[Auer et~al.(1995)Auer, Cesa-Bianchi, Freund, and Schapire]{ACBFS95}
Auer, P., Cesa-Bianchi, N., Freund, Y., and Schapire, R.~E.
\newblock Gambling in a rigged casino: The adversarial multi-armed bandit
  problem.
\newblock In \emph{Proceedings of the 36th Annual Symposium on Foundations of
  Computer Science}, 1995.

\bibitem[Auer et~al.(2002)Auer, Cesa-Bianchi, Freund, and Schapire]{ACBFS02}
Auer, P., Cesa-Bianchi, N., Freund, Y., and Schapire, R.~E.
\newblock The nonstochastic multiarmed bandit problem.
\newblock \emph{SIAM Journal on Computing}, 32\penalty0 (1):\penalty0 48--77,
  2002.

\bibitem[Basar \& Ho(1974)Basar and Ho]{basarInformationalPropertiesNash1974}
Basar, T. and Ho, Y.-C.
\newblock Informational properties of the {{Nash}} solutions of two stochastic
  nonzero-sum games.
\newblock \emph{Journal of Economic Theory}, 7\penalty0 (4):\penalty0 370--387,
  April 1974.
\newblock ISSN 0022-0531.
\newblock \doi{10.1016/0022-0531(74)90110-0}.

\bibitem[Bekka et~al.(2000)Bekka, Bekka, and
  Mayer]{bekkaErgodicTheoryTopological2000}
Bekka, M.~B., Bekka, M. E.~B., and Mayer, M.
\newblock \emph{Ergodic Theory and Topological Dynamics of Group Actions on
  Homogeneous Spaces}, volume 269.
\newblock Cambridge University Press, 2000.

\bibitem[Bhatia et~al.(2013)Bhatia, Norgard, Pascucci, and
  Bremer]{bhatiaHelmholtzHodgeDecompositionSurvey2013}
Bhatia, H., Norgard, G., Pascucci, V., and Bremer, P.-T.
\newblock The {{Helmholtz-Hodge Decomposition}}---{{A Survey}}.
\newblock \emph{IEEE Transactions on Visualization and Computer Graphics},
  19\penalty0 (8):\penalty0 1386--1404, August 2013.

\bibitem[Boone \& Piliouras(2019)Boone and
  Piliouras]{booneDarwinPoincareNeumann2019}
Boone, V. and Piliouras, G.
\newblock From {{Darwin}} to {{Poincar{\'e}}} and von {{Neumann}}:
  {{Recurrence}} and {{Cycles}} in {{Evolutionary}} and {{Algorithmic Game
  Theory}}, October 2019.

\bibitem[Bravo \& Mertikopoulos(2017)Bravo and Mertikopoulos]{BM17}
Bravo, M. and Mertikopoulos, P.
\newblock On the robustness of learning in games with stochastically perturbed
  payoff observations.
\newblock \emph{Games and Economic Behavior}, 103\penalty0 (John Nash Memorial
  issue):\penalty0 41--66, May 2017.

\bibitem[Candogan et~al.(2011{\natexlab{a}})Candogan, Menache, Ozdaglar, and
  Parrilo]{CMOP11}
Candogan, O., Menache, I., Ozdaglar, A., and Parrilo, P.~A.
\newblock Flows and decompositions of games: harmonic and potential games.
\newblock \emph{Mathematics of Operations Research}, 36\penalty0 (3):\penalty0
  474--503, 2011{\natexlab{a}}.

\bibitem[Candogan et~al.(2011{\natexlab{b}})Candogan, Ozdaglar, and
  Parrilo]{candoganDynamicsNearPotentialGames2011}
Candogan, O., Ozdaglar, A., and Parrilo, P.~A.
\newblock Dynamics in {{Near-Potential Games}}.
\newblock 2011{\natexlab{b}}.
\newblock \doi{10.48550/ARXIV.1107.4386}.

\bibitem[Candogan et~al.(2011{\natexlab{c}})Candogan, Ozdaglar, and
  Parrilo]{candoganLearningNearpotentialGames2011}
Candogan, O., Ozdaglar, A., and Parrilo, P.~A.
\newblock Learning in near-potential games.
\newblock In \emph{2011 50th {{IEEE Conference}} on {{Decision}} and
  {{Control}} and {{European Control Conference}}}, pp.\  2428--2433, December
  2011{\natexlab{c}}.
\newblock \doi{10.1109/CDC.2011.6160867}.

\bibitem[Candogan et~al.(2013)Candogan, Ozdaglar, and
  Parrilo]{candoganNearPotentialGamesGeometry2013}
Candogan, O., Ozdaglar, A., and Parrilo, P.~A.
\newblock Near-{{Potential Games}}: {{Geometry}} and {{Dynamics}}.
\newblock \emph{ACM Transactions on Economics and Computation}, 1\penalty0
  (2):\penalty0 11:1--11:32, May 2013.
\newblock ISSN 2167-8375.
\newblock \doi{10.1145/2465769.2465776}.

\bibitem[Candogan(2013)]{candoganDynamicStrategicInteractions2013}
Candogan, U.~O.
\newblock \emph{Dynamic Strategic Interactions : Analysis and Mechanism
  Design}.
\newblock Thesis, Massachusetts Institute of Technology, 2013.

\bibitem[Cheng et~al.(2016)Cheng, Liu, Zhang, and
  Qi]{chengDecomposedSubspacesFinite2016}
Cheng, D., Liu, T., Zhang, K., and Qi, H.
\newblock On {{Decomposed Subspaces}} of {{Finite Games}}.
\newblock \emph{IEEE Transactions on Automatic Control}, 61\penalty0
  (11):\penalty0 3651--3656, November 2016.
\newblock ISSN 1558-2523.
\newblock \doi{10.1109/TAC.2016.2525936}.

\bibitem[Cheung \& Tao(2020)Cheung and Tao]{cheungChaosLearningZerosum2020}
Cheung, Y.~K. and Tao, Y.
\newblock Chaos of {{Learning Beyond Zero-sum}} and {{Coordination}} via {{Game
  Decompositions}}.
\newblock In \emph{International {{Conference}} on {{Learning
  Representations}}}, October 2020.

\bibitem[Conley(1978)]{Con78}
Conley, C.~C.
\newblock \emph{Isolated Invariant Set and the {Morse} Index}.
\newblock American Mathematical Society, Providence, RI, 1978.

\bibitem[Coucheney et~al.(2015)Coucheney, Gaujal, and Mertikopoulos]{CGM15}
Coucheney, P., Gaujal, B., and Mertikopoulos, P.
\newblock Penalty-regulated dynamics and robust learning procedures in games.
\newblock \emph{Mathematics of Operations Research}, 40\penalty0 (3):\penalty0
  611--633, August 2015.

\bibitem[Cressman et~al.(1998)Cressman, Morrison, and
  Wen]{cressmanEvolutionaryDynamicsCrime1998}
Cressman, R., Morrison, W.~G., and Wen, J.-F.
\newblock On the {{Evolutionary Dynamics}} of {{Crime}}.
\newblock \emph{The Canadian Journal of Economics / Revue canadienne
  d'Economique}, 31\penalty0 (5):\penalty0 1101--1117, 1998.
\newblock ISSN 0008-4085.
\newblock \doi{10.2307/136461}.

\bibitem[{de Rham}(1984)]{derhamDifferentiableManifoldsForms1984}
{de Rham}, G.
\newblock \emph{Differentiable {{Manifolds}}: {{Forms}}, {{Currents}},
  {{Harmonic Forms}}}.
\newblock Springer Science \& Business Media, 1984.

\bibitem[do~Carmo(1992)]{carmoRiemannianGeometry1992}
do~Carmo, M.~P.
\newblock \emph{Riemannian {{Geometry}}}.
\newblock Birkh{\"a}user Boston, 1 edition, 1992.

\bibitem[Flokas et~al.(2020)Flokas, Vlatakis-Gkaragkounis, Lianeas,
  Mertikopoulos, and Piliouras]{FVGL+20}
Flokas, L., Vlatakis-Gkaragkounis, E.~V., Lianeas, T., Mertikopoulos, P., and
  Piliouras, G.
\newblock No-regret learning and mixed {Nash} equilibria: {They} do not mix.
\newblock In \emph{NeurIPS '20: Proceedings of the 34th International
  Conference on Neural Information Processing Systems}, 2020.

\bibitem[Friedman(1991)]{Fri91}
Friedman, D.
\newblock Evolutionary games in economics.
\newblock \emph{Econometrica}, 59\penalty0 (3):\penalty0 637--666, 1991.

\bibitem[Fudenberg \& Tirole(1991)Fudenberg and Tirole]{FT91}
Fudenberg, D. and Tirole, J.
\newblock \emph{Game Theory}.
\newblock The MIT Press, 1991.

\bibitem[Hannan(1957)]{Han57}
Hannan, J.
\newblock Approximation to {Bayes} risk in repeated play.
\newblock In Dresher, M., Tucker, A.~W., and Wolfe, P. (eds.),
  \emph{Contributions to the Theory of Games, Volume {III}}, volume~39 of
  \emph{Annals of Mathematics Studies}, pp.\  97--139. Princeton University
  Press, Princeton, NJ, 1957.

\bibitem[Harper(2009)]{harperInformationGeometryEvolutionary2009}
Harper, M.
\newblock Information geometry and evolutionary game theory.
\newblock \emph{arXiv preprint arXiv:0911.1383}, 2009.

\bibitem[Hart \& Mas-Colell(2003)Hart and Mas-Colell]{HMC03}
Hart, S. and Mas-Colell, A.
\newblock Uncoupled dynamics do not lead to {Nash} equilibrium.
\newblock \emph{American Economic Review}, 93\penalty0 (5):\penalty0
  1830--1836, 2003.

\bibitem[Hart \& Mas-Colell(2006)Hart and Mas-Colell]{HMC06}
Hart, S. and Mas-Colell, A.
\newblock Stochastic uncoupled dynamics and {Nash} equilibrium.
\newblock \emph{Games and Economic Behavior}, 57:\penalty0 286--303, 2006.

\bibitem[H{\'e}liou et~al.(2017)H{\'e}liou, Cohen, and Mertikopoulos]{HCM17}
H{\'e}liou, A., Cohen, J., and Mertikopoulos, P.
\newblock Learning with bandit feedback in potential games.
\newblock In \emph{NIPS '17: Proceedings of the 31st International Conference
  on Neural Information Processing Systems}, 2017.

\bibitem[Hodge(1989)]{hodgeTheoryApplicationsHarmonic1989}
Hodge, W. V.~D.
\newblock \emph{The {{Theory}} and {{Applications}} of {{Harmonic Integrals}}}.
\newblock CUP Archive, May 1989.
\newblock ISBN 978-0-521-35881-1.

\bibitem[Hofbauer(1996)]{hofbauerEvolutionaryDynamicsBimatrix1996}
Hofbauer, J.
\newblock Evolutionary dynamics for bimatrix games: {{A Hamiltonian}} system?
\newblock \emph{Journal of Mathematical Biology}, 34\penalty0 (5):\penalty0
  675, May 1996.
\newblock ISSN 1432-1416.
\newblock \doi{10.1007/BF02409754}.

\bibitem[Hofbauer \& Schlag(2000)Hofbauer and
  Schlag]{hofbauerSophisticatedImitationCyclic2000}
Hofbauer, J. and Schlag, K.~H.
\newblock Sophisticated imitation in cyclic games.
\newblock \emph{Journal of Evolutionary Economics}, 10\penalty0 (5):\penalty0
  523--543, September 2000.
\newblock ISSN 1432-1386.
\newblock \doi{10.1007/s001910000049}.

\bibitem[Hofbauer \& Sigmund(1998)Hofbauer and Sigmund]{HS98}
Hofbauer, J. and Sigmund, K.
\newblock \emph{Evolutionary Games and Population Dynamics}.
\newblock Cambridge University Press, Cambridge, UK, 1998.

\bibitem[Hsieh et~al.(2020)Hsieh, Iutzeler, Malick, and Mertikopoulos]{HIMM20}
Hsieh, Y.-G., Iutzeler, F., Malick, J., and Mertikopoulos, P.
\newblock Explore aggressively, update conservatively: {Stochastic}
  extragradient methods with variable stepsize scaling.
\newblock In \emph{NeurIPS '20: Proceedings of the 34th International
  Conference on Neural Information Processing Systems}, 2020.

\bibitem[Hsieh et~al.(2021)Hsieh, Antonakopoulos, and Mertikopoulos]{HAM21}
Hsieh, Y.-G., Antonakopoulos, K., and Mertikopoulos, P.
\newblock Adaptive learning in continuous games: {Optimal} regret bounds and
  convergence to {Nash} equilibrium.
\newblock In \emph{COLT '21: Proceedings of the 34th Annual Conference on
  Learning Theory}, 2021.

\bibitem[Hsieh et~al.(2022)Hsieh, Antonakopoulos, Cevher, and
  Mertikopoulos]{HACM22}
Hsieh, Y.-G., Antonakopoulos, K., Cevher, V., and Mertikopoulos, P.
\newblock No-regret learning in games with noisy feedback: {Faster} rates and
  adaptivity via learning rate separation.
\newblock In \emph{NeurIPS '22: Proceedings of the 36th International
  Conference on Neural Information Processing Systems}, 2022.

\bibitem[Hubbard \& Burke~Hubbard(2015)Hubbard and
  Burke~Hubbard]{hubbardVectorCalculusLinear2015}
Hubbard, J.~H. and Burke~Hubbard, B.
\newblock \emph{Vector {{Calculus}}, {{Linear Algebra}}, and {{Differential
  Forms}}: {{A Unified Approach}} (5th Edition)}.
\newblock Matrix Editions, September 2015.

\bibitem[Hwang \& {Rey-Bellet}(2020)Hwang and
  {Rey-Bellet}]{hwangStrategicDecompositionsNormal2020}
Hwang, S.-H. and {Rey-Bellet}, L.
\newblock Strategic decompositions of normal form games: {{Zero-sum}} games and
  potential games.
\newblock \emph{Games and Economic Behavior}, 122:\penalty0 370--390, July
  2020.
\newblock ISSN 0899-8256.
\newblock \doi{10.1016/j.geb.2020.05.003}.

\bibitem[Jiang et~al.(2011)Jiang, Lim, Yao, and
  Ye]{jiangStatisticalRankingCombinatorial2011}
Jiang, X., Lim, L.-H., Yao, Y., and Ye, Y.
\newblock Statistical ranking and combinatorial {{Hodge}} theory.
\newblock \emph{Mathematical Programming}, 127\penalty0 (1):\penalty0 203--244,
  2011.

\bibitem[Jost(2017)]{jostRiemannianGeometryGeometric2017}
Jost, J.
\newblock \emph{Riemannian {{Geometry}} and {{Geometric Analysis}}}.
\newblock Universitext. Springer International Publishing, Cham, 2017.

\bibitem[Kalai \& Kalai(2010)Kalai and
  Kalai]{kalaiCooperationCompetitionStrategic2010}
Kalai, A.~T. and Kalai, E.
\newblock Cooperation and competition in strategic games with private
  information.
\newblock In \emph{Proceedings of the 11th {{ACM}} Conference on {{Electronic}}
  Commerce}, {{EC}} '10, pp.\  345--346, New York, NY, USA, June 2010.
  Association for Computing Machinery.
\newblock ISBN 978-1-60558-822-3.
\newblock \doi{10.1145/1807342.1807397}.

\bibitem[Khalil(2002)]{khalilNonlinearSystems2002}
Khalil, H.~K.
\newblock \emph{{Nonlinear Systems}}.
\newblock Pearson College Div, Upper Saddle River, N.J, subsequent edizione
  edition, January 2002.
\newblock ISBN 978-0-13-067389-3.

\bibitem[Kimura(1958)]{Kim58}
Kimura, M.
\newblock On the change of population fitness by natural selection.
\newblock \emph{Heredity}, 12:\penalty0 145--167, 1958.

\bibitem[Kwon \& Mertikopoulos(2017)Kwon and Mertikopoulos]{KM17}
Kwon, J. and Mertikopoulos, P.
\newblock A continuous-time approach to online optimization.
\newblock \emph{Journal of Dynamics and Games}, 4\penalty0 (2):\penalty0
  125--148, April 2017.

\bibitem[Lahkar \& Sandholm(2008)Lahkar and Sandholm]{LS08}
Lahkar, R. and Sandholm, W.~H.
\newblock The projection dynamic and the geometry of population games.
\newblock \emph{Games and Economic Behavior}, 64:\penalty0 565--590, 2008.

\bibitem[Laraki \& Mertikopoulos(2013)Laraki and Mertikopoulos]{LM13}
Laraki, R. and Mertikopoulos, P.
\newblock Higher order game dynamics.
\newblock \emph{Journal of Economic Theory}, 148\penalty0 (6):\penalty0
  2666--2695, November 2013.

\bibitem[Laraki \& Mertikopoulos(2015)Laraki and Mertikopoulos]{LM15}
Laraki, R. and Mertikopoulos, P.
\newblock Inertial game dynamics and applications to constrained optimization.
\newblock \emph{SIAM Journal on Control and Optimization}, 53\penalty0
  (5):\penalty0 3141--3170, October 2015.

\bibitem[Lattimore \& Szepesv{\'a}ri(2020)Lattimore and Szepesv{\'a}ri]{LS20}
Lattimore, T. and Szepesv{\'a}ri, C.
\newblock \emph{Bandit Algorithms}.
\newblock Cambridge University Press, Cambridge, UK, 2020.

\bibitem[Lee(2012)]{leeIntroductionSmoothManifolds2012}
Lee, J.~M.
\newblock \emph{Introduction to {{Smooth Manifolds}}}.
\newblock Graduate {{Texts}} in {{Mathematics}}. Springer-Verlag New York, 2
  edition, 2012.

\bibitem[Lee(2018)]{leeIntroductionRiemannianManifolds2018}
Lee, J.~M.
\newblock \emph{Introduction to {{Riemannian Manifolds}}}, volume 176 of
  \emph{Graduate {{Texts}} in {{Mathematics}}}.
\newblock Springer International Publishing, Cham, 2018.
\newblock ISBN 978-3-319-91754-2 978-3-319-91755-9.
\newblock \doi{10.1007/978-3-319-91755-9}.

\bibitem[Letcher et~al.(2019)Letcher, Balduzzi, Racani{\`e}re, Martens,
  Foerster, Tuyls, and Graepel]{letcherDifferentiableGameMechanics2019}
Letcher, A., Balduzzi, D., Racani{\`e}re, S., Martens, J., Foerster, J., Tuyls,
  K., and Graepel, T.
\newblock Differentiable {{Game Mechanics}}.
\newblock \emph{Journal of Machine Learning Research}, 20\penalty0
  (84):\penalty0 1--40, 2019.
\newblock ISSN 1533-7928.

\bibitem[Li et~al.(2016)Li, Liu, He, Cheng, Qi, and
  Hong]{liFiniteHarmonicGames2016}
Li, C., Liu, T., He, F., Cheng, D., Qi, H., and Hong, Y.
\newblock On finite harmonic games.
\newblock In \emph{2016 {{IEEE}} 55th {{Conference}} on {{Decision}} and
  {{Control}} ({{CDC}})}, pp.\  7024--7029, December 2016.
\newblock \doi{10.1109/CDC.2016.7799351}.

\bibitem[Li et~al.(2019)Li, Cheng, and He]{liNoteOrthogonalDecomposition2019}
Li, C., Cheng, D., and He, F.
\newblock A {{Note On Orthogonal Decomposition}} of {{Finite Games}}, May 2019.

\bibitem[Littlestone \& Warmuth(1994)Littlestone and Warmuth]{LW94}
Littlestone, N. and Warmuth, M.~K.
\newblock The weighted majority algorithm.
\newblock \emph{Information and Computation}, 108\penalty0 (2):\penalty0
  212--261, 1994.

\bibitem[Mertikopoulos \& Moustakas(2010)Mertikopoulos and Moustakas]{MM10}
Mertikopoulos, P. and Moustakas, A.~L.
\newblock The emergence of rational behavior in the presence of stochastic
  perturbations.
\newblock \emph{The Annals of Applied Probability}, 20\penalty0 (4):\penalty0
  1359--1388, July 2010.

\bibitem[Mertikopoulos \& Sandholm(2016)Mertikopoulos and
  Sandholm]{mertikopoulosLearningGamesReinforcement2016}
Mertikopoulos, P. and Sandholm, W.~H.
\newblock Learning in games via reinforcement and regularization.
\newblock \emph{Mathematics of Operations Research}, 41\penalty0 (4):\penalty0
  1297--1324, 2016.

\bibitem[Mertikopoulos \& Sandholm(2018)Mertikopoulos and
  Sandholm]{mertikopoulosRiemannianGameDynamics2018}
Mertikopoulos, P. and Sandholm, W.~H.
\newblock Riemannian game dynamics.
\newblock \emph{Journal of Economic Theory}, 177:\penalty0 315--364, 2018.

\bibitem[Mertikopoulos \& Zhou(2019)Mertikopoulos and Zhou]{MZ19}
Mertikopoulos, P. and Zhou, Z.
\newblock Learning in games with continuous action sets and unknown payoff
  functions.
\newblock \emph{Mathematical Programming}, 173\penalty0 (1-2):\penalty0
  465--507, January 2019.

\bibitem[Mertikopoulos et~al.(2018)Mertikopoulos, Papadimitriou, and
  Piliouras]{MPP18}
Mertikopoulos, P., Papadimitriou, C.~H., and Piliouras, G.
\newblock Cycles in adversarial regularized learning.
\newblock In \emph{SODA '18: Proceedings of the 29th annual ACM-SIAM Symposium
  on Discrete Algorithms}, 2018.

\bibitem[Mertikopoulos et~al.(2019)Mertikopoulos, Lecouat, Zenati, Foo,
  Chandrasekhar, and Piliouras]{MLZF+19}
Mertikopoulos, P., Lecouat, B., Zenati, H., Foo, C.-S., Chandrasekhar, V., and
  Piliouras, G.
\newblock Optimistic mirror descent in saddle-point problems: {Going} the extra
  (gradient) mile.
\newblock In \emph{ICLR '19: Proceedings of the 2019 International Conference
  on Learning Representations}, 2019.

\bibitem[Mladenovic et~al.(2021)Mladenovic, Sakos, Gidel, and
  Piliouras]{mladenovicGeneralizedNaturalGradient2021}
Mladenovic, A., Sakos, I., Gidel, G., and Piliouras, G.
\newblock Generalized natural gradient flows in hidden convex-concave games and
  gans.
\newblock 2021.

\bibitem[Monderer \& Shapley(1996)Monderer and Shapley]{MS96}
Monderer, D. and Shapley, L.~S.
\newblock Potential games.
\newblock \emph{Games and Economic Behavior}, 14\penalty0 (1):\penalty0 124 --
  143, 1996.

\bibitem[Munkres(1999)]{munkresTopology1999}
Munkres, J.~R.
\newblock \emph{Topology}.
\newblock Pearson Higher Ed, 1999.

\bibitem[Nagarajan et~al.(2020)Nagarajan, Balduzzi, and
  Piliouras]{nagarajanChaosOrderSymmetry2020}
Nagarajan, S.~G., Balduzzi, D., and Piliouras, G.
\newblock From {{Chaos}} to {{Order}}: {{Symmetry}} and {{Conservation Laws}}
  in {{Game Dynamics}}.
\newblock In \emph{Proceedings of the 37th {{International Conference}} on
  {{Machine Learning}}}, pp.\  7186--7196. PMLR, November 2020.

\bibitem[Piliouras \& Shamma(2014)Piliouras and Shamma]{PS14}
Piliouras, G. and Shamma, J.~S.
\newblock Optimization despite chaos: {Convex} relaxations to complex limit
  sets via {Poincar{\'e}} recurrence.
\newblock In \emph{SODA '14: Proceedings of the 25th annual ACM-SIAM Symposium
  on Discrete Algorithms}, 2014.

\bibitem[Poincar{\'e}(1889)]{poincareProblemeTroisCorps1889}
Poincar{\'e}, H.
\newblock \emph{Sur Le Probl{\`e}me Des Trois Corps et Les {\'E}quations de La
  Dynamique}.
\newblock M{\'e}moire couronn{\'e} du prix de SM le Roi Oscar II, 1889.

\bibitem[Rakhlin \& Sridharan(2013)Rakhlin and Sridharan]{RS13-NIPS}
Rakhlin, A. and Sridharan, K.
\newblock Optimization, learning, and games with predictable sequences.
\newblock In \emph{NIPS '13: Proceedings of the 27th International Conference
  on Neural Information Processing Systems}, 2013.

\bibitem[Ritzberger \& Vogelsberger(1990)Ritzberger and
  Vogelsberger]{ritzbergerNashField1990}
Ritzberger, K. and Vogelsberger, K.
\newblock The {{Nash}} field.
\newblock 1990.

\bibitem[Rosen(1965)]{Ros65}
Rosen, J.~B.
\newblock Existence and uniqueness of equilibrium points for concave
  ${N}$-person games.
\newblock \emph{Econometrica}, 33\penalty0 (3):\penalty0 520--534, 1965.

\bibitem[Sandholm(2010)]{San10}
Sandholm, W.~H.
\newblock \emph{Population Games and Evolutionary Dynamics}.
\newblock MIT Press, Cambridge, MA, 2010.

\bibitem[Shahshahani(1979)]{Sha79}
Shahshahani, S.~M.
\newblock \emph{A New Mathematical Framework for the Study of Linkage and
  Selection}.
\newblock Number 211 in Memoirs of the American Mathematical Society. American
  Mathematical Society, Providence, RI, 1979.

\bibitem[Shalev-Shwartz(2011)]{SS11}
Shalev-Shwartz, S.
\newblock Online learning and online convex optimization.
\newblock \emph{Foundations and Trends in Machine Learning}, 4\penalty0
  (2):\penalty0 107--194, 2011.

\bibitem[Shalev-Shwartz \& Singer(2006)Shalev-Shwartz and Singer]{SSS06}
Shalev-Shwartz, S. and Singer, Y.
\newblock Convex repeated games and {Fenchel} duality.
\newblock In \emph{NIPS' 06: Proceedings of the 19th Annual Conference on
  Neural Information Processing Systems}, pp.\  1265--1272. MIT Press, 2006.

\bibitem[Smith \& Hofbauer(1987)Smith and
  Hofbauer]{smithBattleSexesGenetic1987}
Smith, J.~M. and Hofbauer, J.
\newblock The ``battle of the sexes'': {{A}} genetic model with limit cycle
  behavior.
\newblock \emph{Theoretical Population Biology}, 32\penalty0 (1):\penalty0
  1--14, August 1987.
\newblock ISSN 0040-5809.
\newblock \doi{10.1016/0040-5809(87)90035-9}.

\bibitem[Sorin(2009)]{Sor09}
Sorin, S.
\newblock Exponential weight algorithm in continuous time.
\newblock \emph{Mathematical Programming}, 116\penalty0 (1):\penalty0 513--528,
  2009.

\bibitem[Spivak(1965)]{spivakCalculusManifoldsModern1965}
Spivak, M.
\newblock \emph{Calculus {{On Manifolds}}: {{A Modern Approach To Classical
  Theorems Of Advanced Calculus}}}.
\newblock CRC Press, Boca Raton, 1965.
\newblock ISBN 978-0-429-50190-6.
\newblock \doi{10.1201/9780429501906}.

\bibitem[Taylor \& Jonker(1978)Taylor and Jonker]{TJ78}
Taylor, P.~D. and Jonker, L.~B.
\newblock Evolutionary stable strategies and game dynamics.
\newblock \emph{Mathematical Biosciences}, 40\penalty0 (1-2):\penalty0
  145--156, 1978.

\bibitem[Viossat \& Zapechelnyuk(2013)Viossat and Zapechelnyuk]{VZ13}
Viossat, Y. and Zapechelnyuk, A.
\newblock No-regret dynamics and fictitious play.
\newblock \emph{Journal of Economic Theory}, 148\penalty0 (2):\penalty0
  825--842, March 2013.

\bibitem[Vovk(1990)]{Vov90}
Vovk, V.~G.
\newblock Aggregating strategies.
\newblock In \emph{COLT '90: Proceedings of the 3rd Workshop on Computational
  Learning Theory}, pp.\  371--383, 1990.

\bibitem[Wang et~al.(2017)Wang, Liu, and Cheng]{wangWeightedPotentialGame2017}
Wang, Y., Liu, T., and Cheng, D.
\newblock From weighted potential game to weighted harmonic game.
\newblock \emph{IET Control Theory \& Applications}, 11\penalty0 (13):\penalty0
  2161--2169, 2017.
\newblock ISSN 1751-8652.
\newblock \doi{10.1049/iet-cta.2016.1454}.

\bibitem[Wei et~al.(2021)Wei, Lee, Zhang, and Luo]{WLZL21}
Wei, C.-Y., Lee, C.-W., Zhang, M., and Luo, H.
\newblock Linear last-iterate convergence in constrained saddle-point
  optimization.
\newblock In \emph{ICLR '21: Proceedings of the 2021 International Conference
  on Learning Representations}, 2021.

\bibitem[Weibull(1995)]{Wei95}
Weibull, J.~W.
\newblock \emph{Evolutionary Game Theory}.
\newblock MIT Press, Cambridge, MA, 1995.

\end{thebibliography}
